\documentclass[sigconf]{acmart}

\copyrightyear{2021} 
\acmYear{2021} 
\setcopyright{acmlicensed}\acmConference[CCS '21]{Proceedings of the 2021 ACM SIGSAC Conference on Computer and Communications Security}{November 15--19, 2021}{Virtual Event, Republic of Korea}
\acmBooktitle{Proceedings of the 2021 ACM SIGSAC Conference on Computer and Communications Security (CCS '21), November 15--19, 2021, Virtual Event, Republic of Korea}
\acmPrice{15.00}
\acmDOI{10.1145/3460120.3484586}
\acmISBN{978-1-4503-8454-4/21/11}

\hypersetup{colorlinks,linkcolor={blue},citecolor={blue},urlcolor={red}}   
\usepackage{amsmath,amsfonts,mathtools,amscd}
\usepackage{url}
\usepackage{ifthen}
\usepackage{bm}
\usepackage{multirow}
\usepackage{xcolor,colortbl} 
\usepackage{comment}
\usepackage{paralist,verbatim}
\usepackage{cases}
\usepackage{xargs}
\usepackage{alphalph}
\usepackage{pgfplots}
\usepackage{caption}
\usepackage{subcaption}   
\usepackage{xspace}    
\usepackage{braket} 
\usepackage{cleveref}
\usepackage{color}
\usepackage[noend]{algpseudocode}
\usepackage{algorithm}
\usepackage{numprint}  
\usepackage{lipsum} 
\usepackage[flushleft]{threeparttable}
\usepackage{sidecap} 
\sidecaptionvpos{figure}{c} 

\usepackage{etoolbox}
\patchcmd{\thebibliography}{\chapter*}{\section*}{}{}
	\crefname{theorem}{Theorem}{Theorems}
	\crefname{section}{Section}{Sections}
	\crefname{algorithm}{Algorithm}{Algorithms}	
	\crefname{assumption}{Assumption}{Assumptions}
	\crefname{construction}{Construction}{Constructions}
	\crefname{corollary}{Corollary}{Corollaries}
	\crefname{conjecture}{Conjecture}{Conjectures}
	\crefname{definition}{Definition}{Definitions}
	\crefname{example}{Example}{Examples}
	\crefname{experiment}{Experiment}{Experiments}
	\crefname{counterexample}{Counterexample}{Counterexamples}
	\crefname{lemma}{Lemma}{Lemmata}
	\crefname{observation}{Observation}{Observations}
	\crefname{proposition}{Proposition}{Propositions}
	\crefname{remark}{Remark}{Remarks}
	\crefname{claim}{Claim}{Claims}
	\crefname{fact}{Fact}{Facts}
	\crefname{note}{Note}{Notes}
	\crefname{table}{Table}{Tables}
	\crefname{figure}{Figure}{Figures}
         \crefname{appendix}{Appendix}{Appendices}
	\crefname{equation}{Equation}{Equations}
	
 \crefname{appendix}{App.}{Appendices}
 \crefname{section}{Sec.}{Secs.}
 \crefname{figure}{Fig.}{Figs.}
 \crefname{table}{Tab.}{Tabs.}
 \crefname{remark}{Rem.}{Rems.}
 \crefname{theorem}{Thm.}{Thms.}
 \crefname{definition}{Def.}{Defs.}
 \crefname{equation}{Eq.}{Eqs.}
 \crefname{lemma}{Lem.}{Lems.}

\allowdisplaybreaks 

\DeclareMathAlphabet{\mathpzc}{OT1}{pzc}{m}{it}

\newcommand{\bigset}[1]{\big\{ #1 \big\}}

\newcommand{\tuple}[1]{( #1 )}

\newcommand{\normtwo}[1]{\lVert #1 \rVert_{2}}

\newcommand{\avgcollentropy}{\tilde{\mathbf{H}}_2}
\newcommand{\COL}{\mathtt{COL}} 
\newcommand{\abs}[1]{\left| #1 \right|} 

\newcommand{\SD}{\mathbf{SD}}  

    \newcommand{\GG}{\mathbb{G}}
\newcommand{\NN}{\mathbb{N}}    \newcommand{\ZZ}{\mathbb{Z}}
\newcommand{\RR}{\mathbb{R}}

    \newcommand{\mB}{\mathbf{B}}

\newcommand{\mI}{\mathbf{I}}

    \newcommand{\vb}{\mathbf{b}}
\newcommand{\vc}{\mathbf{c}}    
\newcommand{\ve}{\mathbf{e}}

\newcommand{\vx}{\mathbf{x}}    \newcommand{\vy}{\mathbf{y}}
\newcommand{\vz}{\mathbf{z}}    \newcommand{\vw}{\mathbf{w}}

\newcommand{\vzero}{\mathbf{0}}

\newcommand{\A}{\mathcal{A}}    \newcommand{\B}{\mathcal{B}}

 \newcommand{\GGG}{\mathcal{G}}

\newcommand{\samp}{\leftarrow}
\newcommand{\lruns}{\leftarrow}

\newcommand{\bin}{ \{ 0,1 \} }

\newcommand{\hash}{\mathsf{H}}

\newcommand{\HashFamily}{\mathcal{H}}

\newcommand{\secpar}{\kappa}
\newcommand{\msg}{\mathsf{M}}

\newcommand{\game}{\mathsf{Game}}

\newcommand{\Event}{\mathsf{E}}

\newcommand{\negl}{\mathsf{negl}}

\newcommand{\ggen}{\mathsf{GGen}}

\newcommand{\Sig}{\mathsf{S}.}

\newcommand{\SigSetup}{\Sig\mathsf{Setup}}

\newcommand{\Sigpp}{\mathsf{pp}}

\newcommand{\Sigsk}{\mathsf{sk}}

\newcommand{\SigMesSpace}{\mathcal{M}}

\newcommand{\eucma}{\mathtt{EU\text{-}CMA}}

\newcommand{\SQset}{\mathcal{Q}}

\newcommand{\FuzKeySet}{\mathcal{F}}

\newcommand{\FuzDataSpace}{X}
\newcommand{\FuzDataDist}{\mathcal{X}}

\newcommand{\ErrDist}{\mathit{\Phi}}

\newcommand{\ErrFNMR}{\epsilon}

\newcommand{\FNMR}{\ensuremath{\mathtt{FNMR}}\xspace}
\newcommand{\tFNMR}{\ensuremath{\widetilde{\mathtt{FNMR}}}\xspace}
\newcommand{\sFMR}{\ensuremath{\mathtt{sFMR}}\xspace}
\newcommand{\tsFMR}{\ensuremath{\widetilde{\mathtt{sFMR}}}\xspace}
\newcommand{\FMR}{\ensuremath{\mathtt{FMR}}\xspace}
\newcommand{\tFMR}{\ensuremath{\widetilde{\mathtt{FMR}}}\xspace}
\newcommand{\CFMR}{\ensuremath{\mathtt{ConFMR}}\xspace}
\newcommand{\sCFMR}{\ensuremath{\mathtt{sConFMR}}\xspace}

\newcommand{\tCFMR}{\ensuremath{\widetilde{\mathtt{ConFMR}}}\xspace}

\newcommand{\AR}{\ensuremath{\mathsf{AR}}\xspace}
\newcommand{\sAR}{\ensuremath{\mathsf{sAR}}\xspace}

\newcommand{\bS}{\bar{S}}
\newcommand{\dist}[1]{\mathsf{dist}( #1 )}

\newcommand{\FKSparam}{(\FuzDataSpace, \FuzDataDist, \AR, \ErrDist, \ErrFNMR)}

\newcommand{\Lattice}{\mathcal{L}}
\newcommand{\CV}{\mathsf{CV}}
\newcommand{\VR}{\mathsf{VR}}
\newcommand{\triLattice}{\Lattice_{\mathtt{tri}}}
\newcommand{\triB}{\mB_{\mathtt{tri}}}

\newcommand{\ProtFSig}{\Pi_{\mathsf{FS}}}
\newcommand{\FSig}{\mathsf{FS}.}
\newcommand{\FS}{\mathsf{FS}}
\newcommand{\FSigSetup}{\FSig\mathsf{Setup}}
\newcommand{\FSigKeyGen}{\FSig\mathsf{KeyGen}}
\newcommand{\FSigSign}{\FSig\mathsf{Sign}}
\newcommand{\FSigVrfy}{\FSig\mathsf{Vrfy}}
\newcommand{\FSigpp}{\mathsf{pp}_{\FS}}
\newcommand{\FSigvk}{\mathsf{vk}_{\FS}}

\newcommand{\FSigsig}{\sigma_{\FS}}

\newcommand{\ta}{{\tilde{a}}} \newcommand{\tildeh}{{\tilde{h}}}
 
\newcommand{\shifta}{\Delta a}

\newcommand{\ProtLinS}{\Pi_{\mathsf{LinS}}}

\newcommand{\LinS}{\mathsf{LinS.}}
\newcommand{\LinSSetup}{\LinS\mathsf{Setup}}
\newcommand{\LinSpp}{\mathsf{pp}_{\sf LS}}
\newcommand{\Sketch}{\mathsf{Sketch}}
\newcommand{\DiffRec}{\mathsf{DiffRec}}
\newcommand{\KeyDesc}{\Lambda}

\newcommand{\KeySpace}{\mathcal{K}}

\newcommand{\sketch}{c}
\newcommand{\Simsketch}{\mathsf{M_{\sketch}}}

\newcommand{\tsketch}{\widetilde{\hspace{1pt}\sketch\hspace{1pt}}}

\newcommand{\linHashFamily}{\mathcal{UH}}   

\newcommand{\linhash}{\mathsf{UH}}

\newcommand{\DL}{\ensuremath{\mathsf{DL}}\xspace}
\newcommand{\DLsketch}{\ensuremath{\mathsf{DL}^{\sf sketch}}\xspace}

\newtheorem{remark}{Remark}

\begin{document}
\fancyhead{}

\title{Revisiting Fuzzy Signatures: Towards a More Risk-Free Cryptographic Authentication System based on Biometrics} 

\author{Shuichi Katsumata}
\affiliation{%
  \institution{AIST}
  \city{Tokyo}  
  \country{Japan}
}
\email{shuichi.katsumata@aist.go.jp}

\author{Takahiro Matsuda}
\affiliation{%
  \institution{AIST}
  \city{Tokyo}  
  \country{Japan}
}
\email{t-matsuda@aist.go.jp}

\author{Wataru Nakamura}
\affiliation{%
  \institution{Hitachi, Ltd}
  \city{Tokyo}    
  \country{Japan}  
}
\email{wataru.nakamura.va@hitachi.com}

\author{Kazuma Ohara}
\affiliation{%
  \institution{AIST}
  \city{Tokyo}  
  \country{Japan}
}
\email{ohara.kazuma@aist.go.jp}

\author{Kenta Takahashi}
\affiliation{%
  \institution{Hitachi, Ltd}
  \city{Tokyo}    
  \country{Japan}  
}
\email{kenta.takahashi.bw@hitachi.com}

\begin{abstract}
Biometric authentication is one of the promising alternatives to standard password-based authentication offering better usability and security. 
In this work, we revisit the biometric authentication based on \emph{fuzzy signatures} introduced by Takahashi et al. (ACNS'15, IJIS'19). These are special types of digital signatures where the secret signing key can be a ``fuzzy'' data such as user's biometrics. 
Compared to other cryptographically secure biometric authentications as those relying on fuzzy extractors, the fuzzy signature-based scheme provides a more attractive security guarantee. 
However, despite their potential values, fuzzy signatures have not attracted much attention owing to their theory-oriented presentations in all prior works. 
For instance, the discussion on the practical feasibility of the assumptions (such as the entropy of user biometrics), which the security of fuzzy signatures hinges on, is completely missing.

In this work, we revisit fuzzy signatures and show that we can indeed efficiently and securely implement them in practice.
At a high level, our contribution is threefold:
(i) we provide a much simpler, more efficient, and direct construction of fuzzy signature compared to prior works;
(ii) we establish novel statistical techniques to experimentally evaluate the conditions on biometrics that are required to securely instantiate fuzzy signatures; 
and (iii) we provide experimental results using a real-world finger-vein dataset to show that finger-veins from a single hand are sufficient to construct efficient and secure fuzzy signatures.
Our performance analysis shows that in a practical scenario with 112-bits of security, the size of the signature is 1256 bytes, and the running time for signing/verification is only a few milliseconds.
\end{abstract}



\maketitle


\section{Introduction}
\label{sec:introduction}

\noindent{\underline{\textit{Background.}}
A user authentication system is a central infrastructure in a digital society. One of the most widely used methods for authentication is those using passwords. 
However, today it is becoming increasingly more difficult to protect passwords and to securely manage password-based authentication from the emerging advanced forms of cyberattacks.
For example, the ENISA Threat Landscape 2020 report~\cite{ENISA20} states that 64\% of the publicly exposed personal data due to security breaches in 2019 contained passwords.

One of the most promising alternatives to password-based authentication that has been gradually gaining traction is \emph{biometric authentication}~\cite{Visa19}, where a user's identity is verified through its biometrics such as face, iris, fingerprint, and finger-vein.
A familiar example is those widely implemented on personal smartphones such as the Touch~ID on iPhone. 
These types of authentication relies on the users holding a device embedding some information on their biometrics. 
In contrast, recently, biometric authentication without relying on these personal devices --- \emph{(personal) device-free} biometric authentication --- is beginning to be deployed in commercial and governmental services. 
Here, anybody can authenticate using the \emph{same} publicly available device.
This includes for instance the facial recognition payment service \emph{Alipay} managed by Alibaba in China~\cite{NA19}, and the world's largest biometric ID system \emph{Aadhaar} used in India~\cite{Aadhaar}. 
Due to their convenience and digital inclusiveness~\cite{UN20b}, the demand for such device-freeness is expected to grow further in the future in other applications ranging from payment and ATM transactions to medical systems, immigration control, and for building a national digital identity infrastructure. The focus of this article is on such device-free biometric authentication.

In biometric authentication, the most salient problem is how to securely protect biometric information.
As Visa~\cite{Visa19} stated, the \lq\lq Top concerns of using biometric authentication for payments'' is \lq\lq The risk of a security leak of sensitive information, e.g., you can't change your fingerprint if it is compromised.''
To realize a device-free biometric authentication, biometric information is typically stored and maintained on a central server. 
However, this opens up the risk of exposing user biometric information due to a security breach on the server. 
Although standard practices such as encrypting the database and placing appropriate access control on the users can mitigate the risk, as history shows, these common procedures are not easy to enforce or to execute in real-life due to human errors or lack of a security background. 
For instance, a vulnerability in the Aadhaar system was recently exploited and anybody had unrestricted access to the biometric information of more than 1 billion Indian citizens~\cite{WP18}. 
Since leaking biometric information is has an irreversible damage compared to leaking passwords, minimizing the risk on the server is one of the central problems for biometric authentication.

Biometric template protection (BTP) is designed to protect such biometric information stored on a server and has been standardized in recent years (ISO/IEC 24745~\cite{ISO:11}, 30136~\cite{ISO:18}).
\emph{Fuzzy extractor} (FE) \cite{EC:DodReySmi04} \cite{ISO:11} is one of the most promising tools for constructing a biometric authentication system with BTP whose (cryptographic) security can be formally analyzed. 
Informally, an FE enables to extract a \emph{fixed} secret key from a \emph{fuzzy} biometric. Here, biometrics are inherently fuzzy objects since they can slightly change over time, and measuring them perfectly is impossible due to measurement errors. 
The extracted fixed secret key is then used as a secret key of an ordinary signature scheme to achieve a \lq\lq biometric-based'' signature scheme, which can, in turn, be used for a biometric authentication system with BTP. 
More accurately, a user also needs a user-specific \emph{helper data}%
\footnote{
This is also called associated data or helper string in the literature. 
} to reconstruct the fixed secret key from its biometric. 
Intuitively, a helper data encodes some information on user biometrics to help reconstruct the same fixed secret key from the fuzzy biometric.
Since the helper data does not directly reveal the secret key nor the user biometric, it is considered more secure to store the helper data rather than the user biometric on the server.

A typical flow of an FE-based device-free biometric authentication system is given in \cref{fig:system} (left). 
Each service provider (e.g., bank, supermarket, hospital) has a client device. A user can use any of these devices to authenticate itself to the server by scanning its biometrics. 
A user first accesses the client device and makes an ID claim to the server. The client device downloads the corresponding helper data from the server, and the user then uses its biometric to reconstruct the fixed secret key (denoted as $\mathsf{KeyExtract}$ in \cref{fig:system}) used by the underlying ordinary signature scheme. 
Since the server only needs to store the user's helper data, the FE-based system provides BTP and successfully decreases the level of confidential information stored on the server. 
However, due to the added interaction between the client device and server, this opens up another type of risk. 
Notice that once an attacker obtains a client device, it can freely make ID claims to the server to collect the helper data of any user. 
Therefore, considering the attacker only needs to steal/compromise one of the many client devices, the possibility of a database exposure is much higher compared to the \emph{naive} system without BTP; the system where the biometrics are all stored on the server and the only way to retrieve them is through breaching the server. 
Of course, the concrete amount of biometric information leaked from the helper data in an FE depends on the specific construction and the security parameter used therein. 
However, in any case, we cannot take the risk zero since the adversary can collect many helper data easily and target to break any one of them; this is similar to the issue raised by reverse brute-force attacks. 
Thus, although the FE-based system lowers the level of confidential information stored on the server, it does so by increasing the possibility of such confidential information being exposed. 
Since the security \underline{r}isk of a system ($R$) is given by the product of the \underline{p}ossibility of the data on the server being exposed ($P$) and the \underline{i}mpact of such data being exposed ($I$), this brings us to our central question: 

\emph{Can we lower the level of confidential information stored on the server (as in the FE-based system) while simultaneously lowering the possibility of such information being leaked (as in the naive system)?}

\begin{figure*}[htbp]
 \begin{minipage}{0.45\hsize}
 \begin{center}
	\includegraphics[ width=0.8\textwidth ]{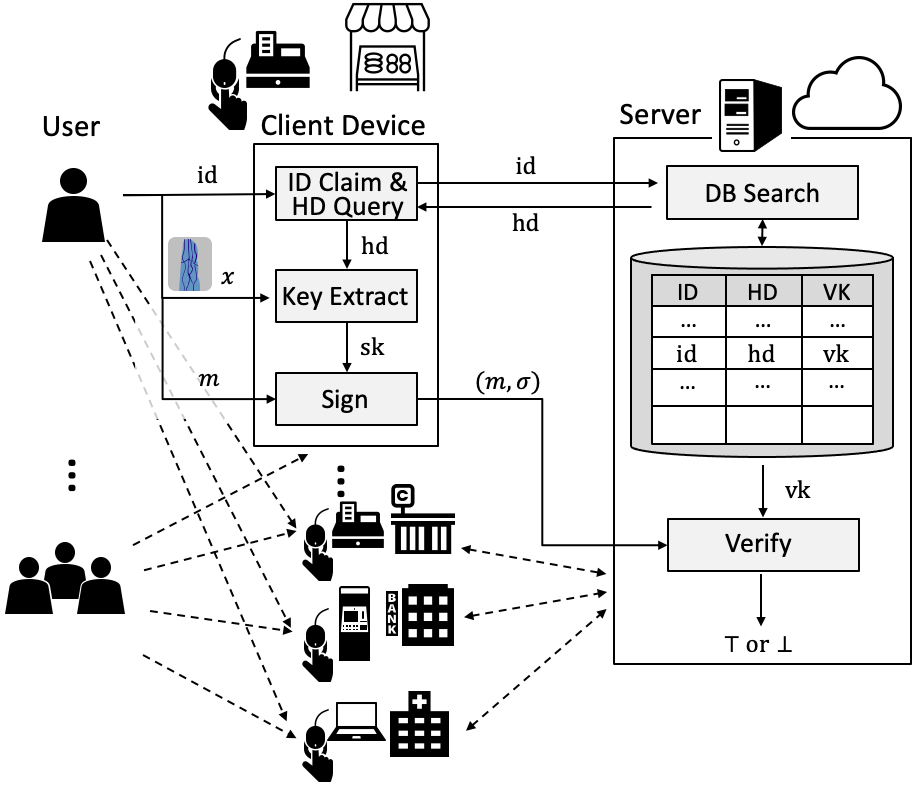}
 \end{center}
 \end{minipage}
 \begin{minipage}{0.45\hsize}
 \begin{center}
	\includegraphics[ width=0.80\textwidth ]{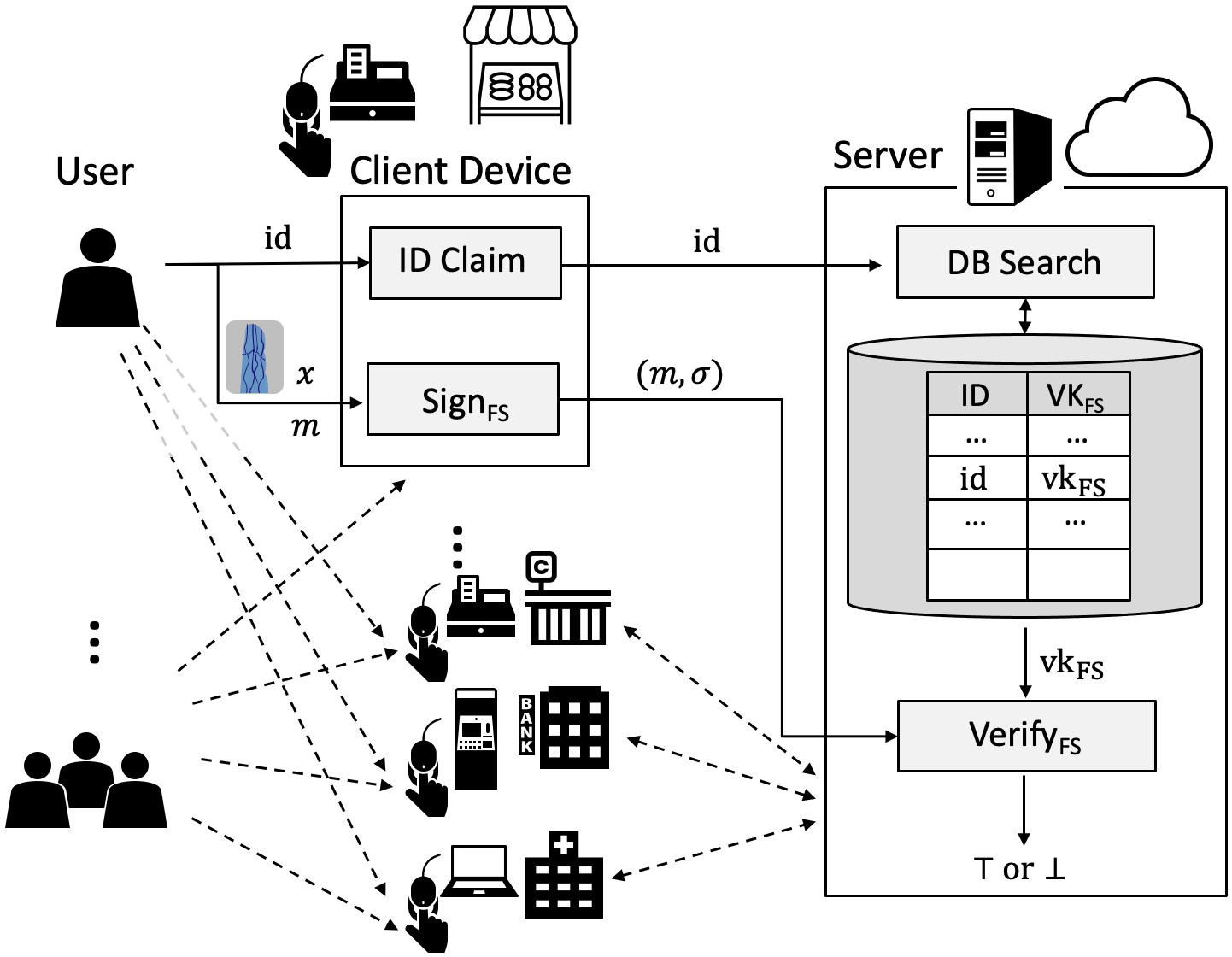}
 \end{center}
 \end{minipage}
 \caption{\small Authentication systems based on fuzzy extractor (left) and fuzzy signature (right). The user authenticates/signs by using its fuzzy biometric $x$ (depicted as a finger-vein). Each service provider (depicted as a supermarket, bank, etc...) has a single or several client devices (depicted as a finger-vein scanner) and a user can use any of them to authenticate itself to the server. $\mathsf{ID}$ denotes the user identity, $\mathsf{HD}$ denotes the helper data, and   $\mathsf{VK}$ and $\mathsf{VK_{FS}}$ denote the verification key of a standard signature scheme and a fuzzy signature scheme, respectively. }
 \label{fig:system}
 \end{figure*}

\smallskip
\noindent{\underline{\textit{Fuzzy signature.}}
The main primitive we focus on in this paper --- \emph{fuzzy signatures} --- can potentially be used to solve this question.
Fuzzy signatures, originally introduced in \cite{ACNS:TMMHN15}, are a special type of signature schemes that allow users to directly use their fuzzy biometrics as the signing key \emph{without} requiring any additional information. 
The description of a fuzzy signature is provided in \cref{fig:architecture}. 
Note that a verification key $\mathsf{vk_{FS}}$ of a fuzzy signature is implicitly associated to the user biometric and informally holds a similar purpose as a helper data for FE. 
\begin{figure} [htbp]
 \centering
	\includegraphics[ width=0.3\textwidth ]{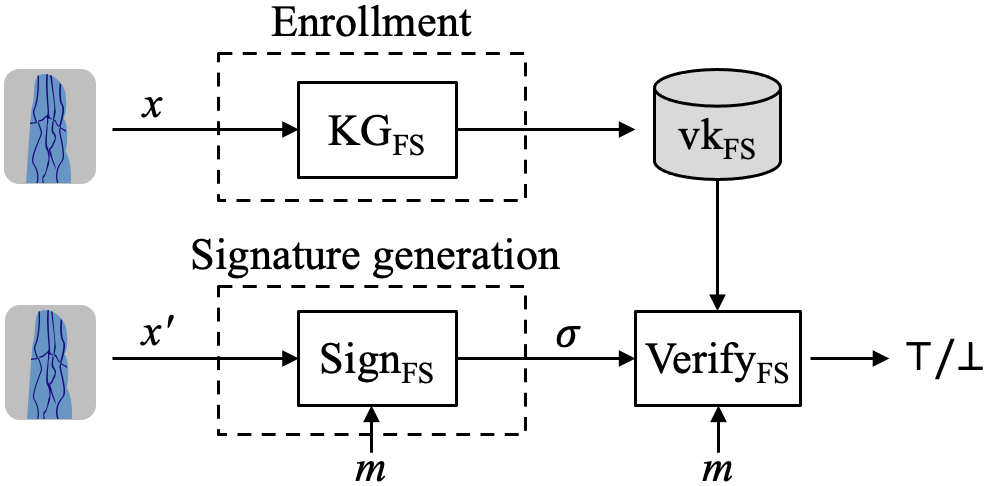}
 \caption{\small Description of fuzzy signature. A user with biometrics $x$ enrolls by generating a verification key $\mathsf{vk}$. The same user with a slightly different biometrics $x'$ can sign \emph{only} using $x'$ and creates a signature $\sigma$ that verifies with respect to $\mathsf{vk}$ generated via $x$. }
 \label{fig:architecture}
 \end{figure}

Using a fuzzy signature, we can construct a device-free biometric authentication system as in \cref{fig:architecture} (right).
From a user experience perspective, it is identical to the naive and FE-based systems: a user can show up empty-handed and authenticate itself by scanning its biometrics. 
In contrast, from a security point of view, the fuzzy signature-based system takes the best of the two systems: it provides BTP since the server no longer needs to store the user biometrics, and an adversary cannot collect user-specific information (e.g., helper data) since the client device and the server communicate non-interactively. 
\cref{tab:risk_of_exposure} gives a qualitative comparison of the risk $(R) = (P) \times (I)$ of the three systems: the naive, FE-based, and fuzzy signature-based systems. 
It can be checked that the fuzzy signature-based biometric authentication system achieves the lowest security risk among the three systems. 

\begin{table}[htb]
\centering
\caption{Comparison of the security risk of three biometric authentication systems.}
\label{tab:risk_of_exposure}
{\small
\begin{threeparttable}
\begin{tabular}{l|l|l|c}
System & \begin{tabular}{l} Possibility of \\ data exposure ($P$)\end{tabular} & \begin{tabular}{l} Impact of \\ data exposure ($I$) \end{tabular} & \begin{tabular}{c} Risk \\ $(R)$ \end{tabular} \\ \hline
Naive & Small \begin{tabular}{l} (from server)\end{tabular} & Very Large ~~(biometrics) & Large\\
FE-based & Large \begin{tabular}{l} (from server + \\ client devices)\end{tabular} & Small\tnote{\dag} \begin{tabular}{l}($\mathsf{hd}$)\end{tabular} & Middle\\
FS-based & Small \begin{tabular}{l}(from server)\end{tabular} & Small\tnote{\dag} \begin{tabular}{l}($\mathsf{vk_{FS}})$\end{tabular} & Small \\
\end{tabular}
 \begin{tablenotes}
  \item[\dag] \footnotesize{Note FE and FS are designed so that recovering biometric information from $\mathsf{hd}$ (helper data) and $\mathsf{vk_{FS}}$ (verification key) are hard, respectively.}
 \end{tablenotes}
\end{threeparttable}
}
\end{table}

So far, fuzzy signatures sound all good and well. However, despite its potential values, subsequent researches on the original paper~\cite{ACNS:TMMHN15} are quite limited and only produced by a small group~\cite{ACNS:MTMH16,ISPEC:KFNTMO17,IJIS:TMMHN19}%
\footnote{
\cite{IJIS:TMMHN19} is the full-version of \cite{ACNS:TMMHN15} and \cite{ACNS:MTMH16} with additional sections. Below, we mainly cite the full-version \cite{IJIS:TMMHN19}.
}.
The main reason for the unfortunate disparity between its potential value and amount of related work seems to stem from the fact that fuzzy signatures are mainly presented in a very cryptographically heavy and theory-oriented manner. 
Indeed, the state-of-the-artwork \cite{IJIS:TMMHN19} provides a generic construction of fuzzy signatures but the underlying building blocks are themselves novel to their work, and it is difficult to extract the practical relevance of such construction. 
An equally large (or perhaps larger) issue is that a critical discussion on whether real-life biometrics can be used to securely instantiate fuzzy signatures is completely missing. 
\cite{IJIS:TMMHN19} builds on the assumption that biometrics can provide large min-entropy.
However, it is not clear whether this is a feasible assumption to make for real-world biometrics, and besides, it is not even clear how to validate the feasibility of such an assumption. 
Therefore, although \cite{IJIS:TMMHN19} provides a potentially elegant solution to a more ideal biometric authentication system, the feasibility of the solution is completely left open to questions. 
We finally note that concrete discussions on biometric entropy is a reoccurring issue for FE as well and this is usually one of the main sources impeding a theoretically sound deployment of biometric authentication in practice.

\subsection{Our Contribution}
 In this work, we show that fuzzy signatures can indeed be efficiently and securely implemented in practice, and advocate the benefit of further practice-oriented research on fuzzy signatures. 
Our contribution is threefold:
(i) we provide a much simpler, more efficient, and direct construction of fuzzy signature compared to \cite{IJIS:TMMHN19}. Very roughly, depending on the amount of min-entropy we can extract from the fuzzy biometric, our construction can be proven secure based on the standard discrete logarithm (\DL) assumption or proven unconditionally secure in the generic group model~\cite{EC:Shoup97};
(ii) we establish novel statistical techniques to experimentally evaluate the conditions on biometrics that are required to securely instantiate fuzzy signatures; and (iii) we provide experimental results using real-world finger-vein dataset to show that finger-veins from a single hand can be used to construct efficient and secure fuzzy signatures. 
The statistical method provided in this work is quite general so we believe this to be an independent interest for other works such as evaluating the biometric entropy to securely instantiate fuzzy extractors. 
Below, we expand on each of our contributions.

\noindent{\underline{\textit{(i) Simple and efficient construction of fuzzy signature.}}
We provide a simple and efficient fuzzy signature scheme by tweaking the classical Schnorr signature scheme~\cite{C:Schnorr89}. 
Similarly to prior works, we rely on a tool called \emph{linear sketch} to bridge fuzzy biometrics and cryptographic primitives (e.g., signing keys). 
In our work, we simplify the definition of linear sketch and provide a conceptually cleaner construction of linear sketch based on a fundamental mathematical object called \emph{lattices}. 
At a high level, the specific type of lattice being used dictates how unwastefully we use the entropy provided by the fuzzy biometrics and how well we approximate the distance metric of the fuzzy biometrics by the distance metric induced by the underlying cryptographic primitive. 
With this abstraction, we show that a so-called \emph{triangular} lattice allows to best approximate the Euclidean distance and observe that previous works \cite{IJIS:TMMHN19} implicitly used a suboptimal lattice. 
The security assumption that underlies the security of our fuzzy signature is a simple-to-state variant of the standard $\DL$ assumption
considered jointly with the security of a linear sketch scheme, which we coin as the $\DL$ \emph{with sketch} ($\DLsketch$) assumption.
We provide discussion on the hardness of $\DLsketch$, and give some collateral evidence that if the quantity that we call \emph{conditional false matching rate} $\CFMR$ of the distribution of fuzzy biometrics is sufficiently small, then the $\DLsketch$ assumption is implied by the standard $\DL$ assumption. 
Moreover, even if $\CFMR$ is relatively large (which may be the case in practice), we show that the $\DLsketch$ assumption holds uncontionally in the generic group model~\cite{EC:Shoup97} . 

\noindent{\underline{\textit{(ii) Statistical method for evaluating fuzzy biometrics.}}}
There are two conditions that fuzzy biometrics must satisfy for fuzzy signatures. 
As mentioned above, one is that $\CFMR$ must be small. 
The other is that another quantity called the \emph{false non-matching rate} $\FNMR$ must be small. 
Roughly, $\FNMR$ and $\CFMR$ dictate the correctness and security of fuzzy signatures, respectively, where concretely we require $\FNMR \lessapprox 5\% ( \approx 2^{-4.32})$ and $\CFMR \lessapprox 2^{-112}$. 
Prior works \cite{ISPEC:KFNTMO17,IJIS:TMMHN19} failed to provide any formal evidence as to whether natural real-world biometrics can provide such amount of $\FNMR$ and $\CFMR$. 
This is a major setback for fuzzy signatures (and possibly one of the reasons why it has not attracted much serious attention) since a user may end up requiring multiple biometrics, say its iris \emph{and} fingerprints of \emph{both} hands, to authenticate itself. 
Such a procedure would severely deteriorate user experience and would defeat the purpose of using fuzzy signatures. 
While $\FNMR$ is a standard metric in the area of biometrics authentication and we know how to empirically estimate them using real-world biometric datasets, no such method is known for $\CFMR$ since it is a metric intertwined with a linear sketch. 
To make matters worse, since $\CFMR$ is a much smaller value (i.e., $2^{-112}$) compared to $\FNMR$ (i.e., $5\%$), we cannot use prior methods to provide any meaningful estimations.

Thus, our second contribution is to establish a systematic procedure to evaluate the values of $\CFMR$ of any fuzzy biometrics. 
We divide the problem of estimating $\CFMR$ into two subproblems and provide details on how to solve them individually. 
The first subproblem is formulated in a way to detach the notion of linear sketch from $\CFMR$ and allows us to view the problem entirely as a biometric problem, while the second subproblem deals with converting the solution of the biometric problem to the initial problem. 
At a high level, our approach to the two subproblems is the following: 
To solve the first subproblem, we borrow techniques from \emph{extreme value analysis} (EVA), a statistical method for evaluating very rare events by using only an ``extreme'' subset of a given dataset \cite{BOOK:CBTD01,NIST:Michael12}. 
This allows us to estimate $\CFMR \lessapprox 2^{-112}$ with high confidence. 
For the second subproblem, we use statistical $t$-tests to (informally) establish that certain statistics of biometrics are uncorrelated with the sketch.

\noindent{\underline{\textit{(iii) Efficiency analysis.}}} 
Finally, we use real-world finger-vein biometrics to conclude that fuzzy signatures can be constructed efficiently and securely using only 4 finger-veins from a single hand. 
That is, a user only needs to put one of their hands on the device to authenticate itself and nothing more. 
We first estimate the concrete values of $\FNMR$ and $\CFMR$ using the method stated above and experimentally show that the conditions $\FNMR \lessapprox 5\%$ and $\CFMR \lessapprox 2^{-112}$ hold.
We then combine everything and provide a concrete set of parameters for our fuzzy signature scheme. 
For instance, to achieve $112$-bits of security, the signature size can be as small as $1256$ bytes, and the running time for both signing and verification is only a few milliseconds.

\bigskip
\noindent\textbf{Organization.}  
In \cref{sec:preliminaries}, we define fuzzy signatures and prepare the notion of \emph{fuzzy key setting} that allows us to handle biometrics in a cryptographically sound manner.
In \cref{sec:linear_sketch_definition}, we define linear sketch: a tool allowing to bridge biometric data and cryptographic keys. 
In \cref{sec:fs_from_dl}, we provide a simple construction of fuzzy signature based on a slight variant of the $\DL$ problem assuming that the biometrics satisfies some conditions.
In \cref{sec:linear_sketch}, we equip the fuzzy key setting with a tool called lattice, and propose a concrete instantiation of a linear sketch scheme. 
In \cref{sec:concrete_fks}, we provide statistical techniques to estimate whether a specific type of biometrics satisfies the above mentioned conditions. 
Finally, in \cref{sec:efficiency_analysis}, we combine all the discussions together and provide a concrete instantiation of fuzzy signature using real-world finger-vein biometrics.


\section{Fuzzy Data and Fuzzy Signatures}
\label{sec:preliminaries}
To formally define fuzzy signatures, we must first formalize how we treat fuzzy data (i.e., biometrics); how are fuzzy data represented, what is the metric to argue closeness of fuzzy data, what kind of error distribution we consider to model ``fuzziness'' of data, and so on. To this end we first define the notion of \emph{fuzzy key setting} below. 

\subsection{Fuzzy Key Setting} \label{sec:fuzzy_key_setting}
A fuzzy key setting $\FuzKeySet$ consists of the following 5-tuple $\FKSparam$ and defines all the necessary information to formally treat fuzzy data in a cryptographic scheme.
\begin{itemize}
	\item[Fuzzy Data Space $\FuzDataSpace:$] This is the space to which a possible fuzzy data $x$ belongs.
We assume that $\FuzDataSpace$ forms an abelian group.
	\item[Distribution $\FuzDataDist:$] The distribution of fuzzy data over $\FuzDataSpace$. I.e., $\FuzDataDist: \FuzDataSpace \rightarrow \RR$.
        \item[Acceptance Region Function $\AR: \FuzDataSpace \to 2^{\FuzDataSpace}:$]
This function maps a fuzzy data $x \in \FuzDataSpace$ to a subspace $\AR(x) \subset \FuzDataSpace$ of the fuzzy data space $\FuzDataSpace$.
(If $x' \in \AR(x)$, then $x'$ is considered \lq\lq close'' to $x$.)
We require $x \in \AR(x)$ for all $x \in \FuzDataSpace$.
Based on $\AR$, the \emph{false matching rate} ($\FMR$) and the \emph{false non-matching rate} ($\FNMR$) are determined.
We define $\FMR$ by
$
\FMR := \Pr[x, x' \samp \FuzDataDist : x' \in \AR(x)].
$ \FNMR is defined below. 
	\item[Error Distribution $\ErrDist:$] This models the measurement error of fuzzy data. We assume the ``universal error model'' where the measurement error is independent of the users. 
	\item[Error Parameter $\ErrFNMR$:] The error parameter $\ErrFNMR \in [0, 1]$ defines $\FNMR$, where $\FNMR := \Pr[x \gets \FuzDataSpace; e \samp \ErrDist: x + e \notin \AR(x)] \le \ErrFNMR$.
\end{itemize}

\subsection{Fuzzy Signatures}
Using the fuzzy key setting, we can formally define fuzzy signatures. Note that in a fuzzy signature scheme, a signing key $\Sigsk$ will not be explicitly defined since the fuzzy data $x$ will play the role of the signing key.

\begin{definition} [Fuzzy Signature] A fuzzy signature scheme $\ProtFSig$ for a fuzzy key setting $\FuzKeySet = \FKSparam$ with message space $\SigMesSpace$ is defined by the following algorithms:
\begin{itemize}
	\item[ $\FSigSetup(1^\secpar, \FuzKeySet) \rightarrow \FSigpp:$ ]The setup algorithm takes as inputs the security parameter $1^\secpar$ and the fuzzy key setting $\FuzKeySet$ as input and outputs a public parameter $\FSigpp$. 
	\item[ $\FSigKeyGen(\FSigpp, x) \rightarrow \FSigvk:$] The key generation algorithm takes as inputs the public parameter $\FSigpp$ and a fuzzy data $x \in \FuzDataSpace$, and outputs a verification key~$\FSigvk$.
	\item[ $\FSigSign(\FSigpp, x', \msg) \rightarrow \FSigsig:$] The signing algorithm takes as inputs the public parameter $\FSigpp$, a fuzzy data $x' \in \FuzDataSpace$ and a message $\msg \in \SigMesSpace$, and outputs a signature $\FSigsig$.
\end{itemize}
\end{definition}

We define correctness and $\eucma$ security for fuzzy signatures. 
Roughly, correctness stipulates that a signature generated using fuzzy data $x$ verifies with respect to a verification key generated by a fuzzy data $x' \in \AR(x)$. 
$\eucma$ security is similar to those of standard signatures except that the challenger uses $x' \in \AR(x)$ to respond to signing queries rather than the original $x$ used to generate the verification key. 

Formally, we define $\delta$-correctness and $\eucma$ security of a fuzzy signature. 
\noindent
\textbf{$\delta$-Correctness.}
We say a fuzzy signature scheme $\ProtFSig$ for a fuzzy key setting $\FuzKeySet$ is \emph{$\delta$-correct} if the following holds for all $\msg \in \SigMesSpace$:
{\small
\begin{align*}
	\Pr[\FSigpp& \lruns \FSigSetup(1^\secpar, \FuzKeySet);~ x \samp \FuzDataSpace; \FSigvk \lruns \FSigKeyGen(\FSigpp, x);~ \\
	&\hspace{1cm} e \samp \ErrDist; \FSigsig \lruns \FSigSign(\FSigpp, x + e, \msg) : \\
	&\hspace{1cm}\FSigVrfy(\FSigpp, \FSigvk, \msg, \FSigsig) = \top] \ge 1 - \delta. 
\end{align*}
}
\noindent \textbf{$\eucma$ Security.}
The security of a fuzzy signature scheme $\Pi_{\mathsf FS}$ for a fuzzy key setting $\FuzKeySet$ is defined by the following game. The model captures the scenario where the signatures are generated by a slightly different fuzzy data each time. 

\begin{itemize}
	\item	[{Setup:}] The challenger runs $\FSigpp \lruns \FSigSetup(1^\secpar, \FuzKeySet)$, $x \samp \FuzDataSpace$, $\FSigvk \lruns \FSigKeyGen(\FSigpp, x)$, and provides the adversary $\mathcal A$ with the public parameter~$\FSigpp$ and the verification key~$\FSigvk$. Finally, it prepares an empty set $\SQset = \emptyset$. 
	
	\item [{Signing Queries:}] The adversary $\A$ may adaptively submit messages. When $\A$ submits a message $\msg \in \SigMesSpace$ to the challenger, the challenger  samples $e \samp \ErrDist$ and runs $\FSigsig \lruns \FSigSign(\FSigpp, x + e, \msg)$. It then provides $\FSigpp$ to $\A$ and updates the set as $\SQset \leftarrow \SQset \cup \set{\msg}$. 
	\item [{Output:}] Finally, $\mathcal A$ outputs a pair $(\msg^*, \FSigsig^*)$. The adversary $\mathcal A$ wins if	$\msg^* \notin \SQset ~\land~  \FSigVrfy(\FSigpp, \FSigvk, \msg^*, \FSigsig^*) = \top$.
\end{itemize}

\noindent
The advantage of $\A$ is defined as its probability of winning the above game. 
A fuzzy signature scheme $\ProtFSig$ is called $\eucma$ secure if the advantage is negligible for all PPT adversaries.


\section{Linear Sketch} \label{sec:linear_sketch_definition}

In this section, we define a \emph{linear sketch}, which has served as the main building block in previous generic constructions of fuzzy signature \cite{IJIS:TMMHN19}.%
\footnote{ 
We slightly deviate from prior definitions: we adopt a \lq\lq key encapsulation''-like syntax while \cite{IJIS:TMMHN19} adopts an \lq\lq encryption''-like syntax. 
Our syntax allows for a more simple, direct, and efficient construction. 
}
Recall the main purpose of this was to \lq\lq bridge'' fuzzy data and standard cryptographic operations.
It is associated with a fuzzy key setting and consists of two main algorithms $\Sketch$ and $\DiffRec$ (see \cref{fig:linear_sketch}). 
The formal definition is provided in Def.~\ref{def:linear_sketch} and a high-level description of the linear sketch follows. 

\begin{figure} [ht]
  \centering
	\includegraphics[ width=0.3\textwidth ]{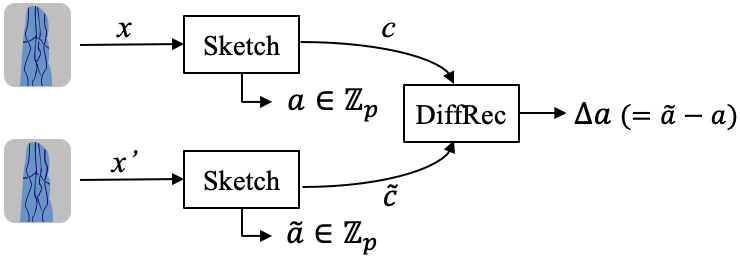}
  \caption{Illustration of the linear sketch when $\Lambda = \ZZ_p$. } 
   \label{fig:linear_sketch}
 \end{figure}

\noindent \textit{Overview.} $\Sketch$ is used to \lq\lq process'' a fuzzy data $x$ to extract a cryptographic secret (that we call a proxy key) $a$ that is an element of some abelian group and with which actual cryptographic operations (such as the signing operation of the Schnorr signature scheme) are performed.
$\Sketch$ also generates a corresponding \lq\lq sketch'' $\sketch$ of $x$ and $a$.
The sketch $\sketch$ is used to \lq\lq absorb'' the fluctuation occurred in measuring fuzzy data.
For example, suppose a fuzzy data is measured twice (e.g. once for key generation and once for signing) and sketch-proxy key pairs $(\sketch, a)$ and $(\tsketch, \ta)$ are generated (see  \cref{fig:linear_sketch}). 
Then by using the \emph{difference reconstruction} algorithm $\DiffRec$ with the two sketches $\sketch$ and $\tsketch$, we are able to compute the difference $\Delta a = \ta - a$.
This difference $\Delta a$ is then used in the verification algorithm of our fuzzy signature scheme to \lq\lq adjust'' the difference in the proxy keys $a$ and $\ta$.
Since the proxy key $a$ and the original fuzzy data $x$ are used as secret information, it is naturally required that the sketch $c$ does not reveal too much of $a$ and $x$. 

\noindent \textit{Definition.} Formally, a linear sketch associated with a fuzzy key setting $\FuzKeySet$ and an abelian group $\KeyDesc$ is defined as follows.
\begin{definition} [Linear Sketch]  \label{def:linear_sketch}
Let $\FuzKeySet = \FKSparam$ be a fuzzy key setting and $\KeyDesc = (\KeySpace, +)$ be a description of a (finite) abelian group.
A linear sketch scheme $\ProtLinS$ for $\FuzKeySet$ and $\KeyDesc$ is defined by the following P PT algorithms:

\begin{itemize}
	\item[$\LinSSetup(\FuzKeySet, \KeyDesc) \rightarrow \LinSpp:$] The setup algorithm takes as input the fuzzy key setting $\FuzKeySet$ and the description $\KeyDesc$, and outputs a public parameter $\LinSpp$.
Here, we assume $\LinSpp$ includes the information of $\KeyDesc = (\KeySpace, +)$.
	\item[$\Sketch(\LinSpp, x) \rightarrow (\sketch, a):$] The deterministic sketch algorithm takes as inputs the public parameter $\LinSpp$ and a fuzzy data $x \in \FuzDataSpace$, and outputs a \emph{sketch}~$\sketch$ and a \emph{proxy key} $a \in \KeySpace$. 
	\item[$\DiffRec(\LinSpp, \sketch, \tsketch) \rightarrow  \Delta a:$] The deterministic difference reconstruction algorithm takes as inputs the public parameter $\LinSpp$ and two sketches $(\sketch, \tsketch)$ (supposedly output by $\Sketch$), and outputs the \emph{difference} $\Delta a \in \KeySpace$.
\end{itemize}
\end{definition}

\noindent
\textbf{Correctness.}
We say a linear sketch scheme $\ProtLinS$ for a fuzzy key setting $\FuzKeySet$ and $\KeyDesc$ is \emph{correct} if for all $x, x' \in \FuzDataSpace$ such that $x' \in \AR(x)$ and all $\LinSpp \in \LinSSetup(\FuzKeySet, \KeyDesc)$, if $(\sketch, a) \lruns \Sketch(\LinSpp, x)$ and  $(\tsketch, \ta) \lruns \Sketch(\LinSpp, x')$, then we have $\ta - a = \DiffRec(\LinSpp, \sketch, \tsketch)$.

\smallskip
\noindent
\textbf{Linearity.}
We say a linear sketch scheme $\ProtLinS$ satisfies \emph{linearity} if there exists a deterministic PT algorithm $\Simsketch$ satisfying the following: For all $\LinSpp \in \LinSSetup(\FuzKeySet, \KeyDesc)$ and all $x, e \in \FuzDataSpace$, if $(\sketch, a) \lruns \Sketch(\LinSpp, x)$  and $(\tsketch, \Delta a) \lruns \Simsketch(\LinSpp, \sketch, e)$, then we have $\Sketch(\LinSpp, x + e) = (\tsketch, a + \Delta a)$. 

In above, we have not formally defined  the intuition that a sketch $c$ does not leak the information of the fuzzy data $x$ and proxy key $a$. 
This is implicitly handled by the hardness assumption underlying the security of the fuzzy signature, and we discuss it in the next section (see Def.~\ref{def:dl} for an overview).


\section{Fuzzy Signature from Discrete Log}
\label{sec:fs_from_dl}
In this section, we provide a simple and efficient construction of fuzzy signature based on a variant of the \DL problem.

\subsection{Construction}
 An overview of the construction of our fuzzy signature scheme $\ProtFSig^\DL$ is depicted in \cref{fig:dl_fuzzy}. 
At a high level, our construction can be seen as providing a wrapper around the classical Schnorr signature  \cite{C:Schnorr89} to additionally handle fuzzy biometrics via the linear sketch. 
During key generation (\underline{KeyGen} in \cref{fig:dl_fuzzy}), a user with biometrics $x$ generates a sketch and a proxy key $(\sketch, a)$ from $x$ using $\Sketch$, and
sets $\FSigvk$ as the sketch $\sketch$ and a verification key $h = g^a$ of the Schnorr signature.
To sign (\underline{Sign} in \cref{fig:dl_fuzzy}), the user with biometrics $x'$ (slightly different from $x$) generates $(\tsketch, \ta)$ from $x'$ using $\Sketch$ and uses $\ta \in \ZZ_p$ as an ``ephemeral" signing key for the Schnorr signature and constructs a Schnorr signature $\tilde{\sigma}$. 
Here, note that the Schnorr verification key of this signature is implicitly set as $\tildeh = g^{\ta}$. The fuzzy signature $\FSigsig$ consists of $\tilde{\sigma}$ and the sketch $\tsketch$.
Finally, to verify (\underline{Verify} in \cref{fig:dl_fuzzy}) a fuzzy signature $\FSigsig$, we first use the algorithm $\DiffRec$ of the linear sketch to recover $\Delta a = \ta - a$. 
Then, we use $\Delta a$ to recover the implicit Schnorr verification key $\tildeh$ from $h$, and use it to verify $\tilde{\sigma}$. 
 
\begin{figure} [ht]
  \centering
	\includegraphics[ width=0.4\textwidth ]{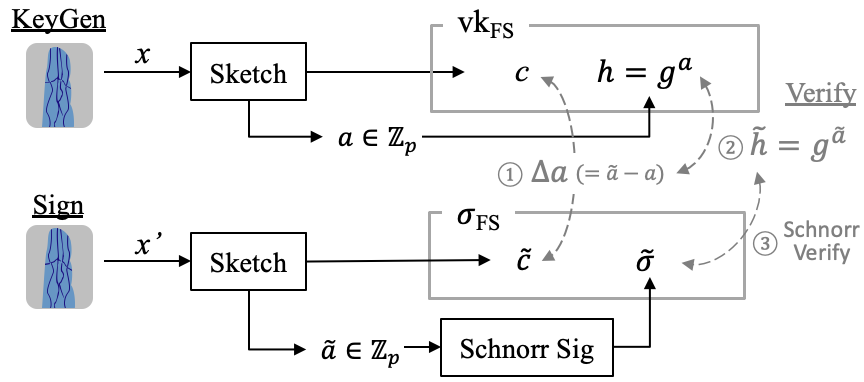}
  \caption{Our fuzzy signature scheme $\ProtFSig^\DL$. The gray items indicate the procedure of the verification algorithm. } 
   \label{fig:dl_fuzzy}
 \end{figure}
The formal description of our fuzzy signature scheme is provided in \cref{fig:fs_from_dl}. 
The image of the hash function $\hash$ is $\ZZ_p$ and is modeled as a random oracle in the security proof. 
\begin{figure}[htb!]
\begin{minipage}[t]{0.24\textwidth}
\underline{ $\FSigSetup(1^\secpar, \FuzKeySet)$: }
    \begin{algorithmic}[1]
	\State \small{$\GGG  \gets \ggen(1^\secpar)$}
	\State $\KeyDesc = (\ZZ_p, +)$
	\State $\LinSpp \gets \LinSSetup(\FuzKeySet, \KeyDesc)$
	\State \Return $\FSigpp = ( \GGG, \LinSpp )$
    \end{algorithmic}
    \vspace{0.15cm}
\end{minipage}
\hfill
%
\begin{minipage}[t]{0.22\textwidth}
\underline{ $\FSigKeyGen(\FSigpp, x )$ }
    \begin{algorithmic}[1] 
    	\State $  ( \GGG, \LinSpp ) \gets \FSigpp$
	\State \small{$  (\sketch, a) \gets \Sketch( \LinSpp, x ) $}   
	\State \Return $ \FSigvk = ( g^a, \sketch )  $
    \end{algorithmic}
    \vspace{0.15cm}
\end{minipage}
%
\begin{minipage}[t]{0.23\textwidth}
\underline{ $\FSigSign( \FSigpp, x', \msg )$: }
    \begin{algorithmic}[1]
    	\State  $  ( \GGG, \LinSpp ) \gets \FSigpp$
	\State \small{$(\tsketch, \ta) \gets \Sketch( \LinSpp, x' )$}   
	\State $r \gets \ZZ_p$                                  
	\State $\beta \gets \hash( g^\ta, g^r, \msg )$
	\State $z \gets  \beta \cdot \ta + r$
	\State \Return $ \FSigsig = ( \beta, z, \tsketch ) $
    \end{algorithmic}
    \vspace{0.15cm}
\end{minipage}
\hfill
%
\begin{minipage}[t]{0.24\textwidth}
\underline{ $\FSigVrfy( \FSigpp, \FSigvk, \msg, \FSigsig)$ }
    \begin{algorithmic}[1] 
    	\State \small{$  ( \GGG, \LinSpp ) \gets \FSigpp$}
	\State $( h, \sketch ) \gets \FSigvk  $
	\State $( \beta,  z, \tsketch ) \gets \FSigsig$    
	\State $\Delta a \gets \DiffRec( \LinSpp, \sketch, \tsketch )$
	\State $\tildeh \gets h \cdot g^{\Delta a}$
	\State $R \gets g^z \cdot \tildeh^{ -\beta }$
	\State \Return $\top$ iff $ \beta = \hash( \tildeh, R, \msg ) $
    \end{algorithmic}
    \vspace{0.15cm}
\end{minipage}
\caption{ 
Construction of fuzzy signature $\ProtFSig^\DL$. 
}
\label{fig:fs_from_dl}
\end{figure}
\noindent{\textit{Efficiency.}}
The verification key consists of one group element in $\GG$ and a sketch. 
The signature consists of two elements in $\ZZ_p$ and a sketch. 
Notably, the only difference from the Schnorr signature is the linear sketch  component. 

\subsection{Correctness and Security Proof}\label{sec:security_dl}
\noindent{\textit{Correctness.}}
The correctness of $\ProtFSig^\DL$ is provided below. As correctness is evident from \cref{fig:dl_fuzzy}, we omit the proof to \cref{app_sec:fs_from_dl}. 

\begin{theorem} \label{thm:dl_correctness}
If the linear sketch $\ProtLinS$ is correct, then the fuzzy signature $\ProtFSig^\DL$ in \cref{fig:fs_from_dl} is $\ErrFNMR$-correct, where $\ErrFNMR$ is the error parameter of the fuzzy key setting~$\FuzKeySet$.
\end{theorem}

\noindent{\textit{Security.}}
The security of our fuzzy signature $\ProtFSig^\DL$ is based on a variant of the $\DL$ problem where the secret exponent is a proxy key $a$ generated by the linear sketch scheme on input a random fuzzy data $x \gets \FuzDataDist$. 
The adversary is given the \DL instance $g^a$ along with the sketch $\sketch$. 
Formally, we define the $\DL$ \emph{with sketch} assumption in Def.~\ref{def:dl}. 
We provide detailed discussions in \cref{sec:analysis_dlsketch} to validate that the $\DL$ with sketch problem is as hard as the standard $\DL$ problem for our specific choice of biometrics and the linear sketch scheme.

\begin{definition}[ $\DL$ with sketch ] \label{def:dl}
Let $\ProtLinS$ be a linear sketch scheme for a fuzzy key setting $\FuzKeySet = \FKSparam$ with respect to a (finite) abelian group $\KeyDesc = (\ZZ_p, +)$. 
We say the discrete logarithm problem \emph{with sketch} ($\DLsketch$) assumption holds (relative to $\ggen$) if for all PPT adversaries $\A$, the following probability is upper bounded by $\negl(\secpar)$: 
{\small
\begin{align*}
\Pr\left[
\begin{array}{cl}
	\GGG = (\GG, p, g)  \gets \ggen(1^\secpar); \\
	 \LinSpp \gets \LinSSetup(\FuzKeySet, \KeyDesc);   \\
	  x \gets \FuzDataDist; 
	  (\sketch, a) \gets \Sketch( \LinSpp, x)
\end{array}
	 :
\begin{array}{rl}
	 \A(  \GGG, \LinSpp, g^a, \sketch) = a
\end{array}
\right].
\end{align*}
}
\end{definition}

The following theorem guarantees security of our fuzzy signature scheme~$\ProtFSig^\DL$  under the $\DLsketch$ assumption. 

\begin{theorem} \label{thm:dl_eucma}
If the $\DLsketch$ problem is hard and the linear sketch scheme $\ProtLinS$ satisfies linearity, then the fuzzy signature scheme $\ProtFSig^\DL$ in \cref{fig:fs_from_dl} is $\eucma$ secure. 
\end{theorem}
The proof is  similar to that of the Schnorr signature \cite{C:Schnorr89}, except that we additionally need to simulate the sketch $\sketch$ in the verification key and signatures without knowledge of the secret fuzzy data $x$. 
At a high level, the sketch in the verification key is handled by the $\DLsketch$ assumption and the sketches in the signatures are handled by the \emph{linearity} of the linear sketch (see Def.~\ref{def:linear_sketch}). 
We omit the full proof to \cref{app_sec:fs_from_dl}.


\section{Instantiating Linear Sketch over Lattices} \label{sec:linear_sketch}

In this section, we present our linear sketch scheme.
Our scheme is constructed over a fuzzy key setting with fuzzy data space $\FuzDataSpace = \RR^n$.%
\footnote{Throughout the rest of the paper we implicitly assume that real numbers are represented by some pre-determined number of significant digits in order to handle them on computers. }
Namely, we consider the natural setting where biometrics are represented by an $n$-dimensional vector in $\RR$. 
However, working directly with fuzzy data in $\RR^n$ is non-trivial since typical computations of cryptographic primitives (and in particular the Schnorr signature scheme) are performed over a discrete space such as $\ZZ_p$.
Moreover, recall that a linear sketch scheme needs to satisfy correctness and linearity, which roughly requires a \emph{linearity preserving} mapping of the fuzzy data space $\FuzDataSpace$ to the sketch and proxy key spaces. 
To deal with these issues, we associate the fuzzy data space $\FuzDataSpace$ with a mathematical object called \emph{lattice} known to have a discretized and linear nature. 
This connects fuzzy data and cryptographic primitives together, and allows to construct a linear sketch scheme. 

We also introduce a specific lattice called a \emph{triangular} lattice and show that it fits well with a fuzzy data space endowed by Euclidean metric. 
We finally discuss the hardness of the $\DLsketch$ assumption with respect to such concrete linear sketch scheme in \cref{sec:analysis_dlsketch}.

\subsection{Fuzzy Key Setting with a Lattice} \label{sec:specific_fuzzy_key_setting}
We first introduce the notion of lattices and then provide a concrete definition of a fuzzy key setting based on lattices. 

\noindent{\underline{\textit{Lattice background.}}}
Let $n \in \NN$ and $\mB \in \RR^{n \times n}$. 
\begin{itemize}
\item A \emph{lattice} spanned by $\mB$, denoted by $\Lattice(\mB)$, is defined by $\Lattice(\mB) := \{ \mB \vz | \vz \in \ZZ^n\}$.
$\mB$ is called the \emph{basis} of $\Lattice(\mB)$.

\item For a vector $\vx \in \RR^n$ and a lattice $\Lattice = \Lattice(\mB)$, the \emph{closest vector} (or \emph{lattice point}) of $\vx$ in $\Lattice$, denoted by $\CV_{\Lattice}(\vx)$, is a vector $\vy \in \Lattice$ satisfying $\| \vx -  \vy\|_2 \le \| \vx - \mB \vz \|_2$ for any $\vz \in \ZZ^n$.%
\footnote{If there are multiple vectors $\vy \in \Lattice$ satisfying this condition, then we consider some canonical ordering of the lattice points in $\Lattice$ and choose the first one according to the ordering to make it unique.}

\item For a lattice $\Lattice = \Lattice(\mB)$ and a vector $\vy \in \Lattice$, the \emph{Voronoi region} of $\vy$, denoted by $\VR_{\Lattice}(\vy)$, is defined by $\VR_{\Lattice}(\vy) := \{\vx | \vy = \CV_{\Lattice}(\vx) \}$.
Due to the translational symmetry of a lattice, we have $\VR_{\Lattice}(\vy) = \VR_{\Lattice}(\vzero) + \vy$.%
\footnote{For $X \subset \RR^n$ and $\vy$, we define $X + \vy =: \{\vx + \vy | \vx \in X\}$.}
See \cref{fig:AR_VR_triangular} for a pictorial example of $\VR_{\Lattice}(0)$.

\end{itemize}
\begin{figure} [ht]
  \centering
	\includegraphics[ width=0.45\textwidth ]{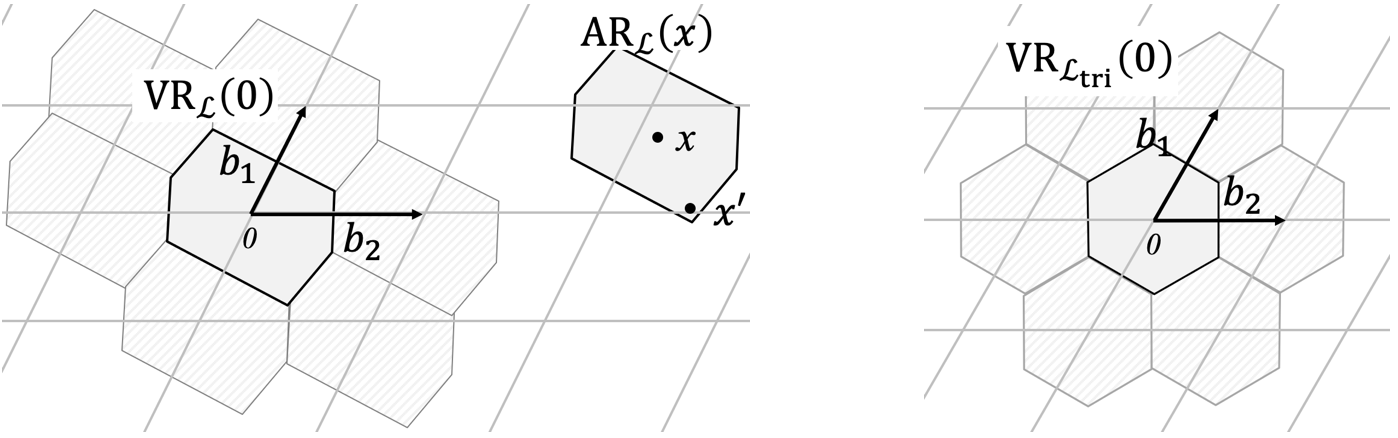}
  \caption{\small The left and right light gray grids depict two different lattices $\Lattice$ and $\Lattice_{\tt tri}$. The regions $\VR_{\Lattice}(0)$ is the set of points that has $0$ as the closest vector in the lattices and the region $\AR_{\Lattice}(\vx)$ denotes all the point that is considered to be ``close" to $\vx$. 
}  
   \label{fig:AR_VR_triangular}
 \end{figure}

\noindent{\underline{\textit{Fuzzy key setting based on a lattice.}}}
We define a fuzzy key setting $\FuzKeySet \allowbreak = \allowbreak (\FuzDataSpace, \allowbreak \FuzDataDist, \allowbreak \AR, \allowbreak \ErrDist, \allowbreak \ErrFNMR)$ with respect to a lattice as follows. 
\begin{itemize}
\item[Fuzzy data space $\FuzDataSpace$:] The fuzzy data space $\FuzDataSpace$ is $\RR^n$,
where $n \in \NN$ is specified by the context (e.g. the device which we use to measure fuzzy data).
We associate $\FuzDataSpace$ with a lattice $\Lattice = \Lattice(\mB)$ spanned by some basis $\mB \in \RR^{n \times n}$ such that the closest vectors in $\Lattice$ can be \emph{efficiently} computed.
We also associate $\FuzDataSpace$ with a natural number $p \in \NN$ that determines the support of $\FuzDataDist$ (see below). 

\item[Distribution $\FuzDataDist$:]
An efficiently sampleable distribution such that
the support of $\FuzDataDist$ satisfies the property that if $\vx \samp \FuzDataDist$, then $\mB^{-1} \vx \in [0,p)^n$.

\item[Acceptance region function $\AR$:]
We define the acceptance region function $\AR$ by $\AR(\vx) = \AR_{\Lattice}(\vx) := \{\vx' | \CV_{\Lattice}(\vx - \vx') = \vzero\}$.
Note that we have $\AR_{\Lattice}(\vx) = \VR_{\Lattice}(\vzero) + \vx$.
See \cref{fig:AR_VR_triangular} for a pictorial example of $\AR_{\Lattice} (\vx)$.

\item[Error distribution $\ErrDist$ and Error parameter $\ErrFNMR$:]
$\ErrDist$ is any efficiently samplable distribution over $\FuzDataSpace$ such that $\FNMR \le \ErrFNMR$.
\end{itemize}

\subsection{Construction of Linear Sketch} \label{sec:construction_of_linear_sketch}

Let $\FuzKeySet \allowbreak = \allowbreak (\FuzDataSpace = \RR^n, \allowbreak \FuzDataDist, \allowbreak \AR, \allowbreak \ErrDist, \allowbreak \ErrFNMR)$ be the fuzzy key setting as defined above.
Let $g_{\Lattice} : \FuzDataSpace \to \Lattice$ be the function%
\footnote{As far as correctness and linearity are concerned, the function $g_{\Lattice}$ can be any efficiently computable deterministic function satisfying (1) $g_{\Lattice}(\vx + \vy) = g_{\Lattice}(\vx) + \vy$ for all $\vx \in X$ and $\vy \in \Lattice$, and (2) $\|\mB^{-1} \vx\|_{\infty} \approx \|\mB^{-1} g_{\Lattice}(\vx) \|_{\infty}$.
We choose this particular function for its simplicity and efficiency.}
$
g_{\Lattice}(\vx) := \mB \bigl \lfloor \mB^{-1} \vx \bigr\rfloor.
$
Let $\linHashFamily = \{ \linhash : \ZZ_p^n \to \ZZ_p \}$ be a family of universal hash functions%
\footnote{Recall that $\linHashFamily = \{\linhash : D \to R\}$ is called universal if for all distinct elements $x,x' \in D$, we have $\Pr_{\linhash \samp \linHashFamily}[\linhash(x) = \linhash(x')] \le |R|^{-1}$.}
that satisfies linearity, namely, for all $\linhash \in \HashFamily$ and $\vx, \vy \in \ZZ_p^n$, we have $\linhash(\vx + \vy) = \linhash(\vx) + \linhash(\vy)$.
Using these ingredients, the description of our linear sketch scheme $\ProtLinS \allowbreak = \allowbreak (\LinSSetup, \allowbreak \Sketch, \allowbreak \DiffRec)$ for $\FuzKeySet$ and the additive group $(\ZZ_p, +)$ ($=: \KeyDesc$) is provided in \cref{fig:linear_sketch_construction}.
The auxiliary algorithm $\Simsketch$ used to show the linearity property is also included. 
A pictorial example (\cref{fig:concrete_linear_sketch}) and an intuitive explanation of $\Sketch$ and $\DiffRec$ follow.

\begin{figure}[htb!]
\begin{minipage}[t]{0.24\textwidth}
\underline{ $\LinSSetup(\FuzKeySet, \KeyDesc = (\ZZ_p, +))$: }
    \begin{algorithmic}[1]
        \State $\linhash \samp \linHashFamily$
        \State \Return $\LinSpp = (\KeyDesc, \linhash)$
    \end{algorithmic}
    \vspace{0.15cm}
\end{minipage}
\hfill
%
\begin{minipage}[t]{0.22\textwidth}
\underline{ $\Sketch(\LinSpp, \vx)$ }
    \begin{algorithmic}[1]
    	\State $\vy \samp g_{\Lattice}(\vx)$
	\State $\vc \samp \vx - \vy$
	\State $a \lruns  \linhash(\mB^{-1} \vy)$
	\State \Return $(\vc, a)$
    \end{algorithmic}
    \vspace{0.15cm}
\end{minipage}
%
\begin{minipage}[t]{0.23\textwidth}
\underline{ $\DiffRec(\LinSpp, \vc, \vc')$: }
    \begin{algorithmic}[1]
    	\State $\Delta \vy \samp \CV_{\Lattice}(\vc - \vc')$
	\State $\Delta a \samp \linhash(\mB^{-1} (\Delta \vy))$
	\State \Return $\Delta a$
    \end{algorithmic}
    \vspace{0.15cm}
\end{minipage}
\hfill
%
\begin{minipage}[t]{0.24\textwidth}
\underline{ $\Simsketch(\LinSpp, \vc, \ve)$ }
    \begin{algorithmic}[1]
    	\State $\vc' \samp \vc + \ve - g_{\Lattice}(\vc + \ve)$
	\State $\vy' \samp g_{\Lattice}(\vc + \ve)$
	\State $\Delta a \samp \linhash(\mB^{-1} \vy')$
	\State \Return $(\vc', \Delta a)$
    \end{algorithmic}
    \vspace{0.15cm}
\end{minipage}
\caption{ \small
Construction of linear sketch $\ProtLinS$ and the auxiliary algorithm $\Simsketch$ for the linearity property.
}\label{fig:linear_sketch_construction}
\end{figure}

\begin{figure} [ht]
  \centering
	\includegraphics[ width=0.35\textwidth ]{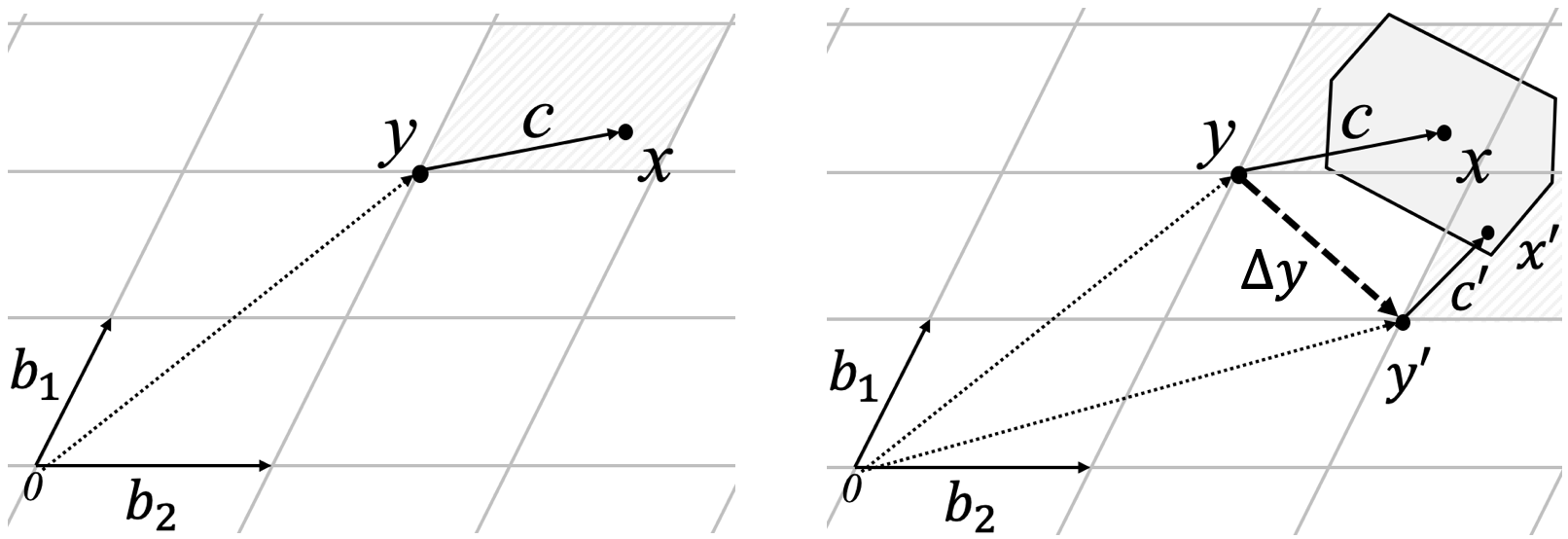}
  \caption{\small The left (resp. right) figure depicts algorithm $\Sketch$ (resp. $\DiffRec$). The shaded gray parallelogram denotes the fundamental parallelepiped spanned by the basis $\mB = [\vb_1, \vb_2]$. The gray hexagon denotes the acceptance region $\AR(\vx)$ of $\vx$. }  
   \label{fig:concrete_linear_sketch}
 \end{figure}

Algorithm $\Sketch(\LinSpp, \vx)$ first deterministically computes a lattice point $\vy = g_{\Lattice}(\vx) \in \Lattice$ with respect to the basis~$\mB$. 
As depicted in \cref{fig:concrete_linear_sketch}, the fundamental parallelepiped spanned by~$\mB$%
\footnote{
The fundamental parallelepiped spanned by $\mB$ is defined as the set $\set{ \mB \vw \mid \vw \in [0, 1)^n }$. 
}
 originated at $\vy$ always contains $\vx$. 
Then, the sketch $\vc$ is simply the shift $\vx - \vy$ and the proxy key $a$ is a hash of some \lq\lq canonical representation'' of the lattice point $\vy$. 
Now, it is clear from \cref{fig:concrete_linear_sketch} that if $\vx' \in \AR(\vx)$  is contained in the \emph{same} fundamental parallelepiped originated at $\vy$, then it produces the same proxy key $a$ since $\vy = g_{\Lattice}(\vx')$. 
However, as in the right figure in \cref{fig:concrete_linear_sketch}, this is not always the case. 
Therefore, we require a mechanism to relate the proxy key $a'$ (or equivalently $\vy'$) generated by $\vx'$ and those by $\vx$ only given their sketches $\vc'$ and~$\vc$. 
Recall that this was the core property of linear sketch that allowed us to meaningfully relate the secret keys generated from different $a$ and $a'$ for our fuzzy signature scheme (see \cref{fig:linear_sketch}).
Now, algorithm $\DiffRec(\LinSpp, \vc, \vc')$ exactly offers this mechanism. 
First, by definition $\vc - \vc' = (\vx - \vx') + (\vy' - \vy) $. 
Then since the vector $\vx - \vx' \in \VR(0)$ (see \cref{fig:concrete_linear_sketch}), $\CV_\Lattice(\vc - \vc')$ is the same as $\CV_\Lattice(\vy' - \vy) = \vy' - \vy$ since $\vy' - \vy$ are points contained in $\Lattice$. 
Hence, we can recover $\Delta \vy = \vy' - \vy$ (or equivalently $\Delta a = a' - a$) only from the sketches $\vc$ and $\vc'$. 

Formally, we have the following theorem. 
\begin{theorem} \label{thm:ls_thm}
The linear sketch scheme $\ProtLinS$ in \cref{fig:linear_sketch_construction} satisfies correctness and linearity (as per Def.~\ref{def:linear_sketch}).
\end{theorem}
\begin{proof}
To prove the theorem, we show that $\ProtLinS$ given in \cref{sec:linear_sketch} satisfies correctness and linearity (Def.~\ref{def:linear_sketch}).

\noindent{\underline{\textit{Correctness.}}}
Fix $\LinSpp = (\KeyDesc = (\ZZ_p, +), \linhash)$, and $\vx, \vx' \in \FuzDataSpace$ such that $\vx' \in \AR(\vx)$, which implies $\CV_{\Lattice}(\vx - \vx') = \vzero$.
Let $\vy = g_{\Lattice}(\vx)$ and $\vy' = g_{\Lattice}(\vx')$, and let
\begin{align*}
(\vc, a) &= (\vx - \vy, \linhash(\mB^{-1} \vy)) = \Sketch(\LinSpp, \vx),\\
(\vc', a') &= (\vx' - \vy', \linhash(\mB^{-1} \vy')) = \Sketch(\LinSpp, \vx').
\end{align*}

We have
\begin{align*}
\Delta \vy &= \CV_{\Lattice}(\vc - \vc') = \CV_{\Lattice} \bigl( (\vx - \vy) - (\vx' - \vy') \bigr)\\
 &\stackrel{(*)}{=} \CV_{\Lattice}(\vx - \vx') + \vy' - \vy \stackrel{(**)}{=} \vy' - \vy,
\end{align*}
where the equality (*) uses the fact that $\vy, \vy' \in \Lattice$, and the equality (**) uses $\CV_{\Lattice}(\vx - \vx') = \vzero$.
Using this, we see that
\begin{align*}
\Delta a &= \linhash \bigl( \mB^{-1} (\Delta \vy) \bigr) = \linhash \bigl( \mB^{-1} (\vy' - \vy) \bigr)\\
 &\stackrel{(*)}{=} \linhash( \mB^{-1} \vy' ) - \linhash( \mB^{-1} \vy) = a' - a,
\end{align*}
where the equality (*) uses the linearity of $\linHashFamily$.
This shows that $\DiffRec(\LinSpp, \vc, \vc') = a' - a$.
Thus, $\ProtLinS$ satisfies correctness.

\bigskip

\noindent{\underline{\textit{Linearity.}}}
We use the auxiliary algorithm $\Simsketch$ in \cref{fig:linear_sketch_construction}.
Fix $\LinSpp = (\KeyDesc = (\ZZ_p, +), \linhash)$ and $\vx, \ve \in \FuzDataSpace$.
Let
\begin{align*}
(\vc, a) &= (\vx - g_{\Lattice}(\vx), \linhash(\mB^{-1} g_{\Lattice}(\vx))) = \Sketch(\LinSpp, \vx),\\
(\vc', a') &= ((\vx + \ve) - g_{\Lattice}(\vx + \ve), \linhash(\mB^{-1} g_{\Lattice}(\vx + \ve)))\\
 &=\Sketch(\LinSpp, \vx + \ve)
\end{align*}
In order to show the linearity of $\ProtLinS$, it is sufficient to show that the following equality holds:
\begin{equation}
\bigl( \vc + \ve - g_{\Lattice}(\vc - \ve),~ \linhash(\mB^{-1} g_{\Lattice}(\vc + \ve) \bigr) = (\vc', a' - a), \label{eqn:ls_linearity}
\end{equation}
since the left hand sice is exactly $\Simsketch(\LinSpp, \vc, \ve)$.

For the first element in \cref{eqn:ls_linearity}, we have
\begin{align*}
\vc + \ve - g_{\Lattice}(\vc - \ve) &= \vx - g_{\Lattice}(\vx) + \ve - g_{\Lattice}(\vx - g_{\Lattice}(\vx) + \ve)\\
 &\stackrel{(*)}{=} \vx - g_{\Lattice}(\vx) + \ve - \bigl(g_{\Lattice}(\vx + \ve) - g_{\Lattice}(\vx) \bigr)\\
 &= \vx + \ve - g_{\Lattice}(\vx + \ve) = \vc',
\end{align*}
where the equality (*) uses the property of $g_{\Lattice}$ that $g_{\Lattice}(\vx' + \vy') = g_{\Lattice}(\vx') + \vy'$ for $\vx' \in \FuzDataSpace$ and $\vy' \in \Lattice$, and that $g_{\Lattice}(\vx) \in \Lattice$.

For the second element in \cref{eqn:ls_linearity}, we have
\begin{align*}
\linhash \bigl( \mB^{-1} g_{\Lattice}(\vc + \ve) \bigr) &= \linhash \bigl(\mB^{-1} (g_{\Lattice}(\vx - g_{\Lattice}(\vx) + \ve)) \bigr)\\
 &\stackrel{(*)}{=} \linhash \bigl(\mB^{-1}(g_{\Lattice}(\vx + \ve) - g_{\Lattice}(\vx)) \bigr)\\
 &=\linhash \bigl( \mB^{-1}g_{\Lattice}(\vx + \ve) - \mB^{-1} g_{\Lattice}(\vx) \bigr)\\
 &\stackrel{(**)}{=} \linhash(\mB^{-1} g_{\Lattice}(\vx + \ve)) - \linhash(\mB^{-1} g_{\Lattice}(\vx))\\
 &= a' - a,
\end{align*}
where the equality again uses the property of $g_{\Lattice}$ that $g_{\Lattice}(\vx' + \vy') = g_{\Lattice}(\vx') + \vy'$ for $\vx' \in X$ and $\vy' \in \Lattice$, and the equality (**) uses the linearity of $\linHashFamily$.
Hence, we can conclude that $\ProtLinS$ satisfies linearity.
\end{proof}

\subsection{Concrete Lattice for Efficient Linear Sketch} \label{sec:triangular_lattice}
Depending on the the type of lattice $\Lattice$ (or equivalently basis $\mB$), the computational complexity of $\CV_{\Lattice}$ and $g_{\Lattice}$ differs greatly. 
For our concrete instantiation of linear sketch, we use \emph{triangular} lattices. 
Geometrically, they are lattices that have regular hexagons as the Voronoi region $\VR$ (see the right hand side of \cref{fig:AR_VR_triangular} for an illustration). 
Over such a lattice, $\CV_{\triLattice}$ can be computed in time $O(n^2)$.
Moreover, other than they allow for efficient computations of $\CV_{\triLattice}$ and $g_{\triLattice}$, the acceptance region $\AR_{\triLattice}$ of triangular lattices reflects nicely the notion of ``closeness'' of most natural biometrics. 
In a typical biometric authentication, the most natural and widely-used way to judge two biometrics $\vx, \vx'$ are \lq\lq close'' is to calculate how close they are with respect to the Euclidean distance $\|\vx - \vx'\|_2$. 
The triangular lattice is a very suitable lattice in the sense that $\AR_{\triLattice}$ (which is a regular hexagon) best approximates the closeness induced by the Euclidean distance compared to any other lattice $\Lattice$. 
We note that casting the linear sketch schemes in \cite{IJIS:TMMHN19} in the framework of lattices, we see that they considered lattices with a square as the $\AR_{\Lattice}$ (i.e., a lattice with basis $\mB = d \cdot \mI_n$ for some positive real $d \in \RR$). See \cref{fig:acceptance_region} for a visual aid.  
Effectively, our lattice allows to extract more entropy from the underlying biometric since we are able to model more accurately the real closeness metric. 

A formal description of triangular lattices and how $\CV_{\triLattice}$ is implemented are provided in \cref{app_sec:cvp_triangular_lattice}.

\subsection{Security of the \DL Assumption with Sketch} \label{sec:analysis_dlsketch}
In \cref{sec:fs_from_dl}, we introduced the $\DLsketch$ assumption on which the security of our fuzzy signature scheme is based.
The main question is of course: how plausible is this assumption?
We argue that for our linear sketch scheme $\ProtLinS$ presented in this section, the $\DLsketch$ assumption is plausible if:
\begin{itemize}
\item the quantity that we call the \emph{conditional false matching rate} ($\CFMR$)
is \lq\lq small'', say, $\approx 2^{-\secpar}$ for a cryptographic security parameter $\secpar$, and
\item the standard $\DL$ assumption holds.
\end{itemize}
Here, for the linear sketch scheme $\ProtLinS$ over a fuzzy key setting $\FuzKeySet = (\FuzDataSpace, \FuzDataDist, \AR, \ErrDist, \ErrFNMR)$ with which a lattice $\Lattice = \Lattice(\mB)$ is associated, we define $\CFMR$ by
{\small
\[
\CFMR:= \Pr \left[\begin{array}{c} \vx, \vx \samp \FuzDataDist; \vc \samp \vx - g_{\Lattice}(\vx);\\ \vc' \samp \vx' - g_{\Lattice}(\vx')\end{array} : \vx' \in \AR(\vx) \middle| \vc = \vc' \right].
\]
}
In other words, $\CFMR$ is the conditional probability that $\vx'$ belongs to $\AR(\vx)$ conditioned on the event that their sketch values $\vc = \vx - g_{\Lattice}(\vx)$ and $\vc' = \vx' - g_{\Lattice}(\vx')$ are identical.

Our argument is based on the following two facts:
\begin{enumerate}
\item If $\CFMR \lessapprox 2^{-(2 \secpar + \omega(\log \secpar))}$, then the standard $\DL$ assumption implies the $\DLsketch$ assumption;
\item If $\CFMR \approx 2^{-\secpar}$ (or even $2^{-\omega(\log \secpar)}$), then the $\DLsketch$ assumption holds in the generic group model \cite{EC:Shoup97}. 
\end{enumerate}

We give an explanation for each item.
Below, recall that for a joint distribution $(\mathcal{X}, \mathcal{C})$, the (average) \emph{conditional collision entropy} of $\mathcal{X}$ given $\mathcal{C}$ is defined by $\avgcollentropy(\mathcal{X}|\mathcal{C}) := - \log_2 \COL(\mathcal{X}|\mathcal{C})$, where
{\small
\begin{equation}
\COL(\mathcal{X}|\mathcal{C})
:= \Pr[(x,c), (x', c') \samp (\mathcal{X}, \mathcal{C}) : x = x' | c = c']. \label{eqn:CR}
\end{equation}
}
Here, $\COL(\mathcal{X}|\mathcal{C})$ is called the \emph{conditional collision probability of $\mathcal{X}$ given $\mathcal{C}$}.
When the context is clear, 
we often abuse notation and write $\avgcollentropy(x|c)$ instead of $\avgcollentropy(\mathcal{X}|\mathcal{C})$, and we do a similar treatment for $\COL$.

\noindent{\underline{\textit{(1) $\DL$ implies $\DLsketch$ when $\CFMR$ is sufficiently small.}}}
Identifying the joint distribution $(\mathcal{X}, \mathcal{C})$ in \cref{eqn:CR} with $\{\vx \samp \FuzDataDist : (\vx, \vc = \vx - g_{\Lattice}(\vx)) \}$, we clearly have $\COL(\vx|\vc) \le \CFMR$.
Moreover, observe $\COL(\vx|\vc) = \COL(\mB^{-1} \vy|\vc)$, since recovering $\vx$ given $\vc$ implies recovering $\mB^{-1} \vy$ given $\vc$ and vice versa due to $\vc = \vx - \vy$. 
Now, suppose we had an upper bound of $\COL(\mB^{-1}\vy|\vc) \le p^{-1} \cdot 2^{-\omega(\log \secpar)}$, or equivalently $\avgcollentropy(\mB^{-1}\vy|\vc) \ge \log_2 p + \omega(\log \secpar)$, when we sample $\vx \samp \FuzDataDist$ and calculate $(\vc, a) \lruns \Sketch(\LinSpp, \vx)$, where $\vc = \vx - \vy = \vx - g_{\Lattice}(\vx)$.
Then, the leftover hash lemma of \cite{SICOMP:DORS08}, formally recalled in \cref{app_sec:lhl}, guarantees that the proxy key $a = \linhash(\mB^{-1} \vy) \in \ZZ_p$ is statistically close to a uniformly random element even given $\vc$, and thus the standard $\DL$ assumption implies the $\DLsketch$ assumption.

Putting things together, if $\CFMR \lessapprox 2^{-(2 \secpar + \omega(\log \secpar))}$, then  standard $\DL$ implies the $\DLsketch$ assumption since $p \approx 2\secpar$.
However, since typically $\secpar \ge 80$, this condition on $\CFMR$ may be somewhat too expensive to assume for fuzzy biometrics. 
Nevertheless, we believe the above provides us an intuition that the $\DLsketch$ assumption is not an esoteric assumption and justifies that $\CFMR$ is the right quantity to care about.

\noindent{\underline{\textit{(2) $\DLsketch$ is hard in the generic group model.}}
The generic group model \cite{EC:Shoup97} is an idealized model of computation for a cyclic group, where algorithms do not use the representation (or, the encoding) of the group elements, other than testing the equality of group elements.
When a new computational problem related to a cyclic group is introduced, this model is typically used to reason about its hardness.
Specifically, if some computational problem is proved to be hard for PPT adversaries in the generic group model, then it formally guarantees that one cannot solve the problem efficiently as long as one is performing only group operations.
To solve it efficiently, one must rely on a weakness of a particular group.
Thus, the hardness of a computational problem in the generic group model serves as a strong evidence that if we use a cyclic group where no such weakness is known (e.g. a group over an elliptic curve).%

Based on existing works, we can observe that if $\CFMR \approx 2^{-\secpar}$ (or even $2^{-\omega(\log \secpar)} = \secpar^{-\omega(1)}$), then there is no PPT algorithm that can break the $\DLsketch$ assumption in the generic group model.
Specifically, it is a well-known fact (and formally shown in \cite[Lemma 6]{TCC:DodYu13}) that a universal hash family is a good \lq\lq strong randomness condenser'' and preserves essentially all the (conditional) collision entropy of the input $\mB^{-1}\vy$ of $\linhash$ to its output $a = \linhash(\mB^{-1}\vy)$.
That is, we have $\COL(a|\linhash,\vc) \approx \COL(\mB^{-1}\vy|\vc) \le \CFMR$.
Moreover, \cite{C:BitCan10} considers a stronger variant of the decisional Diffie-Hellman problem where the exponents are not uniformly distributed but of superlogarithmic min-entropy $\omega(\log \secpar)$, and showed that this problem is hard for any PPT adversary in the generic group model.
This directly implies the hardness of the $\DL$ problem in which the exponent of a problem instance is chosen from a distribution with min-entropy $\omega(\log \secpar)$ in the generic group model.
Finally, min-entropy and collision entropy are linearly related.
Hence, taking average over the choice of $\vc = \vx - \vy$, we can conclude that the $\DLsketch$ problem is hard to solve for any PPT adversary in the generic group model if $\CFMR = 2^{-\omega(\log \secpar)}$.


\section{Experimental Method For Estimating Biometric Entropy} \label{sec:concrete_fks}
The final and most important content to cover is the question of whether we can use real-world biometrics to realize efficient and secure fuzzy signatures. 
For instance, it is clear that if everybody had similar biometrics, then there is no way to realize a secure fuzzy signature since everybody can impersonate each other. 
However, in reality, everyone possesses different biometrics. Therefore, if we used all of our personal biometrics, then fuzzy signatures should intuitively be secure since it would be extremely hard to impersonate someone. 
But obviously, collecting vast amount of biometrics from a user will make the signing procedure very expensive and drastically decrease user experience. 
The question is then, \emph{how much entropy does a specific biometric have, and can it be used to securely and efficiently instantiate fuzzy signatures?} 

This section can be divided into two parts: we first provide an easy-to-state sufficient condition for ``fuzzy signature compatible'' biometrics, and then we establish an experimental method to show that a given biometric satisfies this condition.

\subsection{Preprocessing Biometrics} \label{app:sec:preprocess_fuzzy_biometric} 
Before getting into the main content of this section, we first clarify how fuzzy biometrics are handled in more detail. 
That is, given, say a raw image of a fingerprint, what is the corresponding fuzzy biometric $x$ that we have been abstractly using throughout the paper. 
As with any real-world data, we preprocess (e.g., conduct feature extraction on) raw biometric data obtained via some measurement and represent them in a meaningful way. 
This preprocessed data is in fact what we have been calling ``fuzzy biometrics $x$'' throughout the paper. 
A pictorial explanation is provided in the bottom of \cref{fig:preprocess_biometrics}. 
The method of preprocessing raw biometric data depends on the concrete type of biometrics being used. 
We provide a concrete example in \cref{sec:experiment}, where we conduct experiments using real-world finger-veins. 
In the following, when we mention fuzzy biometrics, we always assume the preprocessed version. 
Moreover, the distribution $\FuzDataDist$ of fuzzy biometrics is the distribution induced by preprocessing a randomly sampled raw biometrics.  
\begin{figure} [ht]
  \centering
	\includegraphics[ width=0.3\textwidth ]{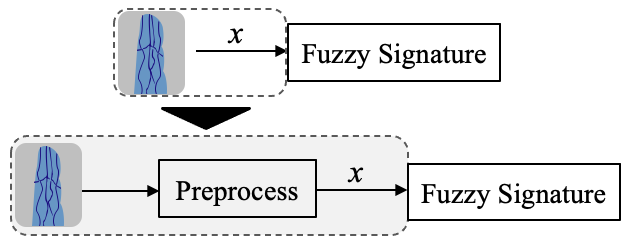}
  \caption{\small The above (resp. bottom) depicts a simplified (resp. realistic) version of how we handle biometrics. 
  } 
   \label{fig:preprocess_biometrics}
 \end{figure}

\subsection{Preparation} \label{sec:preparation}

\noindent{\underline{\textit{What is required from fuzzy biometrics?}}}
As we have seen in \cref{sec:security_dl,sec:analysis_dlsketch}, the concrete values of $\FNMR$ and $\CFMR$ of the fuzzy biometrics dictate the applicability to fuzzy signatures.
Recall the former and latter correspond to the correctness and security of fuzzy signatures, respectively: 

\begin{itemize}
	\item[False Non-Matching Rate (\FNMR):] 
	Informally, this was the probability that two honestly generated fuzzy biometrics $x$ and $x'$ from the same user are ``far''. Formally, $\FNMR := \Pr[  x \gets \FuzDataDist,  e\gets \ErrDist: x + e \not\in \AR(x)]$.
	
	\item[Conditional False Matching Rate (\CFMR):] Informally, this was the collision probability of fuzzy biometrics conditioned on the sketch being identical. Formally, $\CFMR := \Pr[ x, x' \gets \FuzDataDist,  (\sketch, a) \gets \Sketch(x), (\tsketch, \ta) \gets \Sketch(x'): x' \in \AR(x) | ~\sketch = \tsketch ]$, where recall $\Sketch$ is a deterministic algorithm. In particular, the probability is only over the randomness used to sample $x$ and $x'$.%
	\footnote{
	For simplicity, we omit the randomness of the public parameter $\LinSpp$. 
	}
\end{itemize}

Observe that the values of $\FNMR$ and $\CFMR$ are determined uniquely by the following factors: 
the distribution $\FuzDataDist$ of fuzzy biometrics, the definitions of the linear sketch, and the acceptance region $\AR$ used by the linear sketch. 
Furthermore, observe that $\FuzDataDist$ is implicitly defined by the concrete type of biometrics being used, and the linear sketch only depends on the definition of $\AR$ (or equivalently to the lattice as explained in \cref{sec:linear_sketch}). 
Therefore, $\AR$ is  the only parametric term that we can experimentally tune that would affect the values of $\FNMR$ and $\CFMR$. 
Namely, the choice of $\AR$, which roughly is a metric for deciding whether two fuzzy biometrics $x$ and $x'$  are ``similar", is the main term that determines $\FNMR$ and $\CFMR$.
As a rule of thumb, we like to define $\AR$ to be efficiently computable and to reflect the actual closeness metric of the underlying fuzzy biometric.
For instance, if the closeness is measured by the Euclidean metric, the hexagon $\AR$ may be better than the square $\AR$ as in \cref{fig:acceptance_region}. (See also \cref{sec:triangular_lattice}). 

\begin{figure} [ht]
  \centering
	\includegraphics[ width=0.2\textwidth ]{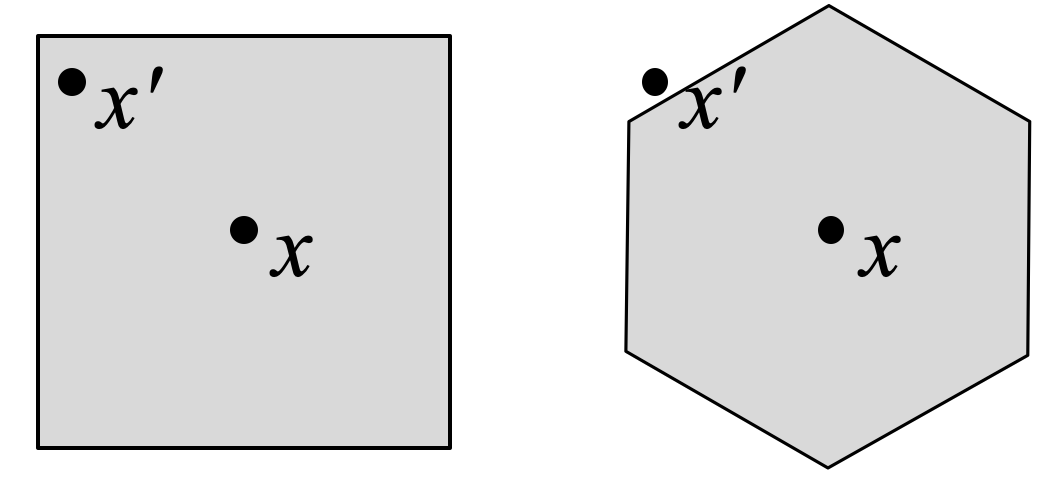}
  \caption{\small The gray area depicts the acceptance region $\AR(x)$ of some fuzzy biometric $x$. Although $x$ and $x'$ are the same, they may be considered to be close (left) or far (right) depending on $\AR$. } 
   \label{fig:acceptance_region}
 \end{figure}

In real-world applications of fuzzy signature, we can typically tolerate correctness error of at most $5\%$ and security level of at least $112$-bits.
We can tolerate the correctness to be much larger than the security level since we can simply retry till signing succeeds. 

To summarize thus far, once we fix a (set of) biometrics, e.g., iris, fingerprint, and finger-vein, used by the fuzzy signature and a description of the linear sketch scheme, the remaining issue is to define an appropriate acceptance region $\AR$ and show that the fuzzy biometric provides $\FNMR \lessapprox 5\% ( \approx 2^{-4.32})$ and $\CFMR \lessapprox 2^{-112}$.%
\footnote{
Note that we need the additional condition that a linear sketch scheme with respect to $\AR$ is efficiently constructible. We intentionally keep this requirement implicit to make the presentation simple. 
}
In the following, we show how to experimentally estimate the values of $\FNMR$ and $\CFMR$ for a given definition of $\AR$. 

\noindent\underline{{\textit{What kind of fuzzy biometrics is required for the experiment?}}}
For the experiments, we assume a natural type of biometric dataset to be provided: $S = \set{ x^{(i)}_j }_{(i, j) \in [N] \times [0: \ell]}$, where $x^{(i)}_j$ is the $j$-th fuzzy biometric of the $i$-th user.%
\footnote{
For a non-abstract treatment of fuzzy biometric, see \cref{sec:efficiency_analysis,app:sec:preprocess_fuzzy_biometric}. 
}
That is, $S$ contains $(\ell + 1)$ fuzzy biometrics from $N$ users. 
Such a dataset can be collected in practice by scanning each user $i$'s biometrics $(\ell + 1)$-times. 
Looking ahead, $x^{(i)}_0$ is a special biometric scanned at the enrollment phase (i.e., generation of the verification key) and $\set{ x^{(i)}_j }_{j \in [\ell]}$ are biometrics scanned during signing. 
Finally, denote $\bS :=   \set{ x^{(i)}_j }_{(i, j) \in [N] \times [\ell]}$ and $\bS^{(i)} :=  \set{ x^{(i)}_j }_{j \in [\ell]}$. 

Note that we can always define $\AR$ such that $x^{(i)}_j \in \AR(x^{(i)}_0)$ for all $j \in [ \ell ]$ and $x^{(i')}_{j'} \notin \AR(x^{(i)}_0)$ for all $i' \neq i$ and $j \in [\ell]$, i.e., a perfect definition of $\AR$ for the specific dataset $S$. 
However, it is clear that such an $\AR$ is overfitting to the particular dataset $S$ and will not generalize to unseen fuzzy biometrics $\FuzDataSpace$. 
Moreover, since typically, such an $\AR$ cannot be computed efficiently we will not be able to efficiently construct an associating linear sketch scheme or perform the experiments explained below. 
Therefore, in practice, we use natural definitions of $\AR$ as those explained in \cref{sec:linear_sketch}. 

\subsection{Estimating \FNMR of Biometrics}\label{sec:fnmr}
We first estimate $\FNMR$ by $\tFNMR$. 
Given a dataset $S$ of the above type, we empirically calculate \tFNMR as follows: 
\begin{align}
	\tFNMR :=  \frac{ \sum_{i \in [N] } \abs{\bigset{ x \in \bS^{(i)} \mid x \not \in \AR(x^{(i)}_0) }}}{ \sum_{i \in [N]} \abs{\bS^{(i)}}  },  \label{eq:estimate_fnmr}
\end{align}
where we assume $\AR$ is efficiently computable. 
It is easy to see that the numerator counts all the fuzzy biometric of each user that does not lie inside the acceptance region $\AR(x_0^{(i)})$. 

\subsection{Estimating \CFMR of Biometrics}\label{sec:cfmr}
\noindent{\underline{\textit{Difficulty of estimation.}}}
We next estimate $\CFMR$ by $\tCFMR$. 
Computing \tCFMR experimentally turns out to be much harder compared to computing \tFNMR. 
The main reason is the value of \CFMR that we wish to evaluate is much smaller than \FNMR; while we only needed to show that \FNMR is smaller than $5\%$, we need to show that \CFMR is smaller than $2^{-112}$ to be cryptographically useful. 
In fact, even if we waived the condition $c = c'$, it is still non-trivial to estimate \CFMR, which is by definition \FMR, since the event we are trying to check happens with probability only $2^{-112}$. 
According to the rule of three \cite{JAMA:HanLip83}, more than $3 \cdot 2^{112}$ independent impostor biometrics (i.e., pairs of $x, x'$ from \emph{different} users such that $x' \in \AR(x)$) are required in the dataset $S$  to conclude that $\FMR$ is smaller than $2^{-112}$ with $95\%$ confidence.
However, collecting such $S$ is highly impractical.%
\footnote{
Looking ahead, in our experiment, we consider settings where we only need to show $2^{-28}$ since we use 4 independent biometrics. However, this is still difficult to collect in practice. }
}
This is in sharp contrast to \FNMR where we only needed to assume that the dataset $S$ contains more than $3 \cdot (1/5\%) \approx 3\cdot 2^{4.32}$ pairs of biometrics $x, x'$ from the \emph{same} user such that $x' \in \AR(x)$ to get a meaningful estimate. 
We note that estimating \FMR, let alone \CFMR, is generally a difficult problem in biometrics due to the difficulty in collecting sufficient data, e.g., \cite{PR:Daugman03,TCSV:Daugman04}.

\noindent{\underline{\textit{Our approach.}}}
We divide the problem of estimating \CFMR into two subproblems as follows: 
\begin{enumerate}
	\item First, evaluate $\tFMR$. Namely, ignore the condition $c = c'$ on the sketch in \CFMR and simply estimate $\FMR$, where $\FMR := \Pr[ x, x' \gets \FuzDataDist: x' \in \AR(x) ]$.   \label{item:subproblem1}
	\item Then, show that \FMR and the value of sketch are uncorrelated. Namely, experimentally show that $\CFMR$ can be approximated by \FMR.  \label{item:subproblem2} 
\end{enumerate}
By individually solving the two subproblems, we eventually estimate the value $\tCFMR$ by $\tFMR$. 
The details of the solution to the individual subproblems follow.

\noindent{{\textbf{Subproblem~\cref{item:subproblem1}.}}}
As mentioned before, the value \FMR is typically too small to perform a simple estimation as we did for \FNMR. 
To overcome this issue, we borrow techniques from \emph{extreme value analysis} (EVA), a statistical method for evaluating very rare events by using only an ``extreme'' subset of the dataset $S$ \cite{BOOK:CBTD01,NIST:Michael12}.

We explain how to estimate \FMR using EVA below. 
First, define a continuous function called \emph{scaled} acceptance region $\sAR(w, x)$ defined for all $x \in \FuzDataSpace$ and $w > 0$ such that $\sAR(1, x) := \AR(x)$ and $\sAR(w, x)$ is an isotropic scaling of the original set $\AR(x)$ by a factor~$w$. 
A pictorial example is provided in \cref{fig:AR_EVS}. 
Notice that although $\sAR(1, x)$ does not include many points from the dataset $\bS$, we can increase them by enlarging $w$ and considering a larger set $\sAR(w, x)$. 
Also, define the scaled false matching rate function $\sFMR(w) := \Pr[x, x' \gets \FuzDataDist: x' \in \sAR(w, x)]$, where we have $\sFMR(1) := \FMR $ by definition. 
In the following, we estimate the probability distribution \sFMR(w), denoted as $\tsFMR(w)$, and then indirectly estimate the desired $\tFMR$ by plugging in $w = 1$ into $\tsFMR(w)$. 
Note that this is different from how we were able to directly estimate \FNMR through the dataset $S$.

 \begin{SCfigure}[2.3][ht]
    \includegraphics[ width=0.15\textwidth ]{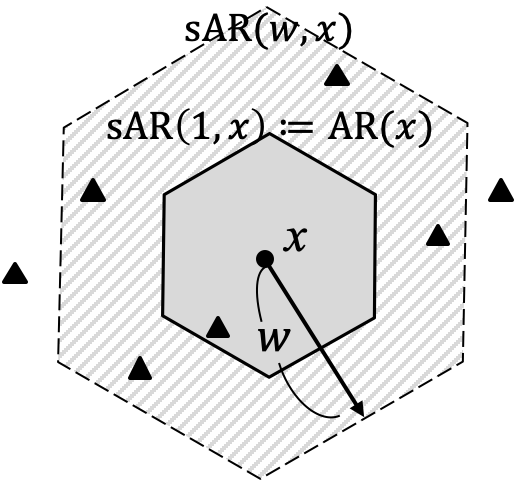}
        \caption{\protect\rule{0ex}{5ex} \small The bold gray area $\sAR(1, x)$ is the original set $\AR(x)$. The triangles are the fuzzy biometrics in $\bS$ that are different from $x$. The shaded area $\sAR(w, x)$ is the set $\AR(x)$ scaled by a factor~$w$. }
   \label{fig:AR_EVS}
\end{SCfigure}

The core of EVA is how to estimate a probability distribution $\sFMR(w)$ in the extremely rare setting $w \approx 1$. The high level idea is as follows. 
We first hypothesize that $\sFMR(w)$ can be explained by a particular class of natural probability density function when $w$ is smaller than some appropriately chosen $w^*$.
For instance, in our case, the class we consider is the family of power distributions $ \mathcal{F}_{\sf pow} = \set{ aw^b}_{a, b >0}$.\footnote{
Note that the class $\mathcal{F}_{\sf pow}$ is not itself a probability density function. We only assume that the probability density function of $\sFMR(w)$ for the narrow range $w \in [0, w^*]$ can be explained by  $\mathcal{F}_{\sf pow}$. 
}
So as not to interrupt the explanation of EVA, we provide rational behind the choice in \cref{rem:choice}. 
Now, since $\sFMR(w^*) = \int_0^{w^*} aw^b dw$ for some unknown values of $a$ and $b$, we first estimate $\tsFMR(w^*)$ from the dataset $S$, and then further estimate $a$ and $b$ via the maximal likelihood analysis (MLA) \cite{BOOK:Bishop06}. 
Here, notice we can properly estimate $\sFMR(w^*)$ from the dataset $S$ for an appropriate value of $w^*$ since enough points in the dataset $S$ will lie in the region $\sAR(w^*, x)$ for large enough $w^*$ (see \cref{fig:AR_EVS}). 
Finally, once $a$ and $b$ are computed via the MLA, we obtain $\tsFMR(1)$ by computing $\int_0^{1} a w^b dw$. 
Note that the main idea behind EVA is to only use an appropriately chosen small $w^*$ so that we can focus on estimating the range where $\sFMR(w)$ has extremely small values, rather than estimating the entire function $\sFMR(w)$. 
Specifically, if we use a too large $w^*$, we may be able to estimate $\sFMR(w)$ well in its entirety, however, it will not produce good estimates when conditioning on $\sFMR(w)$ with small values. 
The appropriate choice of $w^*$ is dataset dependent and we discuss this in \cref{rem:choice}.

We now provide a more formal description of the above procedure. 
First, consider the function $k(w)$ defined as 
\begin{align*}
k(w) =  \sum_{i \in  [N]}  \abs{ \set{ x \in \bS \backslash \bS^{(i)} \mid x \in \sAR(w, x^{(i)}_0) }}. 
\end{align*}
Since the dataset $S$ is discrete, we can efficiently compute a sequence of positive reals $w_1 <  w_2 < \cdots $, where each $w_n$ is the smallest $w$ satisfying $k(w_n) = n$. 
We then pick an appropriate $w^* \in \set{w_n}_n$ as explained in \cref{rem:choice}, and denote $k^* := k(w^*)$, that is, $w^* = w_{k^*}$. 
Also set $k_{\max} = \lim_{w \rightarrow \infty} k(w)$, where by definition $k_{\max}$ is the number of total impostor biometrics. 
We estimate the value of $\tsFMR(w^*)$ by $k^*/ k_{\max}$. 
Then, by the hypothesis that the probability density function $\sFMR(w)$ for small $w \le w^*$ follows $f(w) = a w^b$ for some positive reals $a$ and $b$, we have 
\begin{align*}
	\frac{k^*}{k_{\max}} = \tsFMR(w^*)  \approx \int^{w^*}_0 f(w) dw = \frac{a}{b + 1} w^{*(b + 1)}. 
\end{align*}
Solving the above for $a$ and plugging it into the likelihood function \cite{BOOK:Bishop06}, we obtain the following. 
{\small
\begin{align*}
	L(b) = \prod_{n = 1}^{k^*} f(w_n) = a^{k^*} \prod_{n = 1}^{k^*} w_n^b = \left( \frac{k^*(b + 1)}{k_{\max} w^{*(b + 1)}} \right)^{k^*} \cdot \prod_{n = 1}^{k^*} w_n^b. 
\end{align*}
}

Taking the logarithm of $L(b)$, we can show it is maximized when 
$b = k^*/(k^* \ln w^* - \sum_{n=1}^{k^*} \ln w_n) - 1$. 
Setting $r = b + 1$ and combining everything, we conclude that $\tsFMR(w) = \frac{k^*}{k_{\max} \cdot w^{*r}} w^r$. 
Finally, plugging in $w = 1$, the desired estimate for $\sFMR(1) = \FMR$ is 
\begin{align}
	\tsFMR(1) = \frac{k^*}{k_{\max} \cdot w^{*r}}.  \label{eq:estimate_cfmr}
\end{align}

\begin{remark}[Choice of $w^*$ and $\mathcal{F}_{\sf pow}$] \label{rem:choice}
When using EVA, the particular choice of $w^*$ is data specific, and we typically check whether the choice was reasonable by plotting the estimated function (see \cref{sec:experiment} for a concrete example). 
Noticing that $w^*$ and $k^*$ are in one-to-one relation, we can choose $k^*$ instead. 
The concrete choice of $k^*$ may be data specific but they are typically small values ranging from, say $0.1\%$ to $5\%$. 
Put differently, we use only 0.1\% to $5\%$ of $k_{\max}$ (i.e., the number of impostor biometrics in the dataset $S$) to estimate the extremely rare events. 
Moreover, we hypothesize that the probability density function of $\sFMR(w)$ for very small values of $w$ is contained in $\mathcal{F}_{\sf pow}$ by making a natural assumption that the \emph{local} probability distribution around a fuzzy biometrics $x$ is \emph{smooth}. 
That is, we assume that for any fuzzy biometrics $x$, any $x'$ in the vicinity of  $x$ occurs with equal probability. Let $g(x)$ be the distribution of $x$, i.e., $\FuzDataDist$. 
Then, when $w$ is small, for any $x$ we can approximate $\sFMR(w) = \int_x \int_{x' \in \sAR(w, x)} g(x') dx'dx \approx  \int_x g(x) \int_{x' \in \sAR(w, x)} dx' dx \propto  \int_x g(x) w^r dx = w^r$, where $r$ is the size of the dimension the fuzzy biometric lies in and we used the fact that $f(x) \approx f(x')$.%
\footnote{
We note that $r$ may be smaller than the concrete dimension $n$ of the fuzzy biometrics obtained through some feature extraction.
}
Finally, by taking the derivative of $w^r$, we see that the probability density function of $\sFMR(w)$ is included in $\mathcal{F}_{\sf pow}$. 
\end{remark}

\noindent{{\textbf{Subproblem~\cref{item:subproblem2}.}}}
As explained before, directly estimating \CFMR is difficult since the sketch being identical is an extremely rare event and no practical dataset $S$ will contain such samples.
Therefore, we instead provide an empirical evidence that \CFMR can be approximated by \FMR, and  indirectly estimate the value of $\CFMR$ by $\tsFMR(1)$ obtained above. 

Recall $\Sketch$ is a deterministic function. 
Let $q_c(x)$ be the function that ignores the proxy key $a$ and simply outputs the sketch $c$ of $(\sketch, a) \gets \Sketch(x)$. 
Then, we can rewrite $\CFMR = \Pr_{x, x' \gets \FuzDataDist}[x' \in \AR(x) \mid q_c(x) = q_c(x') ]$. 
Assume the space the sketch $\sketch$ lies in is endowed with the Euclidean metric (which holds true for all known linear sketch scheme). 
Then, for any $\ell \ge 0$, consider a variant of \CFMR defined as
\begin{align*}
	\ell\text{-}\CFMR := \Pr_{x, x' \gets \FuzDataDist}[x' \in \AR(x) \mid \dist{q_c(x), q_c(x')} = \ell ], 
\end{align*}
where $\dist{z, z'} := \normtwo{z - z'}$. 
If we can show that $\ell$ is uncorrelated to the value of $\ell\text{-}\CFMR$ for all $\ell \ge 0$, then we can ignore the sketch condition in $\CFMR$ and conclude that $\CFMR \approx \FMR = \Pr_{x, x' \gets \FuzDataDist}[x' \in \AR(x)]$. However, unfortunately, since the condition on $\ell\text{-}\CFMR$ is still a very rare event, which we cannot expect to have in our dataset, we still cannot empirically estimate $\ell\text{-}\CFMR$. 
To this end, we further relax the condition in $\ell\text{-}\CFMR$. 
For any large enough integer $M$, consider a sequence of reals $0  = \ell_0 <  \cdots < \ell_M$  such that $\Pr_{x, x' \gets \FuzDataDist}[ \dist{q_c(x), q_c(x'}) \in [\ell_{t - 1},  \ell_{t})] = 1/M$ for all $t \in [M]$. 
Then, for each $t \in [M]$, we consider the following alternative variant of $\CFMR$:
{\small
\begin{align*}
	\CFMR_t := \Pr_{x, x' \gets \FuzDataDist}[x' \in \AR(x) \mid \dist{q_c(x), q_c(x')} \in [\ell_{t - 1},  \ell_t) ]. 
\end{align*}
}
Due to how the way we partition the $\ell_t$'s, $\CFMR_i$ is an approximation of $\frac{(\ell_{t - 1} + \ell_t)}{2}\text{-}\CFMR$. 
Hence, our goal now is to show that for all $ t \in [M]$, the value of $\ell_t$ is uncorrelated to the value of $\CFMR_t$, which in particular approximately establishes that any $\ell$ is uncorrelated to $\ell\text{-}\CFMR$. 
Concretely, we will perform a hypothesis test using $t$-statistics on the pair $( \ell_t, \CFMR_t)$ to conclude that $\CFMR_t$ is not significantly correlated with $\ell_t$. 
We refer the standard explanation of statistical $t$-test to textbooks such as \cite{BOOK:GibCha14}.

To perform the statistical $t$-test, we first prepare the values of $\ell_t$ and $\CFMR_t$ for all $t \in [M]$. 
Since we cannot exactly compute them, we estimate them, denoted as $\tilde{\ell}_t$ and $\tCFMR_t$. 
Estimating $\tilde{\ell}_t$ is simple; 
we compute  $\dist{q_c(x^{(i)}_0), q_c(x')}$ for all $i \in [N]$ and $x' \in \bS \backslash \bS^{(i)}$ and sort them.
That is, we compute the distance of the sketches of all impostor pairs in $S$. 
Let the obtained distances be $L_1 \le \cdots \le L_{k_{\max}}$, where recall $k_{\max}$ was the number of total impostor pairs in $S$. 
Then, we set $\tilde{\ell_t} = L_{ \lfloor t \cdot k_{\max}/ M \rfloor}$ for $t \in [M]$. 
To estimate $\CFMR_t$, we use the same method used to estimate $\FMR$ to solve subproblem~\cref{item:subproblem1}. 
Namely, we use EVA to estimate $\CFMR_t$ by parameterizing the acceptance region $\AR$. 
The way the estimation proceeds is exactly the same as before except that we condition on the subset of the dataset $S$ so that the distance of the sketches are within $[\tilde{\ell}_{t - 1}, \tilde{\ell}_t)$.

Finally, after obtaining the samples $\set{ (\tilde{\ell}_t, \tCFMR_t) }_{t \in [M]}$, we perform the statistical $t$-test. 
Below, denote $y_t:= \log( \tCFMR_t )$ for $t \ \in [M]$. 
We perform a hypothesis test against the samples $\set{ (\tilde{\ell}_t, y_t) }_{t \in [M]}$, where the null  hypothesis $H_0$ is that the variables are uncorrelated. 
To this end, we first compute the sample correlation coefficient $r$ as 
{\small
\begin{align}
	r = \frac{ \sum_{t = 1}^M (\tilde{\ell}_t - \bar{\tilde{\ell}})(y_t - \bar{y}) }{  \sqrt{ \sum_{t = 1}^M (\tilde{\ell}_t - \bar{\tilde{\ell}})^2 } \cdot   \sqrt{ \sum_{t = 1}^M (y_t - \bar{y})^2 }}, \label{eq:correlation}
\end{align}
}
where $\bar{\tilde{\ell}}$ and $\bar{y}$ are the average of the samples. 
In case $\ell$ and $y$ are uncorrelated, then the value $t = \frac{r \sqrt{M - 2}}{\sqrt{1-r^2}}$ follows the $t$-distribution with $M- 2$ degree of freedom. 
Therefore, we compute the $p$-value from $t$ and conclude that the null hypothesis $H_0$ is not rejected at the 0.05 significance level if 
\begin{align}
	p \ge 0.05.  \label{eq:estimate_p}
\end{align}

Hence, if $p \ge 0.05$, then we conclude with high confidence that the value of $\ell\text{-}\CFMR$ is uncorrelated with the value of $\ell$, and in particular, approximate $\CFMR \approx \FMR \approx \tsFMR(1)$.


\section{Efficiency Analysis of Our Fuzzy Signature} \label{sec:efficiency_analysis}
In this section, we combine all the tools we developed thus far to show that we can instantiate fuzzy signatures efficiently and securely using real-world biometrics. 
In \cref{sec:experiment}, we first conduct experiments using the statistical methods presented in  \cref{sec:concrete_fks}  with real-world finger-vein biometrics, and conclude that finger-vein biometric from a single hand is sufficient for fuzzy signature. 
Then, in \cref{sec:efficiency_analysis_sub}, we instantiate our fuzzy signature with a concrete set of parameters and provide efficiency analysis of our proposed scheme. 

\subsection{Estimating Quality of Real-World Finger-Vein Biometrics}\label{sec:experiment}
We use real-world finger-veins (see \cref{fig:finger_vein}) to show that 4 finger-vein scans from a single hand is sufficient to instantiate fuzzy signature. 
To this end, we provide an appropriate definition of acceptance region $\AR$ and provide experimental results using the methods presented in \cref{sec:fnmr,sec:cfmr} to conclude $\FNMR \lessapprox 5\%$ and $\CFMR \lessapprox 2^{-112}$, respectively.

 \begin{SCfigure}[2.3][ht]
    \includegraphics[width=0.12\textwidth]{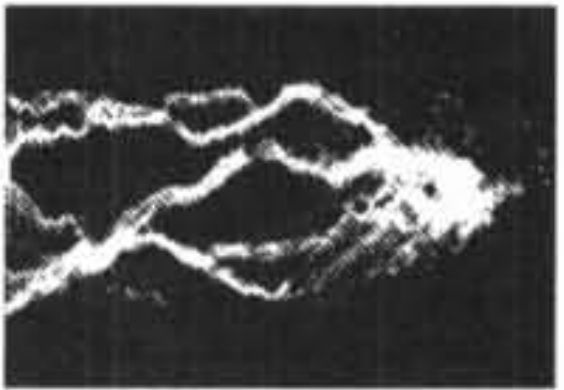}
        \caption{\protect\rule{0ex}{5ex} \small Example of an extracted finger-vein image. The image is taken from \cite[Figure 4]{MVA:MiuNagMiy02}. }
    \label{fig:finger_vein}
\end{SCfigure}

\noindent{\underline{\textit{Description of preprocessing and dataset $S$}}.}
There are several publicly available finger-vein datasets such as SDUMLA-HMT \cite{CCBR:YinLiuSun11} and Hong Kong Polytechnic University Finger Image Database Version 1.0 \cite{TIP:KumZho11}. 
However, in this work, we use the dataset used by \cite{BIC:YanAokOya09} as they contain the largest number of users and finger-vein images (roughly 3 to 5 times more). 

The finger-vein database contains 505 users where each user provided images of 6 fingers (index, middle, and ring fingers for both hands), and the collection for each finger was repeated 3 times to obtain 3 images (one for the enrollment phase and the other two for the signing phase). 
We eliminated 36 users that had finger-veins images that were not properly scanned. 
Moreover, since finger-veins of different fingers from the same user are believed to be independently distributed (a standard assumption used in many prior works \cite{ ICB:NanJaiRos09,ICBTAS:NanRosJai09,TIFS:TaoVel12,IEICE:MurKagTak16}), we can alternatively view the database as containing $(505 - 36) \times 6 = 2814$ users each providing 3 images of a single finger-vein.

We preprocess these raw finger-vein images into data compatible with our linear sketch (see \cref{sec:analysis_dlsketch}.)
We perform feature extraction on the finger-vein and represent it as an $n$-dimensional vector where we experimented with $n = 200, 300, 400$. 
For each $n$, we first randomly selected 1410 users from the 2814 users and performed principal component analysis \cite{BOOK:Bishop06} to extract the $n$-dimensional subspace that best explains the data. 
Then, we projected the remaining 1404 users' finger-vein onto the $n$-dimensional subspace and prepared the dataset $S = \set{ x^{(i)}_j }_{(i, j) \in [1404] \times [0: 2]}$. 
Each $n$-dimensional vector $x^{(i)}_j$ is represented as a 32-bit float. 

\noindent{\underline{\textit{Type of $\AR$}.}} 
We consider a regular hexagon as the acceptance region $\AR$ since we use the triangular lattice to instantiate the linear sketch scheme (see \cref{fig:AR_VR_triangular} and  \cref{sec:triangular_lattice}). 
Denote $d$ as the basis length of the triangular lattice (see \cref{app_sec:cvp_triangular_lattice}). 
Then, since the triangular lattice uniquely defines $\AR$, the value of $d$ indirectly parameterizes $\AR$; a  larger $d$ results in a larger region for $\AR$. 
Therefore, given the dataset $S$, we find the value $d$ that provides us with an $\AR$ that satisfies the conditions $\FNMR \lessapprox 5\%$ and $\CFMR \lessapprox 2^{-112}$. 
Note the concrete value of $d$ has no significant meaning as its length is relative to the scaling of the concrete fuzzy biometrics.

\noindent{\underline{\textit{Estimating \FNMR and \CFMR}.}} 
We use 4 finger-veins for the fuzzy signature%
\footnote{
Although our linear sketch is defined for a single biometric source, it is clear that they generalize to multiple independent biometric sources. For completeness,  details are provided in \cref{app_sec:composing_fks}. 
}
 and since each finger-vein is assumed to be distributed independently, we empirically evaluate whether the following holds for each $n = 200, 300, 400$: 
\begin{itemize}
	\item $\tFNMR \le 1 - (1 - 5\%)^{1/4} \approx 2^{-6.29}$ (see \cref{eq:estimate_fnmr})
	\item $\tCFMR \le (2^{-112})^{1/4} = 2^{-28}$ (see \cref{eq:estimate_cfmr})
	\item $p$-value is larger than 0.05 (see \cref{eq:estimate_p})
\end{itemize}
Here, the first item follows from the fact that we need \emph{all} 4 finger-veins to be correct to obtain a total of 5\% of false non-matching rate. 
Moreover, the last requirement is to check the validity of our estimation method in \cref{sec:cfmr}. Recall that if the $p$-value is larger than $0.05$, then we conclude that $\tCFMR$ can be estimated by $\tFMR$.

The following \cref{tab:experimental_result} summarizes our experimental result. 
For better readability we present the values of $\tFNMR$ and $\tCFMR$ where 4 finger-veins are simultaneously used, denoted as $\widetilde{\FNMR^4}$ and $\widetilde{\CFMR^4}$.   
For each dimension $n$, we varied the basis length $d$ (i.e., acceptance region $\AR$) to see its effect. 
For each dimension $n$, we chose three values for $d$ by targeting $\widetilde{\CFMR^4} = 2^{-80}, 2^{-128}$ and $\widetilde{\FNMR^4} = 5\%$, respectively. Although $112$-bits is the recommended security level for fuzzy signatures, we also benchmarked 80-bits of security since 80-bits may suffice in adversarially restricted scenarios, e.g., the system blocks the account after a few false attempts at signing. 
We also included the correlation coefficient $r$ of the $t$-test (see (\cref{eq:correlation})) to show that their absolute values are all below 0.2.

\begin{table}[htbp]
	\centering
	{\small
	\caption{ \small 
	$n$ denotes the dimension of fuzzy biometrics, $r$ is the correlation coefficient of the $t$-test, and $d$ denotes the basis length of the triangular lattice. } \label{tab:FNMR_and_CFMR}
	\begin{tabular}{|c |c | c| c |c | c|}
		\hline
		$n$ & $\widetilde{\FNMR^4} $ & $\widetilde{\CFMR^4}$ & $p$-value & $r$ & $d$   \\
		\hline \hline
		200 &  $ \textbf{2.4\%}$ & $ 2^{-80}$  & $0.35$ & $0.095$ & 43.4 \\\hline
		200 &  $ \textbf{5\%} $ & $2^{-106.6}$ & $0.27$ & $-0.111$&  39.6  \\ \hline
		200 &  $  9.7\% $ & ${\bf 2^{-128}}$  & $0.15$ & $0.146$ & 36.8 \\ \hline		
		300 &  $ \textbf{1.4\%} $ & $ 2^{-80}$  & $0.57$& $-0.057$& 44.6    \\\hline
		300 &  $ \textbf{5\%} $ & ${\bf 2^{-113.0}}$  & $0.89$ & $-0.015$ & 40.2  \\\hline
		300 &  $ 7.6\% $ & ${\bf 2^{-128}}$  & $0.78$ & $0.028$ & 38.4  \\\hline		
		400 &  $ \textbf{1.4\%}  $ & $ 2^{-80}$  & $0.50$ & $-0.068$  & 44.8  \\ \hline
		400 &  $ \textbf{5\%} $ & ${\bf 2^{-113.6}}$  & $0.65$ & $-0.046$  & 40.3   \\ \hline
		400 &  $ 8.0\% $ & ${\bf 2^{-128}}$  & $0.88$ & $-0.015$ & 38.6 \\
		\hline
	\end{tabular}
	}
	  \label{tab:experimental_result} 
\end{table}

The entries in bold-fonts in  \cref{tab:experimental_result} indicate those satisfying either $\widetilde{\CFMR^4} \le 2^{-112}$ or $\widetilde{\FNMR^4} \le 5\%$. 
When the dimension of the feature vectors of the finger-vein is $n = 300$ (resp. $n = 400$) and when the basis length is $d = 40.2$ (resp. $d = 40.3$), both conditions on $\widetilde{\CFMR^4} $ and $\widetilde{\FNMR^4}$ are satisfied. 
Therefore, our result indicates that 4 finger-veins are sufficient to provide the required properties to instantiate fuzzy signatures by taking those appropriate choices of $n$ and $d$. 
Since a larger dimension $n$ leads to a less efficient linear sketch scheme, taking $n = 300$ suffices. 
In addition, our experimental results also confirm the relationship between the size of $\AR$ and the tradeoff between $\widetilde{\CFMR^4} $ and $\widetilde{\FNMR^4}$. 
Observe that decreasing the size of $\AR$ (i.e., smaller $d$) has the effect of lowering $\widetilde{\CFMR^4}$ while increasing $\widetilde{\FNMR^4}$ as expected; a smaller $\AR$ makes it harder to impersonate while it also makes it more sensitive to measurement error. 
We note that although the false non-matching rate ($\widetilde{\FNMR^4}$) is typically set below 5~\% in practice, we  can tolerate a higher value of correctness error by allowing the signer to repeat until it succeeds. Therefore, in case we require a higher level of security such as 128-bits, then we can achieve this by increasing $\widetilde{\FNMR^4}$.

To see more closely the effect of varying the dimension $n$, we plot the  \emph{detection error tradeoff} (DET) curve \cite[Sec. 4.7]{ISO:06}. The case for $n = 200$ and $300$ is provided in \cref{fig:DET}. 
We refer the case for $n = 400$ to \cref{app:sec:DET} since it is similar to $n = 300$.
For each dimension $n$, the DET curve is plotted by varying the size of $d$ (hence $\AR$). 
Any region that lies above the DET curve is realizable.
For instance, since the coordinate indicating $(\widetilde{\FNMR^4}, \widetilde{\CFMR^4}) = (2^{-112}, 5 \%)$ (denoted as a red star in \cref{fig:DET}) is below the DET curve for $n = 200$, this means that there is no $d$ that satisfies the condition of $\widetilde{\FNMR^4}$ and $ \widetilde{\CFMR^4}$ when representing the finger-vein by a feature vector of only dimension $n = 200$. 
In contrast, for $n = 300$ and $400$, it can be checked that the there exists some choice of $d$ such that the conditions are satisfied since the red star is above their respective DET curves. 
\begin{figure}[htbp]
 \begin{minipage}{0.495\hsize}
  \begin{center}
	\includegraphics[ width=1.0\textwidth ]{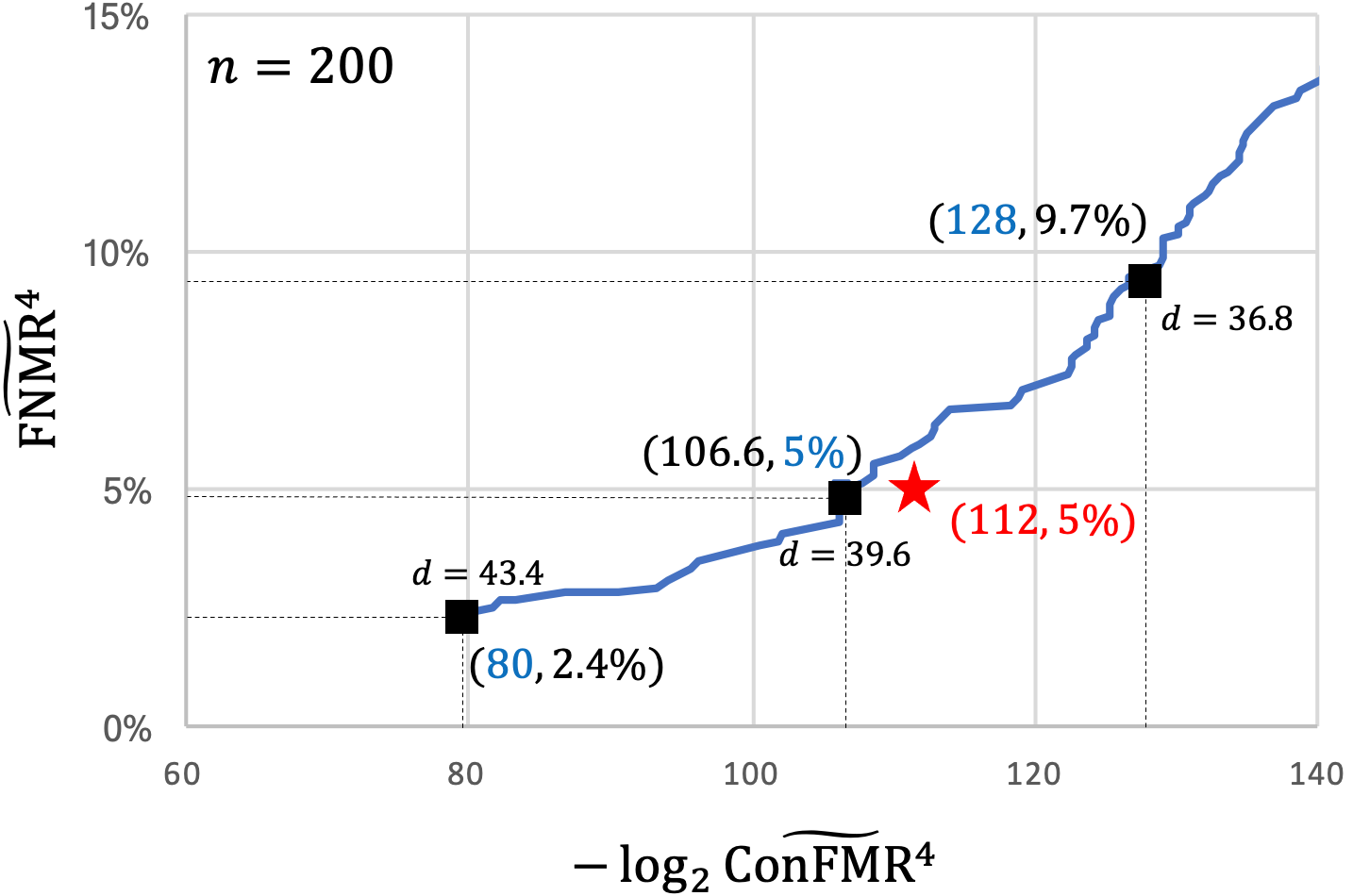}
  \end{center}
 \end{minipage}
 \begin{minipage}{0.495\hsize}
  \begin{center}
	\includegraphics[ width=1.0\textwidth ]{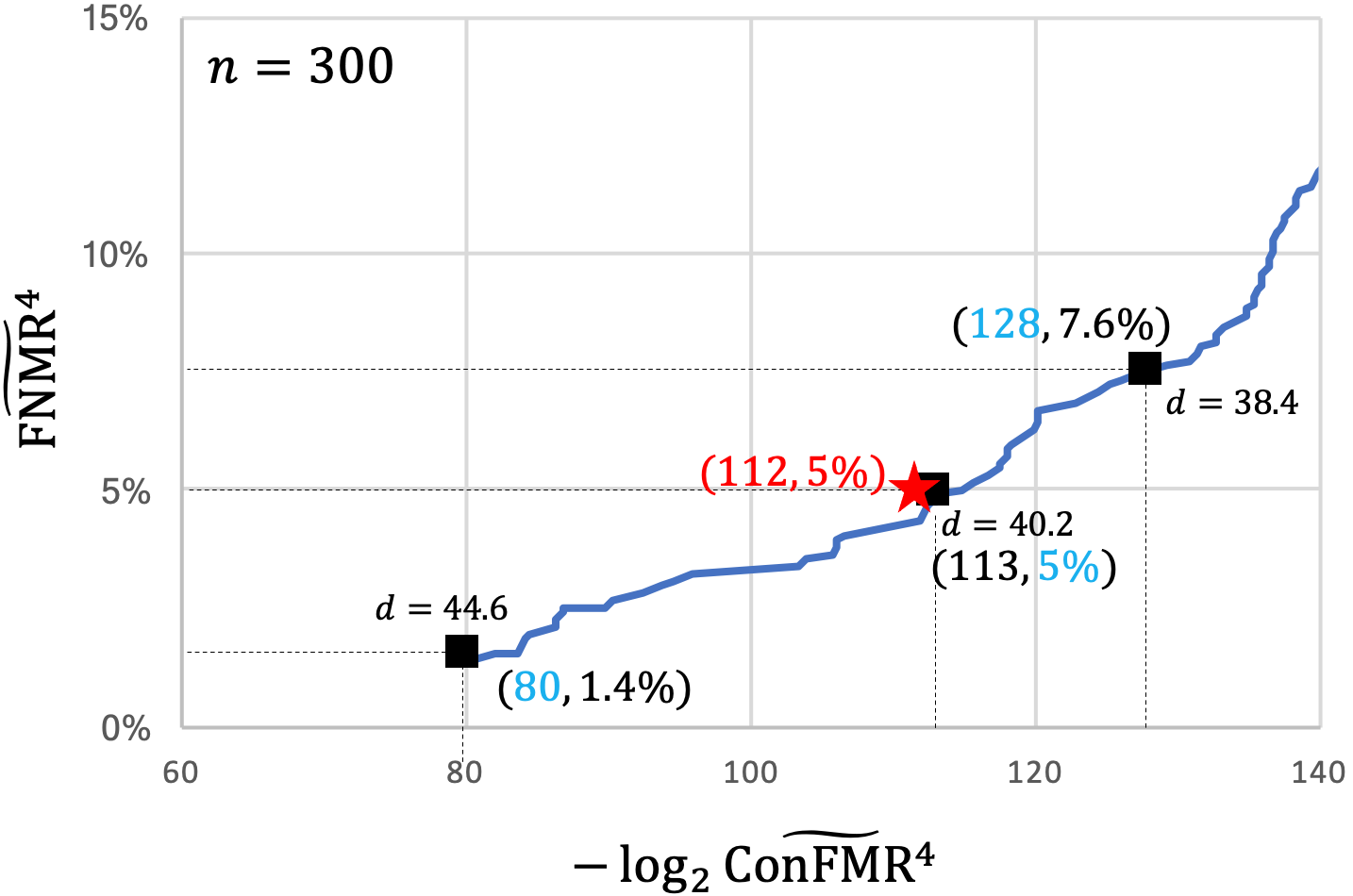}
  \end{center}
 \end{minipage}
 \caption{\small DET Curve for $n = 200$ (left) and $300$ (right).  }
 \label{fig:DET}
\end{figure}

Finally, we provide graphical evidence on the validity of our estimation of $\tFMR$ and $\tCFMR_t$ for $t \in [M]$; such method is one of the standard ways of assessing the quality of EVA \cite{BOOK:CBTD01,NIST:Michael12}. 
To perform EVA to estimate $\tFMR$, we set $k^* = 0.1\% \times k_{\max}$, where recall $k_{\max}$ is the number of total impostor pairs which is equal to $3,925,584$ for our dataset (see \cref{rem:choice}). 
We also set $M = 100$ to define the $M$-variant of $\set{\CFMR_t}_{t \in [M]}$ and perform EVA to estimate each $\tCFMR_t$ by setting $k^* = 0.5 \% \times k_{\max}$. 
The following \cref{fig:validitiy_EVA} illustrates the validity of our estimation for $\tFMR$ when the dimension $n = 300$ and $\tCFMR_t$ for $t = 50$. 
It can be visibly checked looking at the gray region that the estimation (in red line) aligns with the values of $w$ that we were able to measure with our database (in blue line). 
Hence, EVA allows us to conclude that the extremely small values that we were \emph{not} able to measure with our dataset can be approximated with our estimation. 
Additional experiments for other parameters are provided in \cref{app:sec:validity_EVA}. 

\begin{figure}[htbp]
 \begin{minipage}{0.495\hsize}
  \begin{center}
	\includegraphics[ width=1.0\textwidth ]{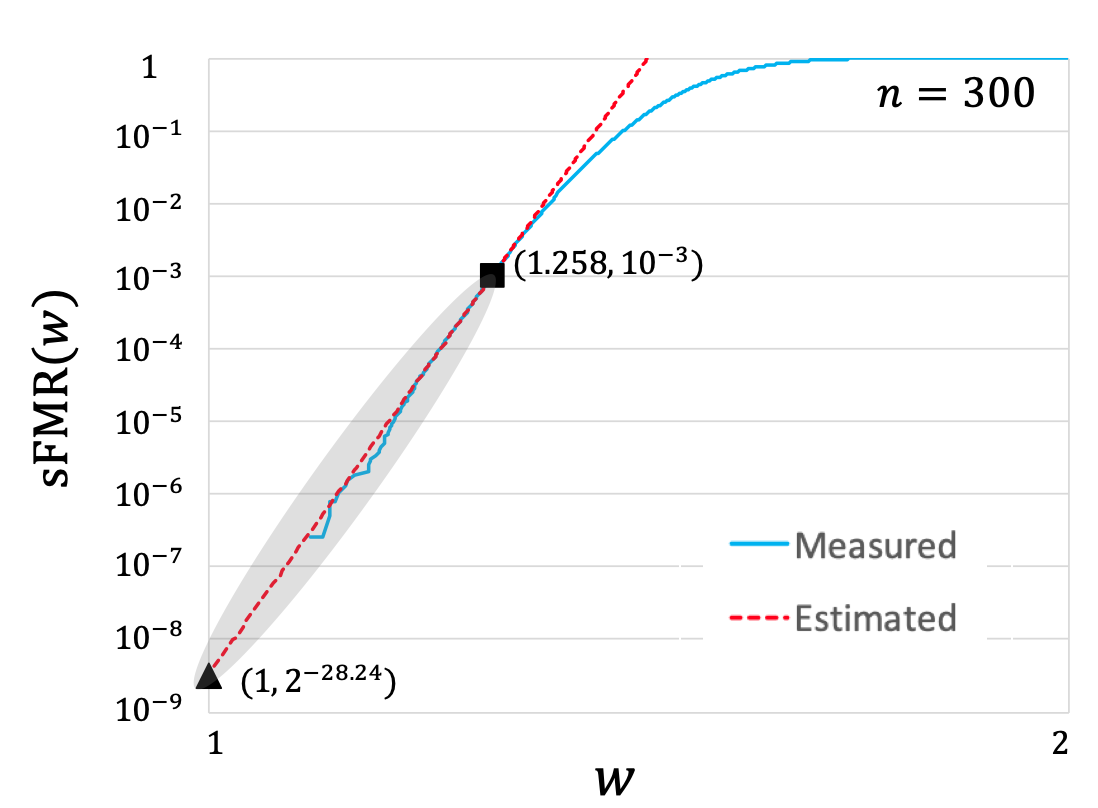}
  \end{center}
  \subcaption{$\sFMR(w)$ for $n = 300$}
  \label{fig:one}
 \end{minipage}
 \begin{minipage}{0.495\hsize}
  \begin{center}
	\includegraphics[ width=1.0\textwidth ]{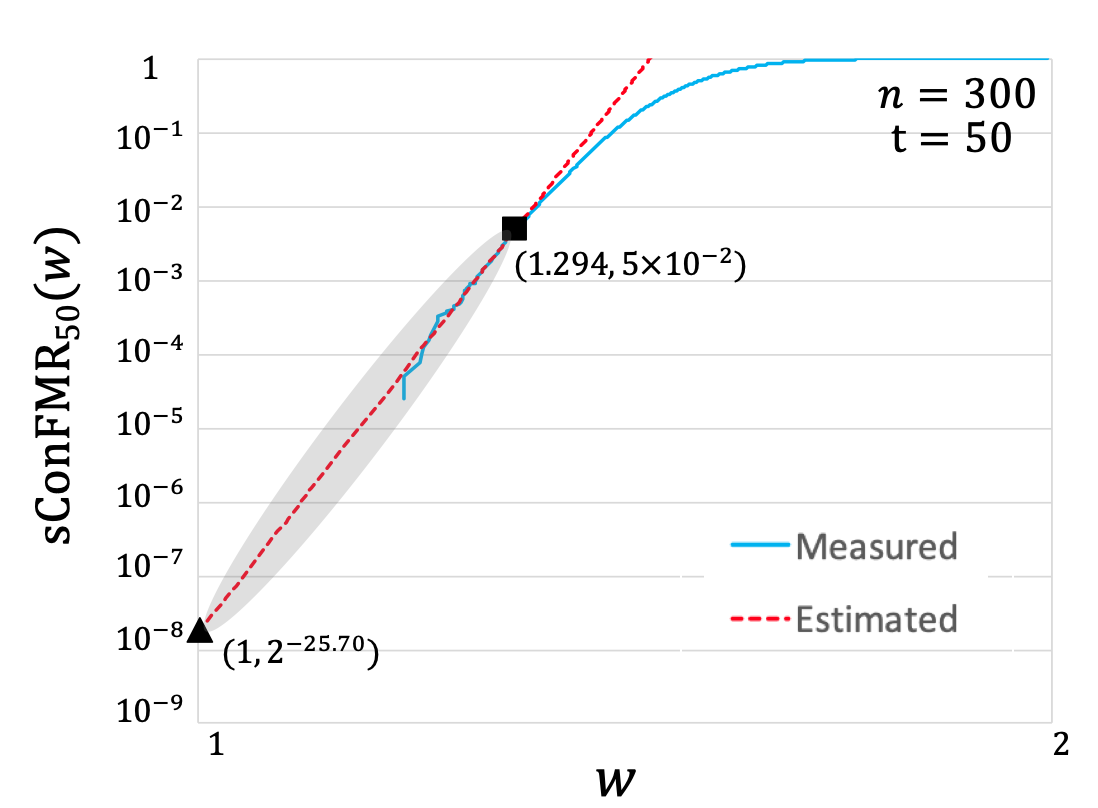}
  \end{center}
  \subcaption{$\sCFMR_{50}(w)$ for $n = 300$}
  \label{fig:two}
 \end{minipage}
 \caption{\small The blue line indicates the measured values of $\sFMR(w)$ and $\sCFMR_{50}(w)$ w.r.t to our dataset $S$. The red line indicates our estimation of the probability distribution of  $\sFMR(w)$ and $\sCFMR_{50}(w)$ via EVA. The gray region is the region for which EVA provides a reliable estimation. The square plots $(w^*, k^*)$ and the triangle plots $(1, X)$, where $X$ is the estimation for $\tFMR$ and $\tCFMR_{50}$.  }
 \label{fig:validitiy_EVA}
\end{figure}

\subsection{Efficiency Analysis of Our Fuzzy Signature} \label{sec:efficiency_analysis_sub}
We finish with a concrete analysis of our fuzzy signature scheme. 
We consider the 112 and 128-bit security levels (i.e., $\CFMR \le 2^{-112}$) and use the settings in \cref{tab:FNMR_and_CFMR} for $n = 300$ to define the fuzzy key setting. 
Recall from \cref{fig:fs_from_dl} that the verification key consists of one group element in $\GG$ and a sketch, and the signature consists of two elements in $\ZZ_p$ and a sketch. 
The running time is the sum of the individual runtime of the linear sketch and the Schnorr signature. 
Specifically, the only difference from the Schnorr signature is the linear sketch  component. 
\cref{tab:concrete_parameter} gives the concrete parameters. 
\begin{table}[htbp]
	\centering
	{\small
	\caption{ Benchmark for our fuzzy signature with $n = 300$.} \label{tab:concrete_parameter}
	\begin{tabular}{|c |c |c | c| c |c |}
		\hline 
		\multirow{2}{*}{\begin{tabular}{l}Sec.\\level\end{tabular}} & \multicolumn{2}{|c|}{Signature} & \multicolumn{2}{|c|} {Verification} & \multirow{2}{*}{\begin{tabular}{l}Correct\\ness err.\end{tabular}} \\ \cline{2-5}
		 & size (byte) & time (ms) & size (byte) & time (ms) & \\
		\hline\hline
		112 & 1256 & 0.50 & 1228 & 1.4 & $ 5.0\%$ \\
		128 & 1264 & 0.50 & 1232 & 1.4 & $7.6\%$ \\
		\hline
	\end{tabular}
	}
	  \label{tab:fuzzy_signature_parameters} 
\end{table} 
In more detail, the sketch has size $4n$ bytes in general, where $4$ bytes is used to represent each element by a 32-bit float. 
Plugging in $n = 300$, it can be checked that the size of the sketch dominates the signature and verification key size. 
The run time of $\Sketch$ and $\DiffRec$ are $0.45$ms and $1.3$ ms for both security levels%
\footnote{The only step dependent on the security parameter in our linear sketch scheme is the field size in $\linhash$, but its computation takes time that is at least two orders of magnitude smaller than computing $g_{\triLattice}$ or $\CV_{\triLattice}$, so its effect on run time is negligible.},
run on a machine with Intel(R) Core(TM) i7-8700K CPU at 3.70GHz.
Here, the universal hash $\linhash$ used within our linear sketch scheme (see \cref{fig:linear_sketch_construction}) simply computes the inner-product with a random $n$-dimensional vector over a prime field defined by the secret key space $\ZZ_p$ of the Schnorr signature scheme.
We also implement the Schnorr signature at the 112 and 128-bit security levels using elliptic curves with 224 and 256-bit primes, respectively, run on a machine with Intel(R) Core(TM) i7-1065G7 CPU at 1.30GHz.
For both security levels, the run times for signing and verification are at most several tens of microseconds, thus at least an order of magnitude smaller than the time taken by the linear sketch scheme.

We note that we can lower the sketch size by a factor of $2$ by representing the fuzzy biometrics by 16 bits rather than 32 bits.
In this case, the signature size will roughly be twice as small.
Here, treating less number of significant digits for the sketch value may affect the correctness (i.e., $\FNMR$) of the scheme, but not its security as formally discussed in \cite[Section 8]{IJIS:TMMHN19}.

\smallskip
\noindent
\textbf{Acknowledgement.} 
A part of this work was supported by JST CREST Grant Number JPMJCR19F6.




\appendix


\section{Leftover Hash Lemma} \label{app_sec:lhl}

We recall the leftover hash lemma of \cite{SICOMP:DORS08}.
To see the connection with the explanation in \cref{sec:analysis_dlsketch}, we state it using conditional collision probability.

Recall that for a joint distribution $(\mathcal{X}, \mathcal{Y})$, the (average) conditional collision probability of $\mathcal{X}$ given $\mathcal{Y}$ is defined by $\COL(\mathcal{X}|\mathcal{Y}) = \Pr_{(x,y), (x', y') \samp (\mathcal{X}, \mathcal{Y})}[x = x' | y = y']$.

Recall also that the statistical distance between two distributions $\mathcal{X}$ and $\mathcal{Y}$ is defined by $\SD(\mathcal{X}, \mathcal{Y}) := \frac{1}{2} \sum_{z}\abs{\Pr[\mathcal{X} = z] - \Pr[\mathcal{Y} = z]}$.
It is known that $\SD(\mathcal{X}, \mathcal{Y})$ upper-bounds the best (computationally unbounded) adversary's advantage in distinguishing the distribution using a single sample.

\begin{lemma}[Slightly adapted from \cite{SICOMP:DORS08}]
Let $\linHashFamily = \{\linhash: D \to R\}$ be a family of universal hash functions.
Let $\mathcal{X}$ and $\mathcal{Y}$ be distributions such that $(\mathcal{X}, \mathcal{Y})$ forms a joint distribution, and the support of $\mathcal{X}$ is contained in $D$.
Then, the statistical distance of the following two distributions is at most $\frac{1}{2} \sqrt{|R| \cdot \COL(\mathcal{X}|\mathcal{Y})}$:
\begin{align*}
&\bigl\{ \linhash \samp \linHashFamily;~ (x, y) \samp (\mathcal{X}, \mathcal{Y}): (\linhash, \linhash(x), y) \bigr\},\\
&\bigl\{ \linhash \samp \linHashFamily;~ (x, y) \samp (\mathcal{X}, \mathcal{Y});~ r \samp R: (\linhash, r, y) \bigr\}.
\end{align*}
In particular, if $\COL(\mathcal{X}|\mathcal{Y}) \le |R|^{-1} \cdot 2^{-\omega(\log \secpar)}$, then the statistical distance is $\negl(\secpar)$.
\end{lemma}

Strictly speaking, \cite{SICOMP:DORS08} showed the above lemma using the (average) conditional min-entropy (rather than conditional collision entropy/probabiltiy).
However, the above lemma can be easily inferred from the proof of \cite[Lemma~2.4]{SICOMP:DORS08} and the fact that the most basic form of the leftover hash lemma \cite{HILL99} (without taking into account the existence of $\mathcal{Y}$) works with collision entropy.

\section{Omitted Proof of Our Fuzzy Signature $\ProtFSig^\DL$}
\label{app_sec:fs_from_dl}

\subsection{Omitted Proof of Correctness: \cref{thm:dl_correctness}}
The complete proof of \cref{thm:dl_correctness} is provided below. 
It establishes the correctness of our fuzzy signature~$\ProtFSig^\DL$. 

\begin{proof}
Recall that by the definition of the fuzzy key setting $\FuzKeySet$, we have $\Pr[x \samp \FuzDataDist; e \samp \ErrDist : x + e \in \AR(x)] \ge 1 - \ErrFNMR$.
Hence, to show correctness, it is sufficient to show that if $x' \in \AR(x)$, then a signature generated using $x'$ is always accepted under a verification key generated using $x$.

Fix arbitrarily a message $\msg$ and fuzzy data $x, x' \in \FuzDataSpace$ such that $x' \in \AR(x)$.
Let $(\sketch, a) = \Sketch(\LinSpp, x)$ and $(\tsketch, \ta) = \Sketch(\LinSpp, x')$.
Also, let $\FSigpp = (\GGG, \LinSpp) \lruns \FSigSetup(1^\secpar, \FuzKeySet)$, $\FSigvk = (h = g^a, \sketch) \lruns \FSigKeyGen(\FSigpp, x)$, and $\FSigsig = (\beta, z, \tsketch) \lruns \FSigSign(\FSigpp, x', \msg)$, where $\beta = \hash(g^{\ta}, g^r, \msg)$ and $z = \beta \cdot \ta + r$.

Now, consider an execution of $\FSigVrfy(\FSigpp, \allowbreak \FSigvk, \allowbreak \msg, \allowbreak \FSigsig)$.
Since $x' \in \AR(x)$, the correctness of $\ProtLinS$ implies $\Delta a = \DiffRec(\LinSpp, \sketch, \tsketch) = \ta - a$.
Hence, $\FSigVrfy$ sets $\tildeh = h \cdot g^{\Delta a} = g^{a + (a' - a)} = g^{\ta}$ and $R = g^z \cdot \tildeh^{-\beta} = g^{\beta \cdot \ta + r} \cdot g^{- \ta \cdot \beta} = g^r$.
Hence, $\beta = \hash(g^{\ta}, g^r, \msg) = \hash(\tildeh, R, \msg)$ holds, and consequently $\FSigVrfy$ outputs $\top$, as desired.
\end{proof}

\subsection{Omitted Proof of Security: \cref{thm:dl_eucma}}
The complete proof of \cref{thm:dl_eucma} is provided below. 
It establishes the security of our fuzzy signature~$\ProtFSig^\DL$  under the $\DLsketch$ assumption. 

\begin{proof}[Proof Overview]
Before diving into the full proof, we provide an overview. 
The proof is similar to that of Schnorr signature \cite{C:Schnorr89}. 
The main difference is that in our proof, we additionally have to simulate the sketch $\sketch$ without knowledge of the secret fuzzy data $x$. 
To this end, we use the \emph{linearity} of the linear sketch (see \cref{def:linear_sketch}) that informally stipulates that given a sketch $\sketch$ where $(\sketch,a) \gets \Sketch(\LinSpp, x)$, there exists an algorithm $\Simsketch$ that simulates a fresh sketch $\tsketch$ for a proxy key $\tilde{a} $ with knowledge of only $(c, \Delta a := \tilde{a} - a)$. 
Specifically, since the $\DLsketch$ problem implicitly provides us with an ``initial'' sketch $\sketch$ of the proxy key (or secret exponent) $a$, we can easily construct an adversary $\B$ against the $\DLsketch$ problem that simulates the $\eucma$ security game to an adversary~$\A$ by running $\Simsketch$.  The full proof follows. 
\end{proof}

\begin{proof}
Let $\A$ be any PPT adversary against the fuzzy signature scheme $\ProtFSig^\DL$ that makes at most $Q$-signing queries and $Q_\hash$-random oracle queries and breaks the $\eucma$ security with probability $\epsilon$. 
Consider the following sequence of games, where the first game is equivalent to the original $\eucma$ game. Let $\Event_i$ denote the event that $\A$ wins in $\game_i$.

\noindent -\ $\game_1$: We define $\game_1$ as the actual game played between the challenger and the adversary $\A$. By assumption the winning probability of $\A$ in this game is $\Pr[\Event_1] = \epsilon$. In this game, the public parameter $\FSigpp$ and the verification key $\FSigvk$ are generated as follows: 
{\small
	\begin{align}
		\big[ &\Sigpp \lruns \SigSetup(1^\secpar);~ \LinSpp \lruns \LinSSetup(\FuzKeySet, \KeyDesc);~ \FSigpp \gets (\Sigpp, \LinSpp); \nonumber \\
		&~~  ~ x \lruns \FuzDataDist;~ (\sketch, a) \gets \Sketch(\LinSpp, x);~ \FSigvk = ( h \gets g^a, \sketch)  \big]. \label{eq:sim_keygen_dl} 
	\end{align}
}
	Furthermore, when $\A$ makes the $i$-th signing query (for $i \in [Q]$) on message $\msg_i$, the challenger generates a signature ${\FSigsig}_i$ as follows:
{\small
	\begin{align*}
		\big[ &e_i \samp \ErrDist;~ (\tsketch_i, \ta_i) \lruns \Sketch(\LinSpp, x + e_i);~ r_i \gets \ZZ_p;\\&~  \beta_i \gets \hash( g^{\ta_i}, g^{r_i}, \msg_i );~ z_i  \gets \beta_i \cdot \ta_i + r_i ;~  {\FSigsig}_i = (\beta_i, z_i, \tsketch_i) \big].
	\end{align*}
}
Throughout the proof, we call $g^{\ta_i} (= \tildeh_i )$ and $\ta_i$ as \emph{ephemeral} verification and signing keys, respectively, as it can be seen as an intermediate key used during the signing phase. 

\noindent -\ $\game_2$: In this game, we change how the signing queries are answered by the challenger.
Instead of using the ephemeral signing key $\ta_i$ to create the sketch $\tsketch_i$ as in the previous game, the challenger uses the auxiliary algorithm $\Simsketch$ of the linear sketch $\ProtLinS$ (\cref{def:linear_sketch}) with input $\sketch$ and $e_i$.
Specifically, when $\A$ makes the $i$-th signing query (for $i \in [Q]$) on message $\msg_i$, the challenger generates a signature ${\FSigsig}_i$ as follows: (where the difference from $\game_1$ is \underline{underlined}.)
{\small
	\begin{align*}
		\big[ &e_i \samp \ErrDist;~ \underline{(\tsketch_i, \shifta_i) \lruns \Simsketch(\LinSpp, \sketch, e_i);~ \ta_i \gets a + \shifta_i};~ r_i \gets \ZZ_p;\\
&~ \beta_i \gets \hash( g^{\ta_i}, g^{r_i}, \msg_i );~ z_i  \gets \beta_i \cdot \ta_i + r_i ;~ ;~  {\FSigsig}_i = (\beta_i, z_i, \tsketch_i) \big]. 
	\end{align*}
}
By the linearity of the linear sketch scheme $\ProtLinS$, the distribution of $\tuple{\tsketch_i}_{i \in [Q]}$ generated in $\game_1$ and $\game_2$ are identical.
Therefore, we have $\Pr[\Event_1] = \Pr[\Event_2]$.

\noindent -\ $\game_3$: In this game, we further modify how the signing queries are answered by the challenger. 
In the previous game, after $(g^{\ta_i}, g^{r_i}, \msg_i)$ were set, the challenger checked whether the random oracle $\hash$ was set on that point.
If not, it sampled a random $\beta_i \gets \ZZ_p$ and set the random oracle as $\hash(g^{\ta_i}, g^{r_i}, \msg_i) := \beta_i$.
Otherwise, it outputs the already programmed output. 
In this game, the challenger will abort the game when the input was already programmed. 
Since the random oracle is ever programmed on at most $(Q + Q_\hash)$ inputs, and $r_i$ is randomly sampled from $\ZZ_p$, the probability of an abort occurring on any of the signing query can be upper bounded by $Q \cdot (Q + Q_\hash) / p$. Hence, $\abs{ \Pr[\Event_2] - \Pr[\Event_3]} \le Q \cdot (Q + Q_\hash)  / p$. 

\noindent -\ $\game_4$: In this game, we make a final modification on how the signing queries are answered by the challenger. 
In particular, we alter the signing procedure so that the challenger no longer requires the secret key $a$ to sign; instead it will indirectly use the public key $h = g^a$. 
Conditioning on an abort not occurring, the challenger performs the following: (where the difference from $\game_3$ is \underline{underlined}.)
{\small
	\begin{align}
		\big[ e_i & \samp \ErrDist;~ (\tsketch_i, \shifta_i) \lruns \Simsketch(\LinSpp, \sketch, e_i);~  \beta_i \gets \bin^{2\secpar};~ \nonumber \\
	 &   \underline{z_i \gets \ZZ_p;~R_i \gets g^{z_i} \cdot (h \cdot g^{\shifta_i})^{- \beta_i};~ \hash( h \cdot g^{\shifta_i}, R_i, \msg_i ) := \beta_i};~\nonumber \\	
	 &\hspace {5.1cm} {\FSigsig}_i = (\beta_i, z_i, \tsketch_i) \big]. \label{eq:sim_sign_dl}
	\end{align}
}
The only difference from the previous game is the order of which $r_i$ and $z_i$ are constructed. 
In the previous game, a uniform random $r_i \gets \ZZ_p$ was sampled and then $z_i$ was set as $\beta_i \cdot \ta_i + r_i = \beta_i \cdot (a + \shifta_i) + r_i$. 
However, in this game, a uniform random $z_i \gets \ZZ_p$ is sampled and then $r_i$ is implicitly set to $z_i -  \beta_i \cdot (a + \shifta_i) $. We say \lq\lq implicitly'' since the challenger actually only computes $r_i$ in the exponent, that is, $R_i = g^{r_i}$. 
Since the joint distribution of $(r_i, z_i)$ is identical in $\game_3$ and $\game_4$, we conclude $\Pr[\Event_3] = \Pr[\Event_4] $.

Summarizing thus far, we upper bound the advantage of $\A$ winning the $\eucma$ game as follows: 
{\small
\begin{align}
	\epsilon =	\Pr[\Event_1] &\le \sum_{i = 1}^{3} \abs{\Pr[\Event_i] - \Pr[\Event_{i + 1}]} + \Pr[\Event_4] \nonumber \\	
	& \le \Pr[\Event_4] + \frac{Q \cdot (Q + Q_\hash)}{p}.  \label{eq:eucma_dl}
\end{align}
}
Therefore, in order to conclude the proof, it suffices to show that $\epsilon_4 := \Pr[\Event_4]$ is negligible. 
Below, we show that an adversary $\A$ against $\game_4$ can be used to construct an adversary $\B$ against the $\DLsketch$ assumption. This is a direct consequence of the forking lemma \cite{JC:PoiSte00,CCS:BelNev06}. 
The description of $\B$ follows:

\smallskip
\noindent
-~$\B( \GGG, \LinSpp, h, \sketch ):$ Given a $\DLsketch$ instance, $\B$ simulates the $\game_4$-challenger to $\A$ by appropriately programming the random oracle. Note that $\B$ can answer all queries via \cref{eq:sim_keygen_dl,eq:sim_sign_dl}. 
If $\A$ outputs a valid forgery $(\msg^*, \FSigsig^* = (\beta^*, z^*, \tsketch^*))$, $\B$ then checks if it ever replied back to $\A$ with $\beta^*$ to a random oracle query of the form $( \tildeh, R, \msg^* )$. If not, $\B$ aborts. 
Otherwise, assume $\A$ queried $( \tildeh, R, \msg^* )$ to the random oracle as its $I^*$-th query, where $I^* \in [Q]$. 
$\B$ then reruns $\A$ on the same randomness tape and answers the random oracle queries identically to the previous run up until the $I^*$-th query and with fresh random outputs from the $I^*$-th query. 
Then, if $\A$ outputs another valid forgery $(\msg'^*, \FSigsig'^* = (\beta'^*, z'^*, \tsketch'^*) )$, $\B$ then checks if $\msg'^* = \msg^*$, $\beta'^* \neq \beta^*$, and that $\beta'^*$ is the output of the $I^*$-th random oracle query. If not $\B$ aborts. 
Otherwise, $\B$ outputs $\frac{z^* - z'^*}{\beta^* - \beta'^*} - \Delta a$ as the solution to the $\DLsketch$ problem and terminates, where $\Delta a \gets \DiffRec(\LinSpp, \sketch, \tsketch^*)$

Let us analyze algorithm $\B$. 
It is easy to see that the first run simulates the view of $\game_4$ perfectly to $\A$. 
Therefore, standard argument using the forking lemma \cite{JC:PoiSte00,CCS:BelNev06} tells us that $\B$ outputs something (i.e., will not abort) with probability $\epsilon_4 \cdot( \frac{\epsilon_4}{Q} - \frac{1}{p} )$ and runs in time about twice as $\A$. 
Next, we show that condition on $\B$ outputting something, it solves the $\DLsketch$ problem with probability 1. 
Observe that since the two runs are identical up till the point $\A$ makes the $I^*$-th random oracle query, we must have that $\A$ queried $( \tildeh, R, \msg^* )$ to the random oracle as its $I^*$-th query in the second run as well. 
Due to validity of the forgery in the two runs, we have $\tildeh = h \cdot g^{\Delta a}$, $R = g^{z^*} \cdot \tildeh^{-\beta^*}$, and  $R = g^{z'^*} \cdot \tildeh^{-\beta'^*}$.
Simple calculation shows that $\mathrm{dlog}_g(h) = \frac{z^* - z'^*}{\beta^* - \beta'^*} - \Delta a$. 
Therefore, $\B$ correctly solves the $\DLsketch$ problem.

The above shows
\begin{align*}
	\epsilon_4 = \Pr[\Event_4] \le \sqrt{ Q \cdot \left(\epsilon_{\DLsketch}(\secpar) + \frac{1}{p} \right)},
\end{align*}
where $\epsilon_{\DLsketch}(\secpar)$ is the  maximum advantage of a PPT adversary against the $\DLsketch$ problem. 
Combining this with \cref{eq:eucma_dl} completes the proof of \cref{thm:dl_eucma}.
\end{proof}


\section{Further Details on Triangular Lattices} \label{app_sec:cvp_triangular_lattice}

We introduce the formal definition of triangular lattices here:
Let $d$ be any positive real and let $\triB = [\vb_1,\dots,\vb_n]$ be a basis matrix such that
\begin{enumerate}
\item $\| \vb_i\|_2 = d$ for all $i \in [n]$, and
\item $\vb_i \cdot \vb_j = d^2/2$ for all $i,j \in [n]$ with $i \ne j$, where \lq\lq $\cdot$'' denotes the inner product.
\end{enumerate}
We call $\triLattice = \Lattice(\triB)$ the triangular lattice (with basis length $d$). 
Notice that changing the $d$ has the effect of changing the size of the acceptance region $\AR$. That is, using a larger $d$ results in a larger $\AR$. 

A triangular lattice enjoys the property that
for any $\vx \in \RR^n$, we can calculate $\CV_{\triLattice}(\vx)$ efficiently in terms of the dimension $n$.
Concretely, its computational cost is $O(n^2)$.
Since the following description of the closest vector algorithm is invariant to the choice of $d$, we assume $d = 1$. 
The details follow.

For simplicity, we use the representation with respect to $\triLattice$ for the target vector $\vx$.
If an input vector is with respect to the standard basis, then it can be converted to one with the representation with respect to $\triLattice$ by multiplying with $\mB^{-1}$.
Let $\vx = x_1\vb_1 + \dots + x_n \vb_n$ (where $x_i \in \RR$ for each $i \in [n]$) be the target vector for which we would like to compute the closest vector $\vy = \CV_{\triLattice}(\vx)$.
For simplicity, we first explain the case $x_i \in [0,1)$ for each $i \in [n]$, and later explain how to extend it to the general case ($x_i \in \RR$).
In this case, $\vy = \CV_{\triLattice}(\vx)$ can be written as $\vy = y_1 \vb_1 + \dots + y_n \vb_n$ with $y_i \in \{0,1\}$ for each $i \in [n]$.

Due to the property of the triangular lattice, for all $i,j \in [n]$, we have the following properties:
\begin{align*}
x_i \le x_j &\Rightarrow (y_i, y_j) \in \{(0,0), (0,1), (1,1)\},\\
x_i \ge x_j &\Rightarrow (y_i, y_j) \in \{(0,0), (1,0), (1,1)\}.
\end{align*}
In other words, the magnitude relation among the coordinates $\{x_i\}_{i \in [n]}$ of the target vector and that among the coordinates $\{y_i\}_{i \in [n]}$ of the closest vector are synchronized.
Hence, the above two relations can be equivalently written as
\begin{align*}
x_i \le x_j &\Rightarrow y_i \le y_j,\\
x_i \ge x_j &\Rightarrow y_i \ge y_j.
\end{align*}

Using this fact, we consider the sorting of $\{(x_i, y_i)\}_{i \in [n]}$ in ascending order by using $\{x_i\}_{i \in [n]}$ as the sorting key.
Let $(x^*_i, y^*_i)_{i \in [n]}$ be the result of the sorting.
Then, due to the above relations, we have
\begin{equation}
y^*_1 \le y^*_2 \le \dots \le y^*_n. \label{eqn:tri}
\end{equation}
Since we are considering the case that $y_i \in \{0,1\}$ for each $i \in [n]$, the sequence $(y^*_i)_{i \in [n]}$ has the property that there exists an index $k \in \{0, 1,\dots,n\}$ such that $y^*_1 = \dots = y^*_k = 0$ and $y^*_{k+1} = \dots = y^*_n = 1$.
There are $n+1$ candidates for $k$.
Hence, by computing the distance between the target vector $\vx$ and $n + 1$ vectors satisfying \cref{eqn:tri}, we can compute the closest vector $\vy$.

The concrete procedure for computing the closest vector $\vy = \CV_{\triLattice}(\vx) = \sum_{i \in [n]} y_i \vb_i$ (with $y_i \in \{0,1\}$) from a target vector $\vx = \sum_{i \in [n]} x_i \vb_i$ (with $x_i \in [0,1)$) is as follows:
\begin{enumerate}
\item Sort $\{x_i\}_{i \in [n]}$ in ascending order.
Let $(x'_1,\dots,x'_n)$ be the result of the sorting, and let $\sigma: [n] \to [n]$ be the permutation representing this sorting. 
Namely, we have $x'_i = x_{\sigma(i)} \Leftrightarrow x_i = x'_{\sigma^{-1}(i)}$ for each $i \in [n]$.
\item For each $k \in \{0,1,\dots,n\}$, let $\vy_k = y_{k,1} \vb_1 + \dots + y_{k,n} \vb_n)$ be the vector satisfying $y_{k,\sigma(j)} = 0$ for $j \le k$ and $y_{j, \sigma(j)} = 1$ for $j > k$.
Note that $\{\vy_0, \vy_1, \dots, \vy_k\}$ is the set of candidates of the closest vector $\CV_{\triLattice}(\vx)$.
\item Compute $\|\vx - \vy_k\|_2$ for each $k \in \{0,1,\dots,n\}$, and find the index $k^* \in \{0,1,\dots,n\}$ of the smallest vector such that $\vy_{k^*} = \min_k \|\vx - \vy_k\|_2$.
\item Output $\vy_{\min} := \vy_{k^*}$ as the closest vector $\CV_{\triLattice}(\vx)$ of $\vx$.
\end{enumerate}

The computational cost of the above procedure in terms of $n$ can be estimated as follows:
The sorting in Step 1 costs $O(n \log n)$.
The calculation of $\|\vx - \vy_k\|_2$ in Step 3 for each $k \in \{0,1,\dots,n\}$ costs $O(n)$.
Since we calculate the distance $n+1$ times, Step 3 costs in total $O(n^2)$.
Hence, in total we can calculate $\CV_{\triLattice}(\vx)$ with computational cost $O(n^2)$.

The above algorithm can be extended to cover the general case where $\vx = \sum_{i \in [n]} x_i \vb_i$ with $x_i \in \RR$ for each $i \in [n]$.
Specifically,
before executing the above algorithm, we decompose each $x_i$ as $x_i = z_i + x'_i$ where $z_i \in \ZZ$ and $x'_i \in [0,1)$.
We then apply the above algorithm to $\{x'_i\}_{i \in [n]}$.
Let $\vy_{\min}$ be the result.
Then, the closest vector $\CV_{\triLattice}(\vx)$ of $\vx$ is $\vy_{\min} + \sum_{i \in [n]} z_i \vb_i$.
It is easy to see the correctness of this algorithm, and that the asymptotic computational cost in terms of $n$ remains the same.


\section{Composing Multiple Fuzzy Key Settings} \label{app_sec:composing_fks}
Our formalization of a lattice-based fuzzy key setting and the linear sketch scheme in \cref{sec:linear_sketch} can easily be adapted to handle a \lq\lq composed'' fuzzy key setting and associated linear sketch scheme.

Specifically, suppose we have $m$ kinds of fuzzy data, and for $i \in [m]$, let $\FuzKeySet_i = (\FuzDataSpace_i, \FuzDataDist_i, \AR_i, \ErrDist_i, \ErrFNMR_i)$ be a lattice-based fuzzy key setting for the $i$-th fuzzy data, where the $i$-th fuzzy data space $\FuzDataSpace_i = \RR^{n_i}$ is associated with a lattice with the basis matrix $\mB_i \in \RR^{n_i \times n_i}$.
For simplicity, assume that the parameter $p$ is common for all of $\{\FuzKeySet_i\}_{i \in [m]}$.
Then, we can consider the composed fuzzy key setting $\FuzKeySet^* = (\FuzDataSpace^*, \FuzDataDist^*, \AR^*, \ErrDist^*, \ErrFNMR^*)$ that is a natural combination of the fuzzy key settings $\{\FuzKeySet_i\}_{i \in [m]}$:
The fuzzy data space $\FuzDataSpace^*$ is the direct product $\prod_{i \in [m]} \FuzDataSpace_i$ for which the lattice $\mB^*$ of the following form is associated:
\[
\mB^* = \left[
\begin{array}{cccc}
\mB_1 &       &       &      \\
      & \mB_2 &       &      \\
      &       & \ddots &      \\
      &       &       & \mB_m\\
\end{array}
\right];
\]
The fuzzy data distribution $\FuzDataDist^*$ is the joint distribution $(\FuzDataDist_1,\dots,\FuzDataDist_m)$, and the same for the error distribution $\ErrDist^*$;
The acceptance region function $\AR^*$ has the property that for $\vx^* = (\vx_1,\dots,\vx_m), \vx'^{*} = (\vx_1', \dots, \vx'_m) \in \FuzDataSpace^*$, we have $\vx'^* \in \AR^*(\vx^*)$ if and only if $\vx'_i \in \AR_i(\vx_i)$ for all $i \in [m]$;
The error parameter $\ErrFNMR^*$ can be upper-bounded by $\sum_{i \in [m]} \ErrFNMR_i$ by the union bound.

Furthermore, the algorithms of the linear sketch scheme for the composed fuzzy key setting $\FuzKeySet^*$ can be computed by computing those for the linear sketch scheme for each fuzzy key setting $\FuzKeySet_i$ with which $\mB_i$ is associated, and concatenate the results, except for the proxy key $a$ in $\Sketch$ and the difference $\Delta a$ in $\DiffRec$.
For $a$ in $\Sketch$ and $\Delta a$ in $\DiffRec$, we need an application of the universal hash function $\linhash$ for the combined linear sketch scheme whose domain is $(\ZZ_p)^{\sum_{i \in [n_i]}n_i}$ and which takes the concatenated results as input.


\section{Further Experimental Results} \label{app:sec:experimental_result}
This section provides details on experimental results that were omitted in the main body.

\subsection{DET curve for dimension $n = 400$}\label{app:sec:DET}
\cref{fig:DET_400} plots the detection error tradeoff (DET) curve for fuzzy biometrics with dimension $n = 400$. 
Any region that lies above the DET curve is realizable.
It can be checked that for a small (resp. large) size of $d$ (hence the acceptance region $\AR$), we achieve better values for $\tFNMR$ (resp. \tCFMR) as expected.

\begin{figure}[htbp]
  \begin{center}
	\includegraphics[ width=0.3\textwidth ]{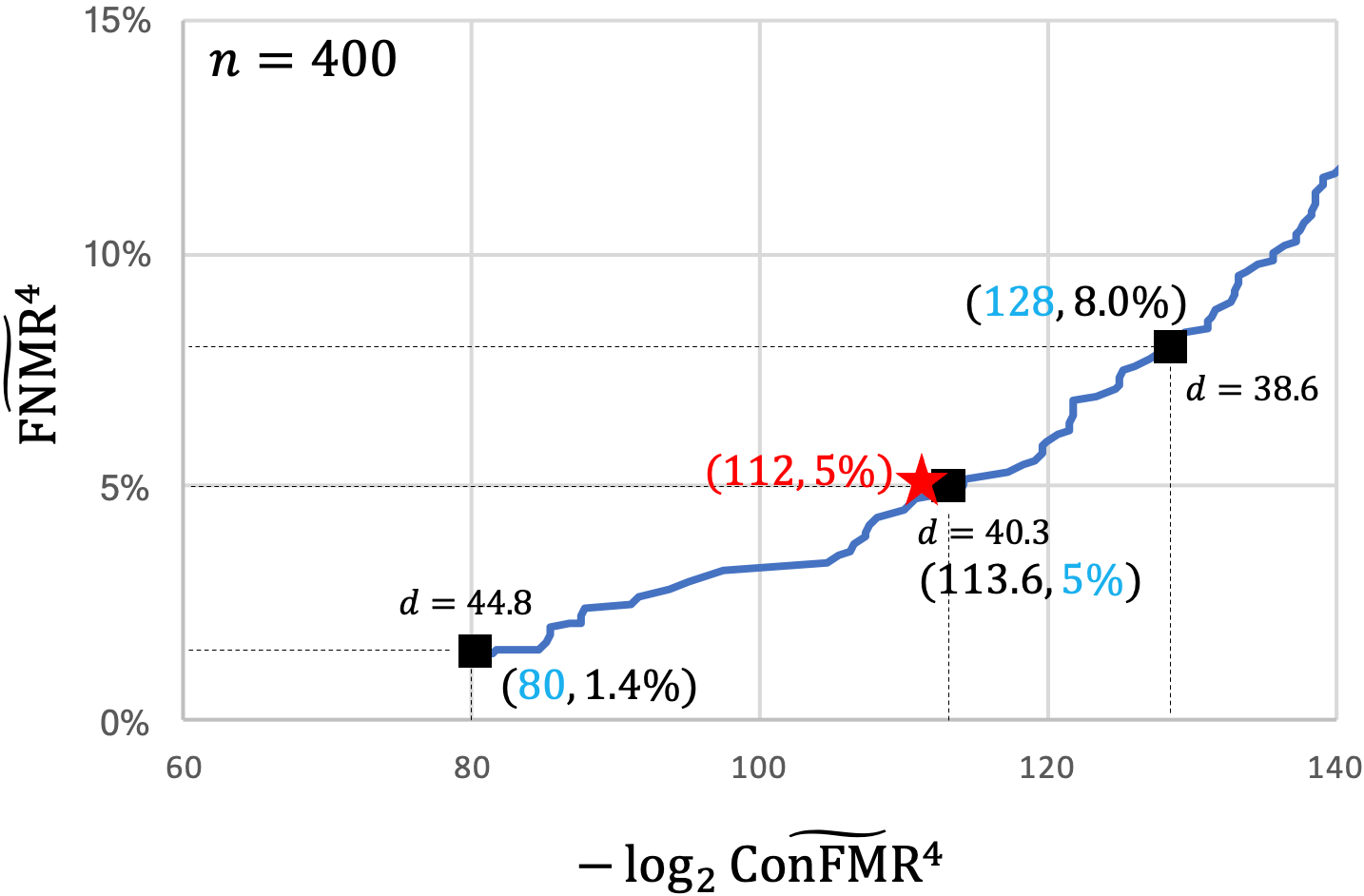}
  \end{center}
 \caption{\small DET Curve for $n = 400$.  }
 \label{fig:DET_400}
\end{figure}

\subsection{Validity of EVA result} \label{app:sec:validity_EVA}
We provide the omitted graphical evidence on the validity of our estimation of $\tFMR$ and $\tCFMR_t$ for $t \in [M]$. 
Recall that we set $k^* = 0.1\% \times k_{\max}$, where $k_{\max} = 3,925,584$ to estimate $\tFMR$ using EVA. 
We also set $M = 100$ to define the $M$-variant of $\set{\CFMR_t}_{t \in [M]}$ and performed EVA to estimate each $\tCFMR_t$ by setting $k^* = 0.5 \% \times k_{\max}$. 
\cref{fig:validitiy_EVA_appendix} illustrates the validity of our estimation for $\tFMR$ when the dimension $n \in \set{200, 300, 400}$ and $\tCFMR_t$ for $t \in \set{25, 50, 100}$.

\begin{figure*}[htbp]
 \begin{minipage}{0.245\hsize}
  \begin{center}
	\includegraphics[ width=1.05\textwidth ]{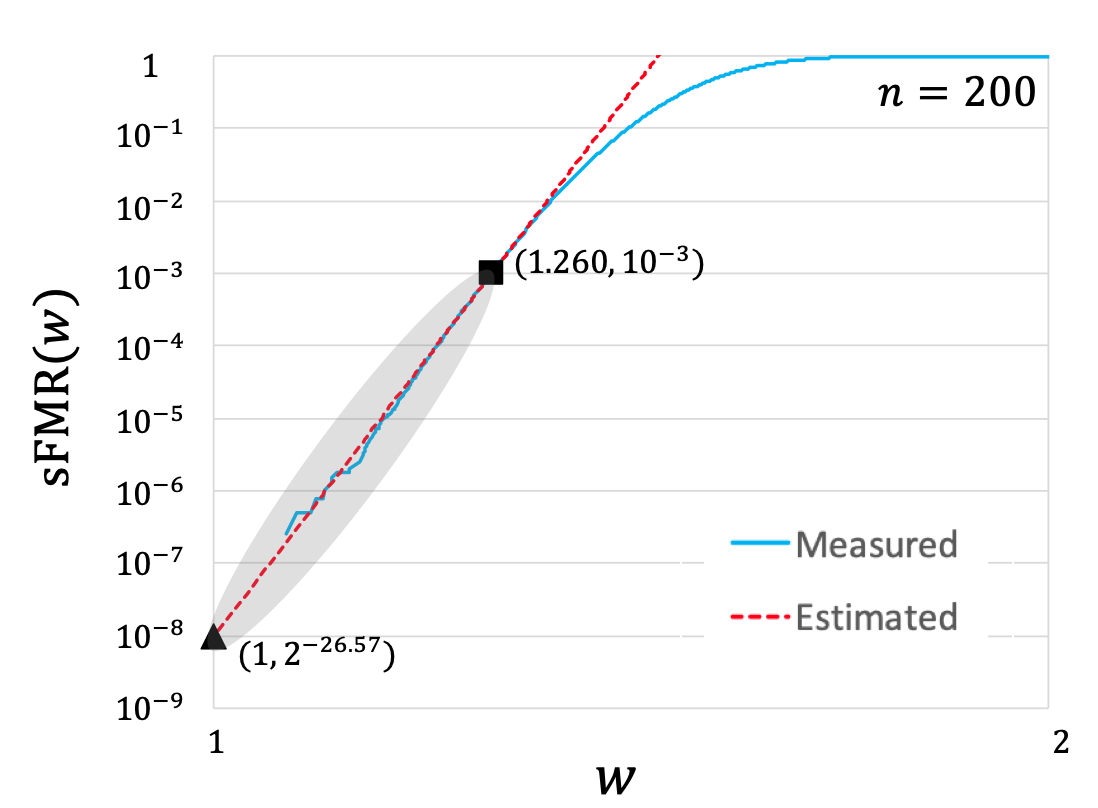}
  \end{center}
  \subcaption{$\sFMR(w)$ for $n = 200$}
 \end{minipage}
 \begin{minipage}{0.245\hsize}
  \begin{center}
	\includegraphics[ width=1.05\textwidth ]{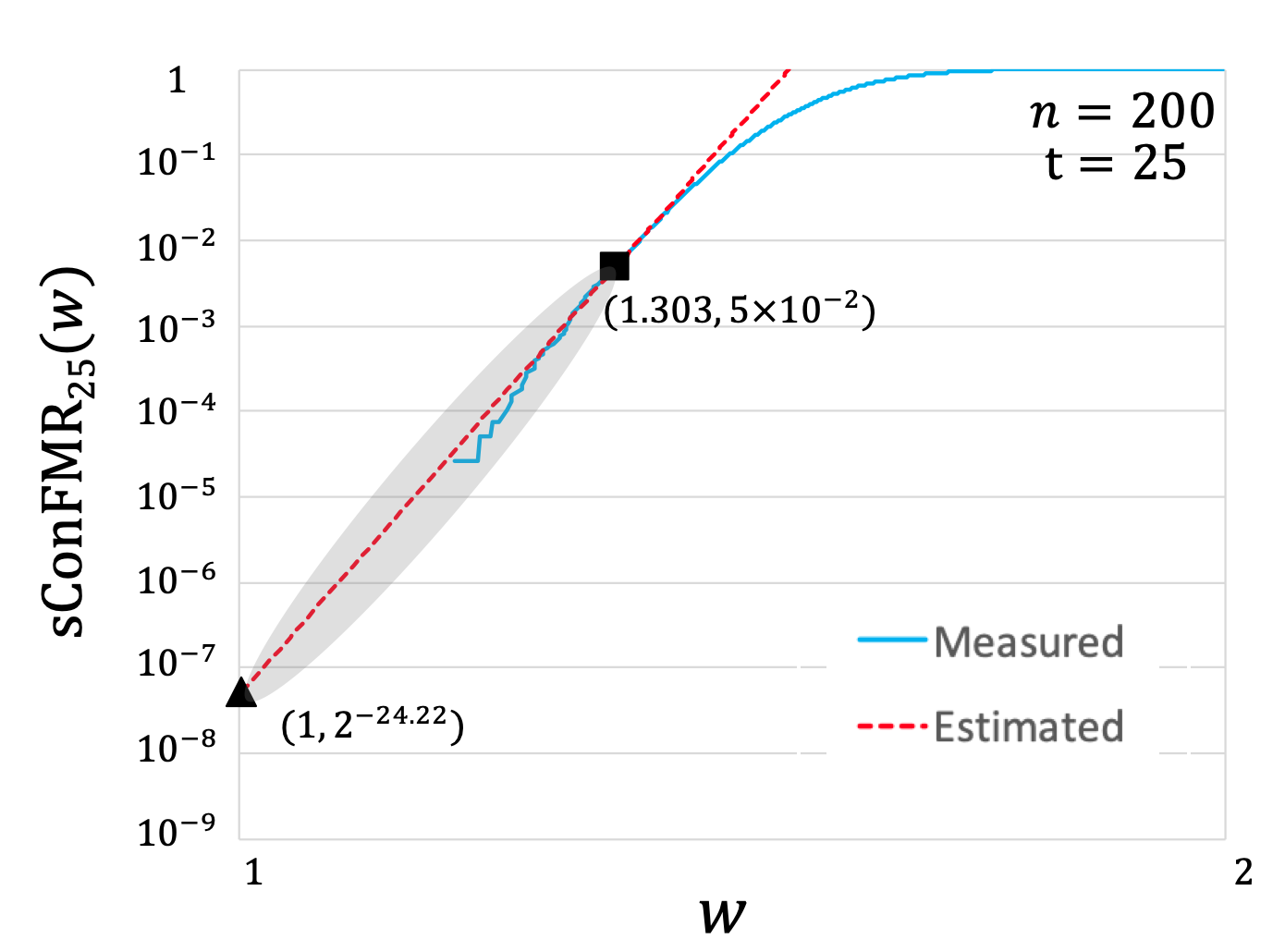}
  \end{center}
  \subcaption{$\sCFMR_{25}(w)$ for $n = 200$}
 \end{minipage}
  \begin{minipage}{0.245\hsize}
   \begin{center}
	\includegraphics[ width=1.05\textwidth ]{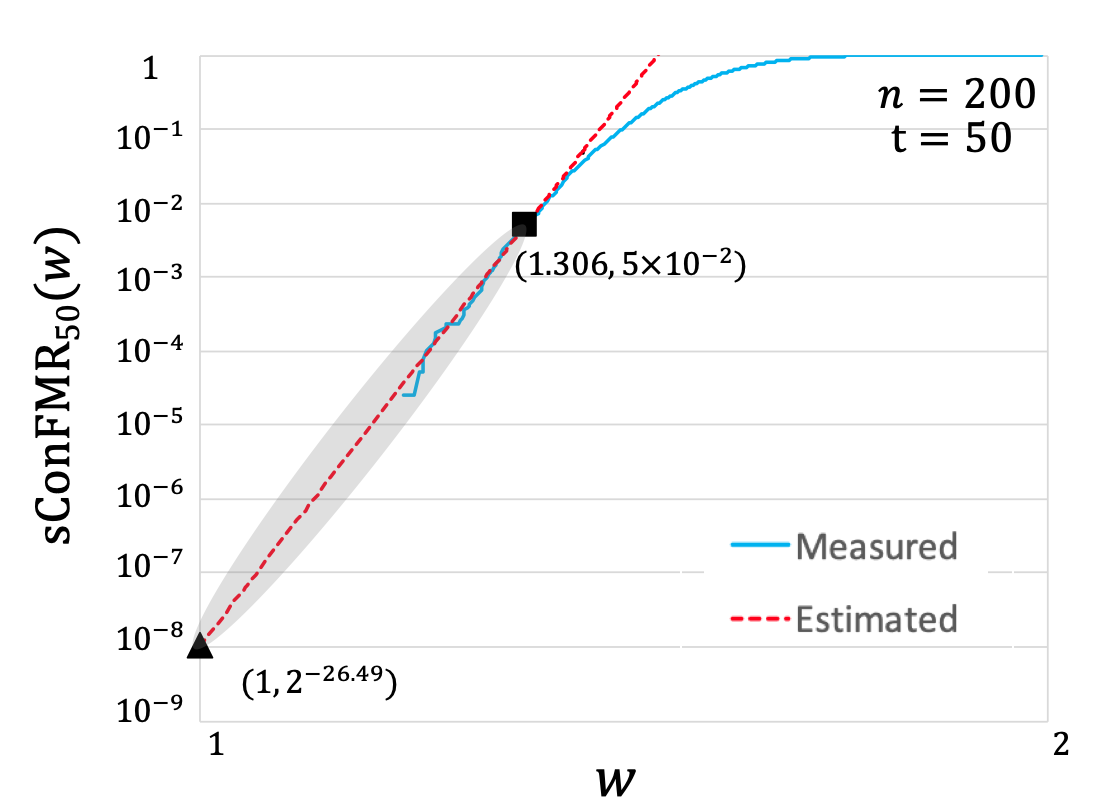}
  \end{center}
  \subcaption{$\sCFMR_{50}$ for $n = 200$}
 \end{minipage}
   \begin{minipage}{0.245\hsize}
   \begin{center}
	\includegraphics[ width=1.05\textwidth ]{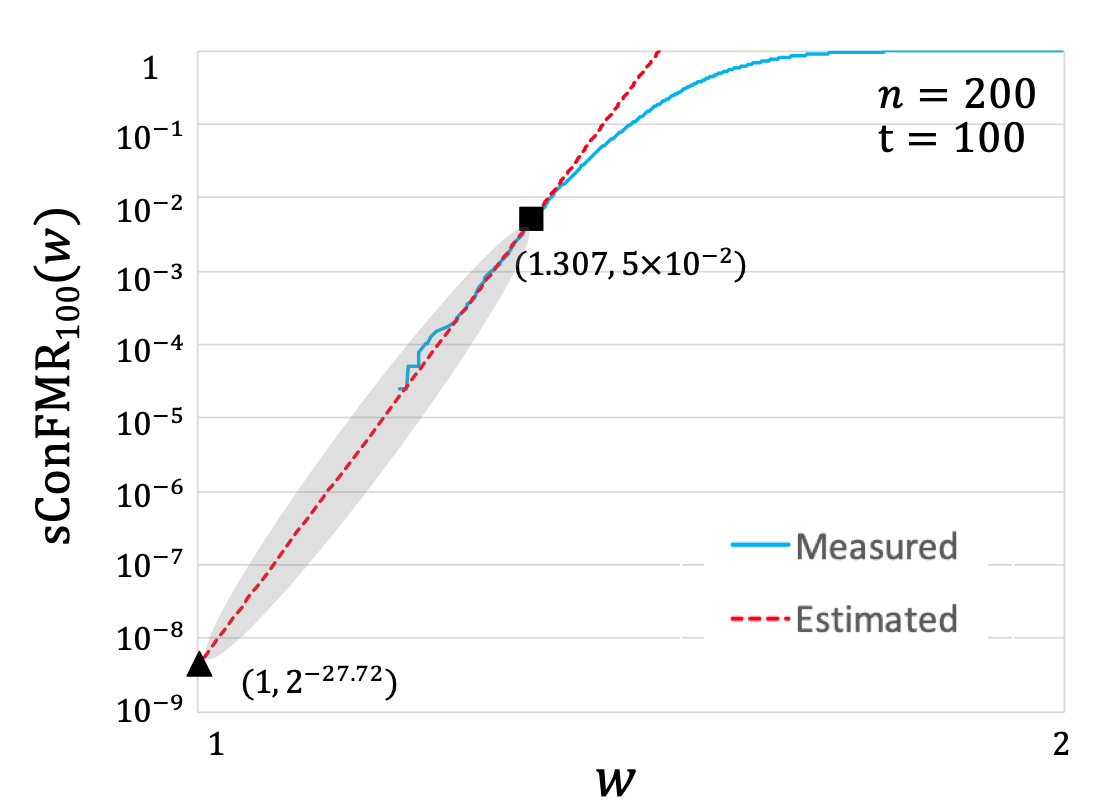}
  \end{center}
  \subcaption{$\sCFMR_{100}$ for $n = 200$}
 \end{minipage}

 \begin{minipage}{0.245\hsize}
  \begin{center}
	\includegraphics[ width=1.05\textwidth ]{Input_Figure/fmr_300.png}
  \end{center}
  \subcaption{$\sFMR(w)$ for $n = 300$}
 \end{minipage}
 \begin{minipage}{0.245\hsize}
  \begin{center}
	\includegraphics[ width=1.05\textwidth ]{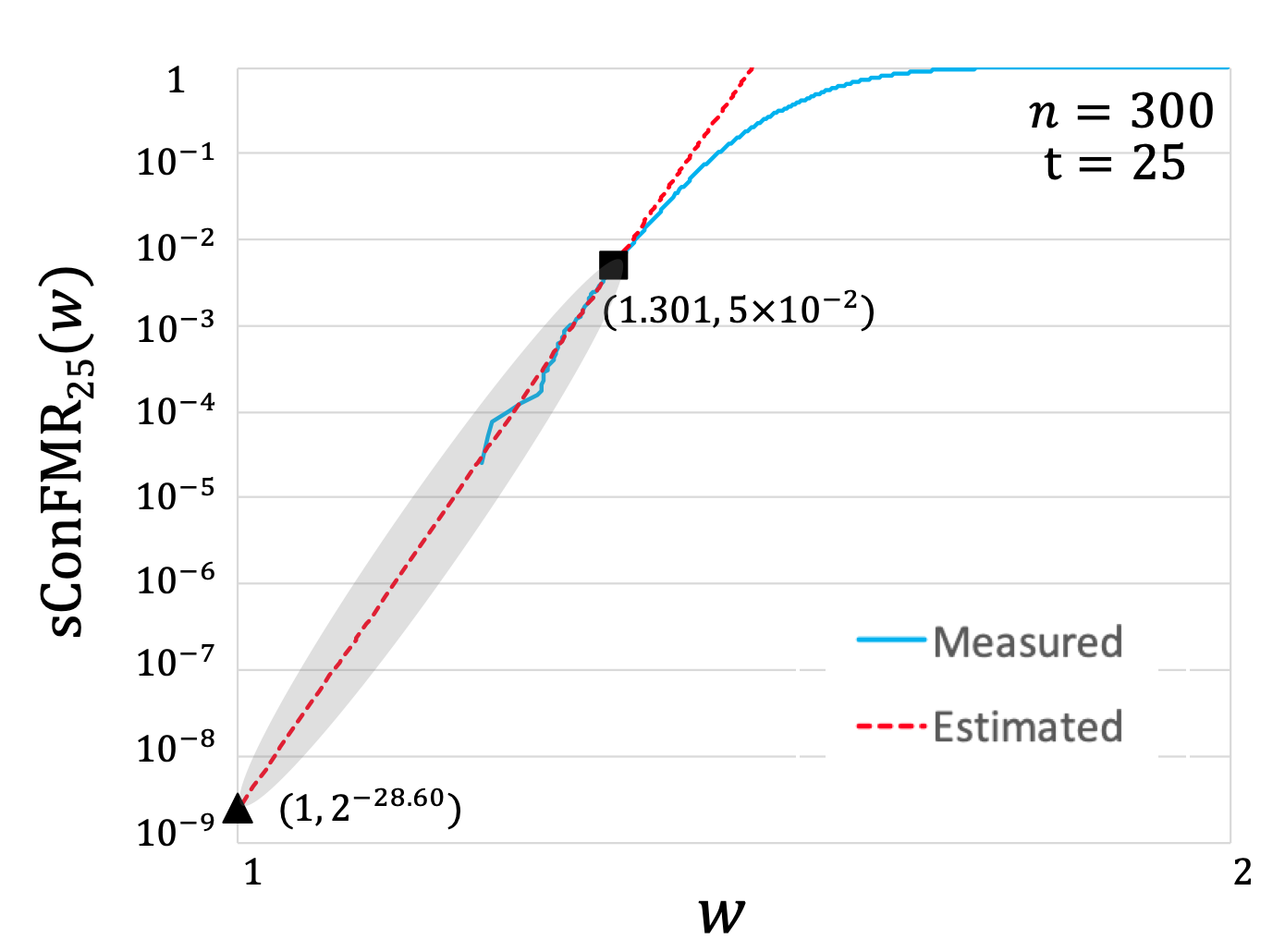}
  \end{center}
  \subcaption{$\sCFMR_{25}(w)$ for $n = 300$}
 \end{minipage}
  \begin{minipage}{0.245\hsize}
   \begin{center}
	\includegraphics[ width=1.05\textwidth ]{Input_Figure/cfmr_300_50.png}
  \end{center}
  \subcaption{$\sCFMR_{50}$ for $n = 300$}
 \end{minipage}
   \begin{minipage}{0.245\hsize}
   \begin{center}
	\includegraphics[ width=1.05\textwidth ]{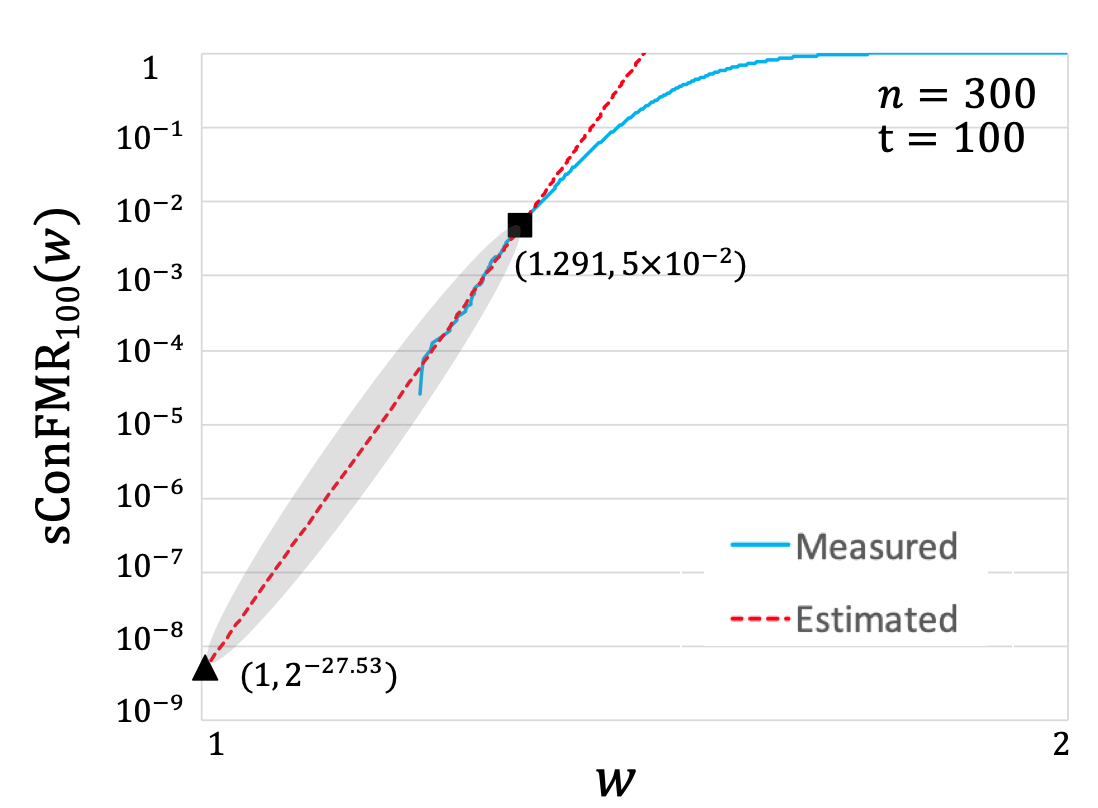}
  \end{center}
  \subcaption{$\sCFMR_{100}$ for $n = 300$}
 \end{minipage}

 \begin{minipage}{0.245\hsize}
  \begin{center}
	\includegraphics[ width=1.05\textwidth ]{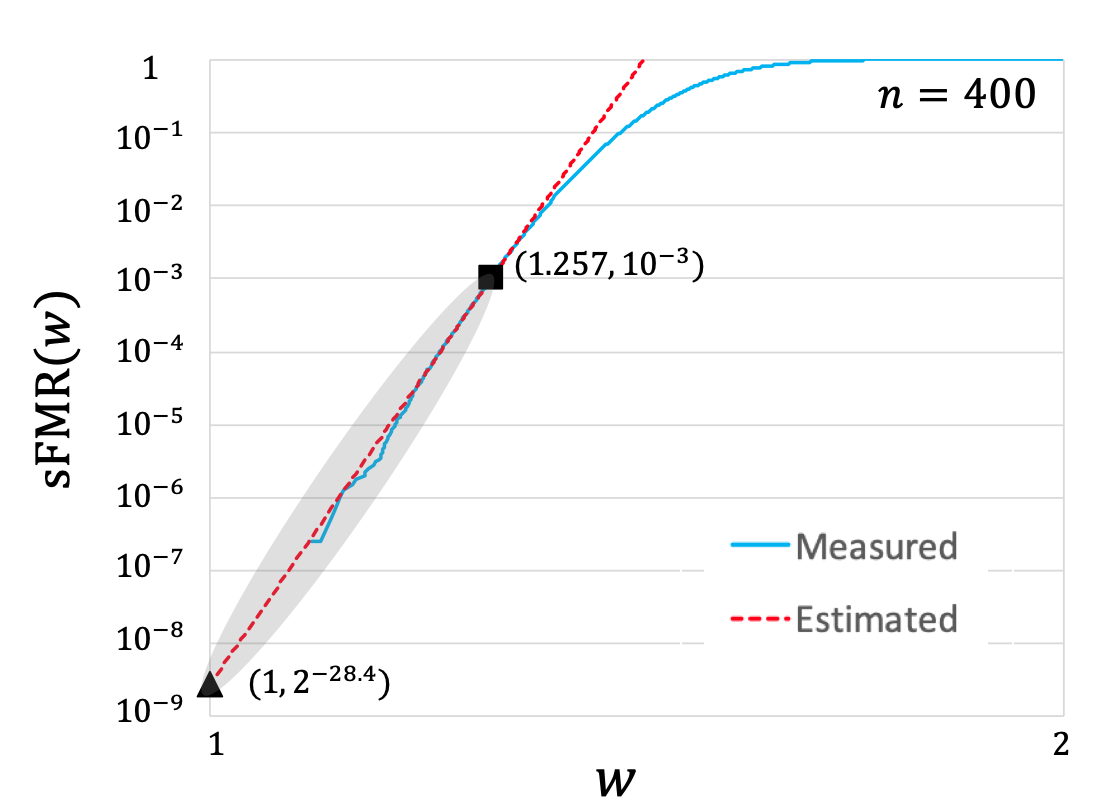}
  \end{center}
  \subcaption{$\sFMR(w)$ for $n = 400$}
 \end{minipage}
 \begin{minipage}{0.245\hsize}
  \begin{center}
	\includegraphics[ width=1.05\textwidth ]{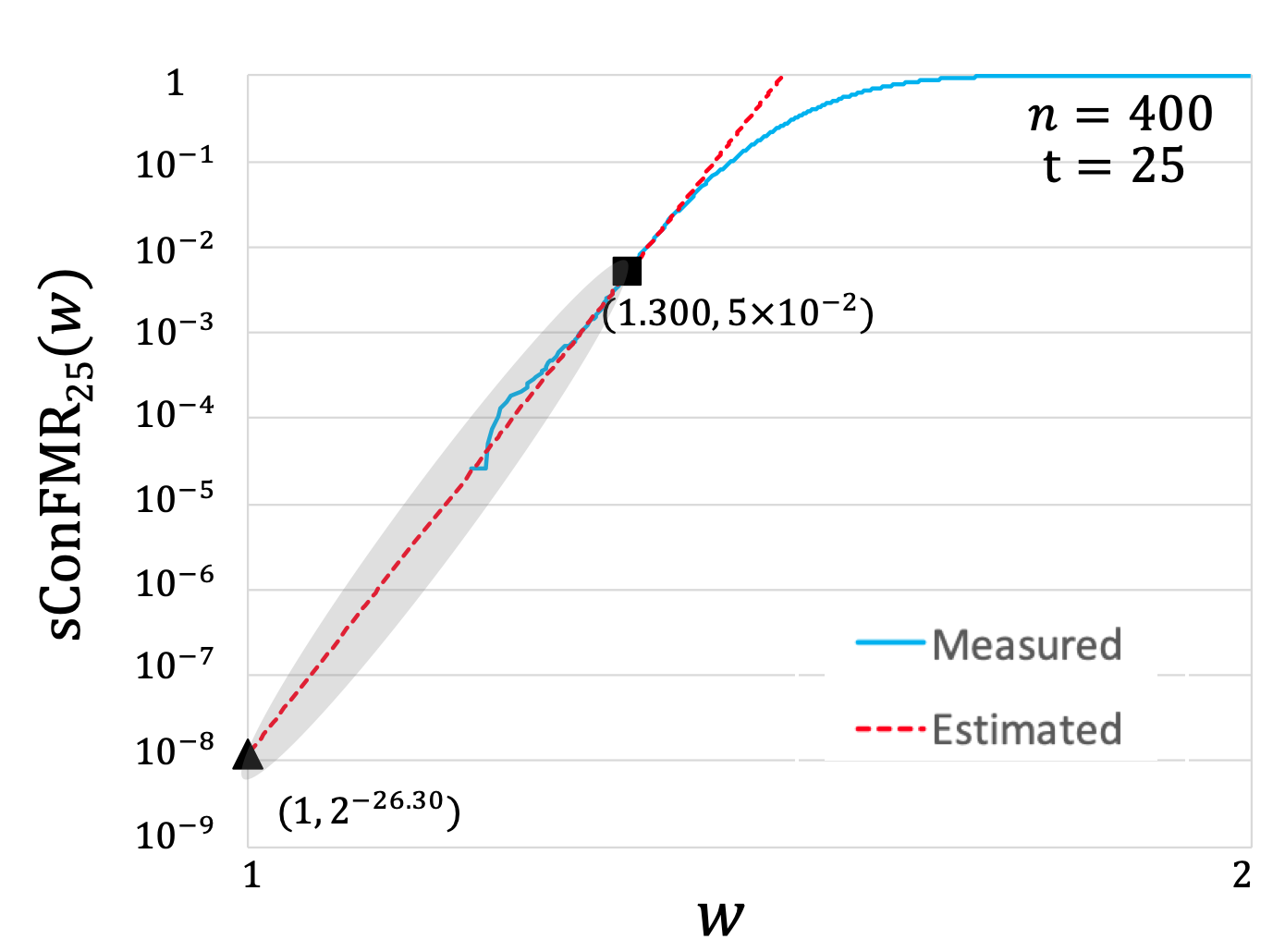}
  \end{center}
  \subcaption{$\sCFMR_{25}(w)$ for $n = 400$}
 \end{minipage}
  \begin{minipage}{0.245\hsize}
   \begin{center}
	\includegraphics[ width=1.05\textwidth ]{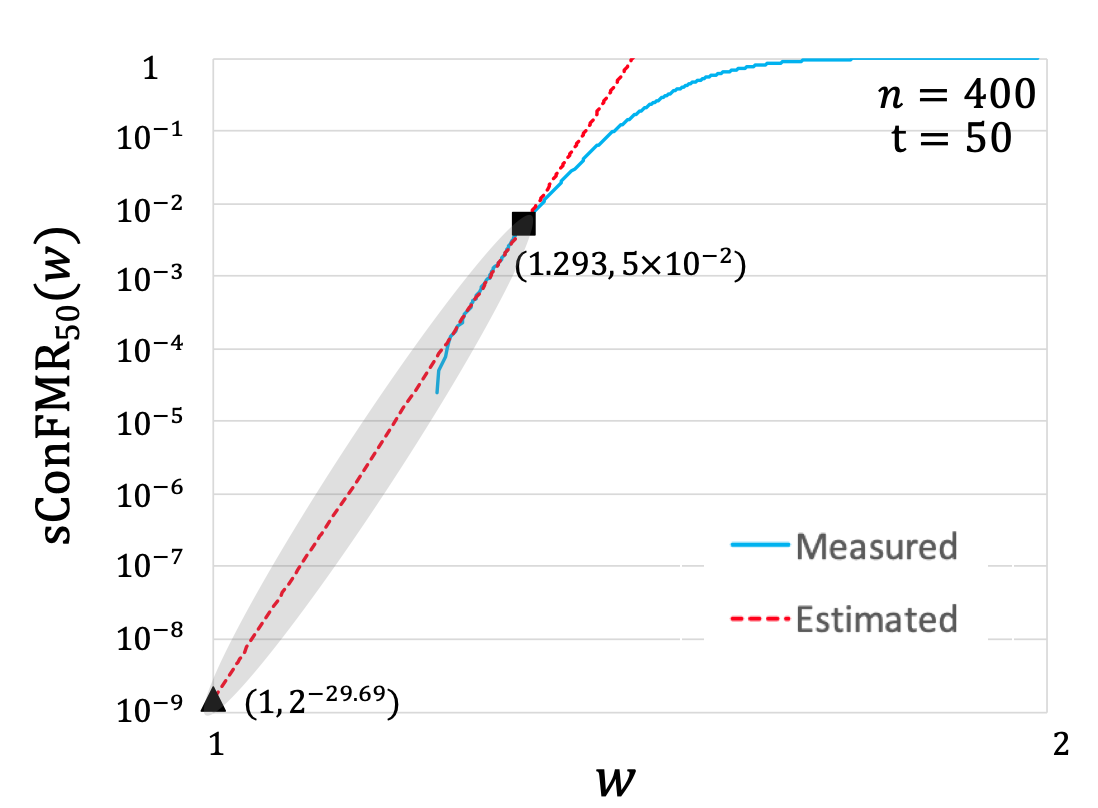}
  \end{center}
  \subcaption{$\sCFMR_{50}$ for $n = 400$}
 \end{minipage}
   \begin{minipage}{0.245\hsize}
   \begin{center}
	\includegraphics[ width=1.05\textwidth ]{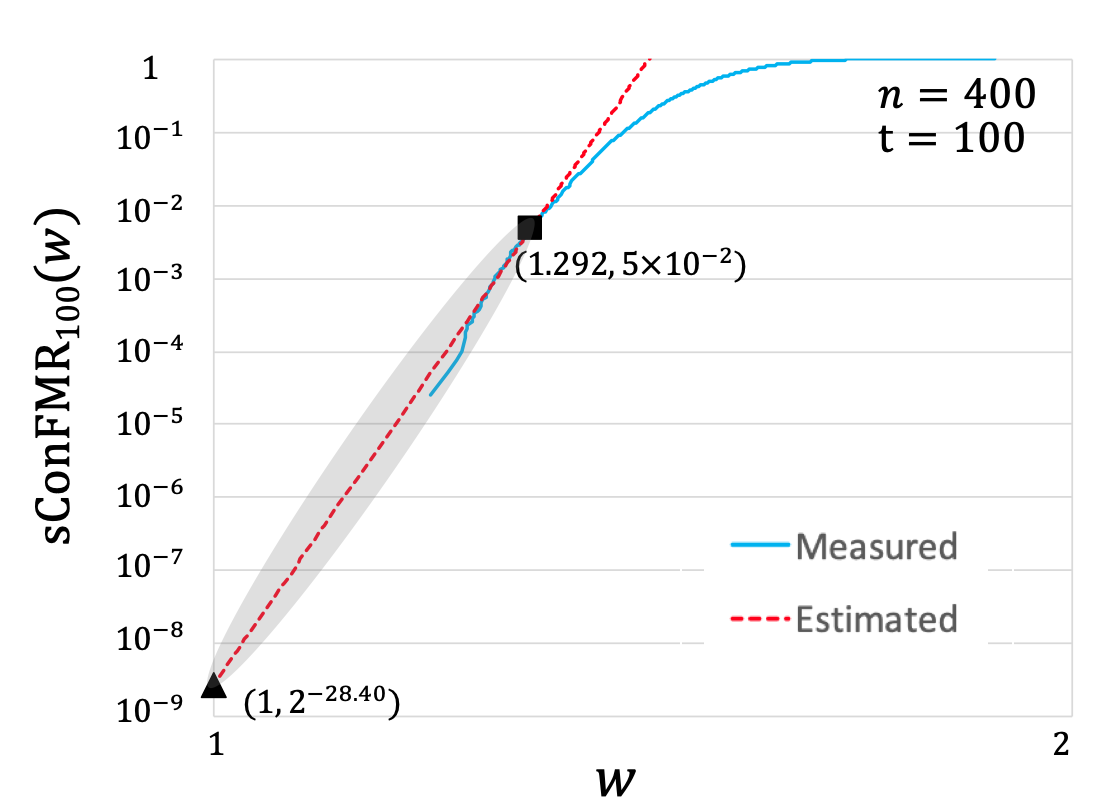}
  \end{center}
  \subcaption{$\sCFMR_{100}$ for $n = 400$}
 \end{minipage}
 
 \caption{\small The blue line indicates the measured values of $\sFMR(w)$ and $\sCFMR_{50}(w)$ w.r.t to our dataset $S$. The red line indicates our estimation of the probability distribution of  $\sFMR(w)$ and $\sCFMR_{t}(w)$ for $t \in \set{25, 50, 100}$ via EVA. The gray region is the region for which EVA provides a reliable estimation. The square plots $(w^*, k^*)$ and the triangle plots $(1, X)$, where $X$ is the estimation for $\tFMR$ and $\tCFMR_{t}$ for $t \in \set{25, 50, 100}$.  }
 \label{fig:validitiy_EVA_appendix}
\end{figure*}

\newpage
\setcounter{tocdepth}{2}
\tableofcontents

\end{document}